%% file: main.tex
\documentclass{article}

\PassOptionsToPackage{numbers, compress}{natbib}

\usepackage[final]{neurips_2025}

\usepackage[utf8]{inputenc} 
\usepackage[T1]{fontenc}    
\usepackage{hyperref}       
\usepackage{url}            
\usepackage{booktabs}       
\usepackage{amsfonts}       
\usepackage{nicefrac}       
\usepackage{microtype}      

\usepackage{multirow,tabularx}
\usepackage{amsmath} 
\usepackage{enumitem}
\usepackage{algorithm,algorithmicx,algpseudocode}

\usepackage{bm}
\usepackage[table]{xcolor}
\usepackage{graphicx}
\usepackage[most]{tcolorbox} 
\newcommand{\cmark}{{\color{green!60!black}\checkmark}}
\newcommand{\xmark}{{\color{red}\ding{55}}}
\definecolor{lightgray1}{gray}{0.95}  
\definecolor{lightgray2}{gray}{0.85}  

\usepackage{amsmath, amssymb, amsthm}
\usepackage{makecell}                 
\usepackage{booktabs}

\DeclareMathOperator{\supp}{supp}
\theoremstyle{plain}
\newtheorem{theorem}{Theorem}
\newtheorem{proposition}{Proposition}

\usepackage{xcolor}
\usepackage{pifont}
\usepackage{subcaption}  

\newcommand{\ie}{\emph{i.e., }}
\newcommand{\eg}{\emph{e.g., }}

\newcommand{\za}[1]{{\color{black}{#1}}}
\newcommand{\hugq}[1]{{\color{black}{#1}}}

\title{Fading to Grow: Growing Preference Ratios via Preference Fading Discrete Diffusion for Recommendation}

\newcommand{\USTC}{University of Science and Technology of China}

\newcommand{\IR}{Independent Researcher}

\author{
  Guoqing Hu\textsuperscript{\ddag}~~
  An Zhang\textsuperscript{\ddag}\thanks{Corresponding authors.}~~
  Shuchang Liu\textsuperscript{\S}~~
  Wenyu Mao\textsuperscript{\ddag}~~
  Jiancan Wu\textsuperscript{\ddag}\AND
  Xun Yang\textsuperscript{\ddag}~~
  Xiang Li\textsuperscript{\S}~~
  Lantao Hu\textsuperscript{\S}~~
  Han Li\textsuperscript{\S}~~
  Kun Gai\textsuperscript{\S}~~
  Xiang Wang\textsuperscript{\ddag}\\[8pt]
  \textsuperscript{\ddag}\USTC \\
  \textsuperscript{\S}\IR \\
  \texttt{\{HugoChinn\}@mail.ustc.edu.cn} \\
  \texttt{\{an\_zhang,xyang21\}@ustc.edu.cn} \\
  \texttt{\{wenyumao2, wujcan, xiangwang1223\}@gmail.com}
}

\begin{document}

\maketitle

\begin{abstract}
Recommenders aim to rank items from a discrete item corpus in line with user interests, yet suffer from extremely sparse user preference data.
Recent advances in diffusion models have inspired diffusion-based recommenders, which alleviate sparsity by injecting noise during a forward process to prevent the collapse of perturbed preference distributions.
However, current diffusion‑based recommenders predominantly rely on continuous Gaussian noise, which is intrinsically mismatched with the discrete nature of user preference data in recommendation.
In this paper, building upon recent advances in discrete diffusion, we propose \textbf{PreferGrow}, a discrete diffusion-based recommender system that models preference ratios by fading and growing user preferences over the discrete item corpus.
PreferGrow differs from existing diffusion-based recommenders in three core aspects:
(1) Discrete modeling of preference ratios:
PreferGrow models relative preference ratios between item pairs, rather than operating in the item representation or raw score simplex.
This formulation aligns naturally with the discrete and ranking-oriented nature of recommendation tasks.
(2) Perturbing via preference fading:
Instead of injecting continuous noise, PreferGrow fades user preferences by replacing the preferred item with alternatives---physically akin to negative sampling---thereby eliminating the need for any prior noise assumption.
(3) Preference reconstruction via growing:
PreferGrow reconstructs user preferences by iteratively growing the preference signals from the estimated ratios.
PreferGrow offers a well-defined matrix-based formulation with theoretical guarantees on Markovianity and reversibility, and it demonstrates consistent performance gains over state-of-the-art diffusion-based recommenders across five benchmark datasets, highlighting both its theoretical soundness and empirical effectiveness.
Our codes are available at \url{https://github.com/Hugo-Chinn/PreferGrow}.
\end{abstract}

\input{chapters/1-intro_new}
\input{chapters/2-prelim}

\input{chapters/3-method}

\input{chapters/4-exp}
\input{chapters/6-else}

\bibliographystyle{unsrtnat}
\bibliography{ref}


\clearpage
\appendix
\input{chapters/5-relatedworks}
\input{chapters/7-appendix}
\input{chapters/8-exp_details}



\end{document}

%% file: chapters/1-intro_new.tex
\section{Introduction}\label{sec:intro}



\begin{figure}[ht!]
    \centering
    \includegraphics[width=\textwidth]{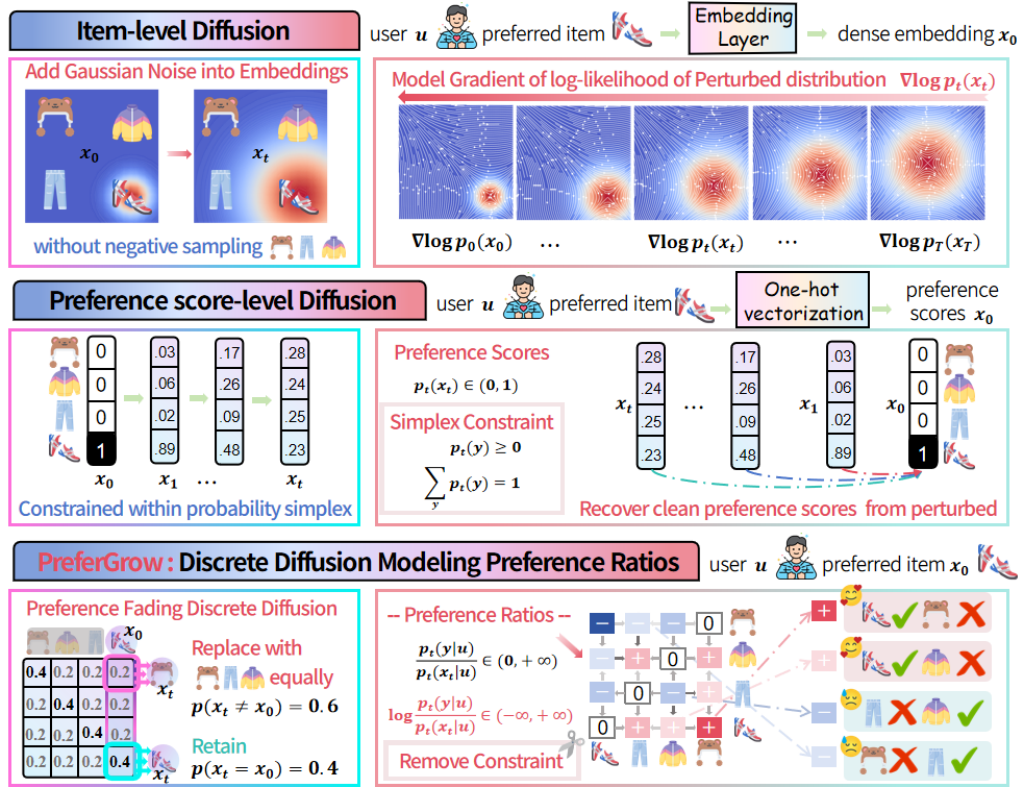}
    \caption{Comparison of diffusion-based recommenders in terms of modeling targets and perturbation strategies.
    Item-level diffusion recommenders (top) add Gaussian noise into dense item embeddings but overlook negative signals.
    Preference score-level diffusion recommenders (middle) perturb one-hot preference vectors under the constraints of probability simplex.
    PreferGrow (bottom) directly models discrete preference ratios via preference fading.
    }
    \vspace{-15pt}
    \label{fig:sparse}
\end{figure}

Recommender systems aim to rank items from a discrete item set that align with user interests, where ones in the user-item interaction matrix denote observed interactions, and zeros indicate unobserved or missing entries.
In real-world scenarios, this interaction matrix is often extremely sparse \cite{sparse1, sparse2}, which poses a significant challenge for recommenders in accurately modeling user preferences \cite{S3Rec, sparse4}.
Recent advances in diffusion models (DMs) offer a promising solution: by injecting noise during the forward process, DMs transform sparse data into smoother, denser distributions while preventing collapse into isotropic zero values \cite{NCSN, DDIM, DDPM}.
Motivated by these properties, diffusion-based recommenders \cite{DreamRec, DiffuRec, DiffRec1, CDDRec, DiffuASR, DDRM, DiffCDR, DiffGT, CFDiff, CaDiRec, DMSR, DiffCL, DMCDR, MoDiCF, DiQDiff, DRGO, DiffuRec2, iDreamRec, PreferDiff, DiffRec2, RecFusion, LD4MRec, PDRec, EdgeRec, DifFaiRec, D3Rec, DDSR, DiffMM, SDiff} have proliferated in recent studies, showing strong potential in addressing data sparsity and improving preference modeling.
 
Current diffusion-based recommenders predominantly adopt continuous noise as the perturbation mechanism by adding noise into either dense item embeddings or one-hot preference score vectors. 
As shown in Figure \ref{fig:sparse}, item-level diffusion recommenders \cite{DreamRec, DiffuRec, DiffRec1, CDDRec, DiffuASR, DDRM, DiffCDR, DiffGT, CFDiff, CaDiRec, DMSR, DiffCL, DMCDR, MoDiCF, DiQDiff, DRGO, DiffuRec2, iDreamRec, PreferDiff} apply Gaussian noise to user-interacted item embeddings during the forward process, and learn a score function---the gradient of the perturbed distribution's log-likelihood---to recover item embeddings during the reverse process.
However, these methods overlook the negative signals in recommendations, lacking mechanisms such as negative sampling to differentiate user preferences \cite{DreamRec, PreferDiff}.
In contrast, preference score-level diffusion recommenders \cite{DiffRec2, RecFusion, LD4MRec, PDRec, EdgeRec, DifFaiRec, D3Rec, DDSR, DiffMM, SDiff} model users' preference scores among the full item corpus, \ie user interacted one-hot vectors, by perturbing them within the probability simplex in the forward process.
While the reverse process attempts to reconstruct the preference scores, the simplex constraints---non-negativity and normalization---introduce additional optimization difficulties \cite{DDSE}.
Worse still, both diffusion-based recommenders assume a prior noise in the sampling process, \eg Gaussian \cite{DreamRec} or Bernoulli \cite{RecFusion}, which may poorly reflect the inherently discrete and sparse nature of the user preference data in recommendation scenarios.

Building on recent advances in discrete diffusion models \cite{concretescore, DDSE, RADD}, we propose \textbf{PreferGrow}, a discrete diffusion-based recommender that models the \textit{preference ratios} by fading and growing user preferences through forward and backward processes.
Unlike existing continuous diffusion-based recommenders that model preference scores (\eg probabilities that users interacted with items), PreferGrow directly models relative preference ratios over a discrete item set.
This formulation aligns more naturally with the discrete and ranking-oriented nature of recommendation tasks and avoids the strong constraints imposed by the probability simplex.
In the forward perturbation, PreferGrow fades user preference by replacing the preferred item with alternatives, enabling explicit negative sampling and alleviating data sparsity without relying on predefined noise distribution.
In the backward generation, PreferGrow reconstructs user preferences by iteratively growing the preference signal from the estimated ratios.
We further provide theoretical analysis showing that these forward and backward processes preserve key properties of discrete diffusion models.

PreferGrow offers a unified and theoretically grounded framework for discrete diffusion-based recommendation, characterized by a well-defined matrix-based formulation, flexible and interpretable preference ratio modeling aligned with diverse negative sampling strategies, and consistently superior empirical performance.
\begin{itemize}[leftmargin=*]
    \item \textbf{Theoretical foundation:} 
    At the core of the formulation for PreferGrow is an idempotent fading matrix used to replace preferred items during the forward perturbation process. 
    We provide a closed-form solution of this preference fading matrix and theoretically prove that its idempotent property is critical for ensuring both the Markov property and reversibility of the diffusion process.
    \item \textbf{Flexible modeling:} By parameterizing the preference fading matrix, PreferGrow flexibly supports point-wise, pair-wise, and hybrid preference ratios modeling. 
    These variants correspond to distinct and physically interpretable negative sampling strategies, enabling the framework to adapt to diverse modeling targets in recommendation.
    \item \textbf{Empirical validation:} Extensive experiments on five benchmark datasets demonstrate that PreferGrow consistently outperforms existing diffusion-based recommenders, validating the practical effectiveness of its theoretical foundations.
\end{itemize}

%% file: chapters/2-prelim.tex
\section{Preliminaries}\label{sec:prelim}

\textbf{User preference data} is represented as a pair $(u, i)$, where $u$ denotes the user and $i$ is the preferred item. 
The preference score $p(i|u)$ indicates the probability that the user $u$ interacts with item $i$.
In sequential recommendation, $u$ denotes the item sequence that the user has interacted with.

\textbf{Diffusion-based Recommenders} generally consist of three major components: a forward noise-adding process for perturbation, a generative modeling target, and a backward denoising process for reconstruction.
Specifically, \textbf{item-level diffusion-based recommenders} encode a preferred item $i$ into its dense embedding $\mathbf{x}_0$, and instantiate three components as follows \cite{GodDiff}:
\begin{itemize}[leftmargin=*]
    \item \textit{forward perturbation}:
    $\mathbf{x}_t = \sqrt{\alpha_t} \mathbf{x}_0 + \sqrt{1 - \alpha_t} \epsilon_t$, where $\alpha_t\in[0,1],\epsilon_t \sim \mathcal{N}(0, I)$. 
    \item \textit{modeling target}:
    network $\mathbf{s}_{\Theta}(\mathbf{x}_t,t,u)\approx \nabla_{\mathbf{x}_t} \log p_{t}(\mathbf{x}_t|u) = \frac{\sqrt{\alpha_t} \mathbf{x}_0 - \mathbf{x}_t}{1 - \alpha_t} = -\frac{\epsilon_t}{\sqrt{1 - \alpha_t}}$.
    \item \textit{backward generation}:
    $\mathbf{x}_s = \sqrt{\alpha_s} \frac{\mathbf{x}_t+(1-\alpha_t)\mathbf{s}_{\Theta}(\mathbf{x}_t,t,u)}{\sqrt{\alpha_t}} - \sqrt{1 - \alpha_s} \sqrt{1 - \alpha_t}\mathbf{s}_{\Theta}(\mathbf{x}_t,t,u), s<t$.
\end{itemize}
\za{\textbf{Preference score–level diffusion-based recommenders} represent the preferred item $i$ as a one‑hot preference score $\mathbf{x}_0$, modeling user preference over the entire item space. 
Early works \cite{DiffRec2,LD4MRec,D3Rec} adopt Gaussian noise priors $\mathcal{N}(\cdot,\cdot)$; however, Gaussian perturbations are incompatible with the probability simplex constraints inherent to preference scores. 
To address this mismatch, subsequent studies \cite{RecFusion,DDSR} introduce discrete priors to replace the Gaussian assumption that preserves the simplex structure throughout the diffusion process.
For instance, RecFusion \cite{RecFusion} adopts a Bernoulli prior $\mathcal{B}(\cdot;\cdot)$, resulting in a binomial diffusion formulation:}
\begin{itemize}[leftmargin=*]
    \item \textit{forward perturbation}:
    $\mathbf{x}_t \sim \mathcal{B}\left(\mathbf{x}_t;\alpha_t\mathbf{x}_0 +\frac{\alpha_t(1 - \alpha_t)}{2}\right), \alpha_t\in[0,1]$. 
    \item \textit{modeling target}:
    $\mathcal{B}\left(\mathbf{x}_{t-1}; \mathbf{s}_{\Theta}(\mathbf{x}_t,t,u)\right)\approx p(\mathbf{x}_{t-1} \mid \mathbf{x}_t)$.
    \item \textit{backward generation}:
    $\mathbf{x}_{t-1} \sim \mathcal{B}\left(\mathbf{x}_{t-1}; \mathbf{s}_{\Theta}(\mathbf{x}_t,t,u)\right)$.
\end{itemize}
The other approach employs a Categorical prior $\text{Cat}(\cdot;\cdot)$:
\begin{itemize}[leftmargin=*]
  \item \textit{forward perturbation:}
        $p_t(\mathbf{x}_t \mid \mathbf{x}_{t-1})
        = \text{Cat}(\mathbf{x}_t; \overline{\mathbf{Q}}_t \mathbf{x}_{0})$ where $\overline{\mathbf{Q}}_t = \prod_{i=1}^{t} \mathbf{Q}_i$.
  \item \textit{modeling target:} 
        $\mathbf{s}_{\Theta}(\mathbf{x}_t,t,u)\approx
        \mathbf{x}_0$.
  \item \textit{backward generation:}
        $
          p_t\left(\mathbf{x}_{t-1} \mid
          \mathbf{x}_t, \mathbf{x}_0 = \mathbf{s}_{\Theta}(\mathbf{x}_t, t, u)\right)
          = \text{Cat}\left(
              \mathbf{x}_{t-1};
              \frac{\mathbf{Q}_t^\top \mathbf{x}_t \odot
              \overline{\mathbf{Q}}_{t-1} \mathbf{s}_{\Theta}(\mathbf{x}_t, t, u)}
              {\mathbf{x}_t^\top \overline{\mathbf{Q}}_t
                \mathbf{s}_{\Theta}(\mathbf{x}_t, t, u)}
              \right)$.
\end{itemize}


\textbf{Preference Ratios} $\textcolor{orange}{\log \frac{p(i_p \mid u)}{p(i_d \mid u)}}$ characterize the relative preference between items, and a positive value indicates a more preferred item $i_p$ and a less preferred item $i_d$.
Notably, preference modeling in RLHF \cite{RLHF1, OnlineRLHF, IterRLHF} and DPO \cite{DPO, BetaDPO, SDPO, AlphaDPO} is grounded in the Bradley–Terry model \cite{BTPrefer}, which explicitly models preference ratios as in Equation \eqref{BTM}. 
This underscores the effectiveness of preference ratios in more faithfully capturing user preferences.
\begin{equation}
\label{BTM}
p(i_p \succ i_d \mid u) = \frac{p(i_p \mid u)}{p(i_p \mid u) + p(i_d \mid u)} = \frac{1}{1 + \exp{\left( - \textcolor{orange}{\log \frac{p(i_p \mid u)}{p(i_d \mid u)}} \right)}}=\sigma(\textcolor{orange}{\log \frac{p(i_p \mid u)}{p(i_d \mid u)}}).
\end{equation}


\noindent\textbf{Discrete Diffusion Models} previously are formulated based on the following Kolmogorov forward and backward equations \cite{DDM1, DDM2, DDSE}:
\begin{equation}
\label{DDMF}
\frac{\partial \mathbf{P}_{t|s}}{\partial t} = \mathbf{Q}_{t} \mathbf{P}_{t|s}  ,\qquad
\frac{\partial \mathbf{P}_{s|t}}{\partial t} = \mathbf{R}_{s} \mathbf{P}_{s|t},
\end{equation}
where $\mathbf{P}_{t|s}$ denotes the transition probability matrix from time $s$ to time $t$, $\mathbf{Q}_t$ is the forward transition rate matrix at time $t$, and $\mathbf{R}_s$ is the reverse-time transition rate matrix at time $s$.
Discrete diffusion models further model the ratios of data distributions through the following score entropy loss \cite{DDSE}:
\begin{align}
\label{SE}
&\mathcal{L}_{SE} =\mathbb{E}_{x_0\sim p_{data}}\mathbb{E}_{t\in U[0,T]}\mathbb{E}_{x_t\sim p_{t|0}(\cdot|x_0)}\left[\sum\limits_{y\in \mathcal{X}}\mathbf{Q}_t(x_t,y) \cdot l_{SE}(x_0,x_t,y) \right], \\
\label{SE2}
& l_{SE}(x_0,x_t,y) = e^{\mathbf{s}_{\Theta}(y,x_t)}- \frac{p_{t|0}(y|x_0)}{p_{t|0}(x_t|x_0)}\mathbf{s}_{\Theta}(y,x_t)+\frac{p_{t|0}(y|x_0)}{p_{t|0}(x_t|x_0)}[\log \frac{p_{t|0}(y|x_0)}{p_{t|0}(x_t|x_0)} -1],
\end{align}
where $\mathcal{X}$ denotes the discrete state space, and $\mathbf{s}_{\Theta}(y,x_t)$ is the output of a neural network that estimates the preference ratio between $y$ and $x_t$ at timestep $t$.






%% file: chapters/3-method.tex
\section{Method}\label{sec:method}
\begin{figure}[t!]
    \vspace{-15pt}
    \centering
    \includegraphics[width=\textwidth]{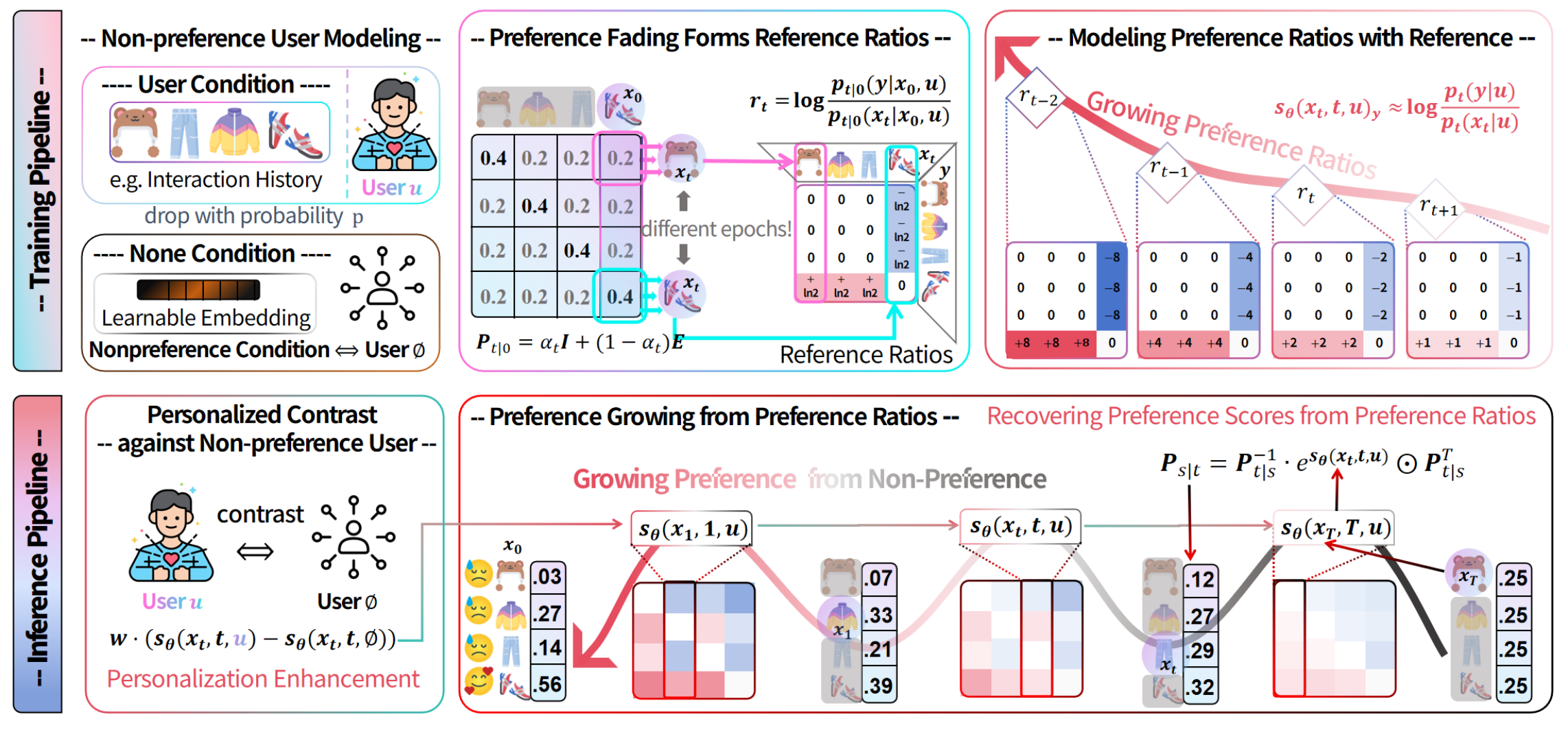}
    \caption{The overall training and inference pipeline of PreferGrow under the pair-wise setting.}
    \vspace{-10pt}
    \label{fig:framework}
\end{figure}

As illustrated in Figure \ref{fig:framework}, the pipeline of PreferGrow consists of three stages: the preference fading discrete diffusion process (Section~\ref{sec:preferfade}), the modeling of preference ratios with reference via score entropy loss (Section~\ref{sec:preferratios}), and the preference growing reverse generation process (Section~\ref{sec:prefergrow}).  
Additionally, PreferGrow models a non-preference user and achieves personalized enhancement (Section~\ref{sec:non-prefer}). 
Theoretical proofs are provided in Appendix~\ref{app:proofs}, while the computation of score entropy loss under different fading matrix settings is detailed in Appendix~\ref{app:variants}. 

\subsection{Forward Perturbation: Preference Fading Discrete Diffusion Process} \label{sec:preferfade}


PreferGrow operates directly on the discrete item set $\mathcal{X}$, wherein $x_0$ represents the preferred item $i$ of user $u$.
PreferGrow first fades user preference based on whether to retain or not to retain the preferred item $x_0$. 
The whole fading process is further achieved by progressively decreasing the probability of retention $\alpha_t$ from timestep $0$ to $T$, fading user preferences towards non-preference.

\subsubsection{Preference Fading Forms Reference Ratios}
Building upon recent advances in discrete DM \cite{concretescore, DDSE, RADD}, we introduce a preference fading discrete diffusion model tailored for recommendation.
For a user preference data $(u, x_0)$, we fade user preference by retaining the target item $x_t = x_0$ with probability $\alpha_t$ and replacing $x_0$ according to discrete distribution $\mathbf{E}(\mathcal{X}, x_0)$ as the following equation:
\begin{equation}
\label{P0t_element}
p_{t|0}(x_t|x_0)=\alpha_t\delta_{x_0}(x_t)+(1-\alpha_t)\mathbf{E}(x_t,x_0), x_t\in\mathcal{X}.
\end{equation}
where Dirac delta function $\delta_{x_0}(x_t)$ equals 1 when $x_t = x_0$, and 0 otherwise.
Also, in matrix form:
\begin{equation}
\label{P0t_matrix}
\mathbf{P}_{t|0} = \alpha_t \mathbf{I} + (1 - \alpha_t) \mathbf{E},
\end{equation}
where $\mathbf{I}$ is the identity matrix, representing retention, and $\mathbf{E}$ is a matrix whose column sums equal $1$, defining the replacing mode for perturbation.
For simplicity, we refer to $\mathbf{E}$ as fading matrix.
We further demonstrate that the idempotence of fading matrix $\mathbf{E}$ ensures the Markov property and reversibility of preference fading discrete diffusion process, which is crucial for the reverse generation.

\begin{theorem} \label{thm:preferfade}
Suppose $\alpha_t : [0,T] \to [0,1]$ is a strictly decreasing function with $\alpha_0 = 1$ and $\alpha_T = 0$. 
If $\mathbf{E} \in \mathbb{R}^{N \times N}$ is idempotent, \ie $\mathbf{E}^2 = \mathbf{E}$,
the following properties hold:
\begin{itemize}[leftmargin=*]
    \item \textbf{Markov property:} $\{\mathbf{P}_{t|0}\}_{t=0}^{T}$ is a Markov process and satisfies the Chapman-Kolmogorov equation:
    \begin{equation}
    \label{Pst}
    \mathbf{P}_{t|s} := \frac{\alpha_t}{\alpha_s} \mathbf{I} + \left(1 - \frac{\alpha_t}{\alpha_s} \right) \mathbf{E},
    \end{equation}
    \begin{equation}
    \label{CKE}
    \mathbf{P}_{t|r} = \mathbf{P}_{t|s} \mathbf{P}_{s|r}, \quad \text{for all } 0 \leq r \leq s \leq t \leq 1.
    \end{equation}
    \item \textbf{Invertibility:} Each $\mathbf{P}_{t|s}$ is invertible, and its inverse is given by:
    \begin{equation}
    \label{Inv_Pts}
    \mathbf{P}_{t|s}^{-1} = \frac{\alpha_s}{\alpha_t} \mathbf{I} + \left(1 - \frac{\alpha_s}{\alpha_t} \right) \mathbf{E}.
    \end{equation}
\end{itemize}
\end{theorem}

Upon obtaining $x_t$ through preference fading, as illustrated in Figure \ref{fig:framework}, we can compute the reference ratios for all items $\mathcal{X}$ at timestep $t$ with fading awareness $1-\alpha_t$ as follows:
\begin{equation}
    \label{ref_ratios}
    r_t(x_0, x_t, y)=\log \frac{p_{t|0}(y \mid x_0, u)}{p_{t|0}(x_t \mid x_0, u)}=\log \frac{\alpha_t\delta_{x_0}(y)+(1-\alpha_t)\mathbf{E}(y,x_0)}{\alpha_t\delta_{x_0}(x_t)+(1-\alpha_t)\mathbf{E}(x_t,x_0)}, \forall y\in\mathcal{X}.
\end{equation}

As shown in Figure \ref{fig:framework}, when the preferred item $x_0$ is retained, \ie $x_t = x_0$, the reference ratio $r_t(x_0, x_t=x_0, y\ne x_0)$ is negative, reflecting a tendency to stay with the current item. 
Conversely, when $x_0$ is replaced, \ie $x_t \ne x_0$, the reference ratio $r_t(x_0, x_t\ne x_0, y= x_0)$ tends to be positive, highlighting the preference for the positive item $x_0$.
In this case, the faded item $x_t$ can be interpreted as a negative item, establishing a conceptual connection to the mechanism of negative sampling.


\subsubsection{Design Paradigms of the Idempotent Fading Matrix} \label{sec:replacement}

Given that designing such an idempotent fading matrix is crucial, we derive that, if the preference fading process converges to a unified state, the fading matrix $ \mathbf{E} $ can be expressed in closed form. 

\begin{proposition} \label{pro:fade}
    Suppose the preference fading Markov process $\{\mathbf{P}_{t|0}\}_{t=0}^{T}$ converges to a unified non-preference state $\vec{p}_T$. 
    Then, the fading matrix $\mathbf{E}$ can be expressed in closed form as:
    \begin{equation}
    \label{E_rank1}
    \mathbf{E} = \frac{ \vec{p}_T \vec{1}^\top }{ \vec{1}^\top \vec{p}_T }.
    \end{equation}
\end{proposition}



\hugq{
Equation~\eqref{E_rank1} gives a rank-$1$ solution for the fading matrix $\mathbf{E}$. 
Most existing discrete diffusion models \cite{DDSE, concretescore, RADD, MDD1, MDD2, LDR, SDDM} implicitly adopt such rank-$1$ instances (i.e., the absorbing and uniform cases). 
In Appendix~\ref{app:variants}, examine several representative rank-$1$ configurations of $\mathbf{E}$ that, in a physical sense, correspond to distinct negative sampling strategies and thus induce different forms of preference ratios.
Specifically, by parameterizing the rank-$1$ preference-fading matrix, PreferGrow flexibly supports point-wise, pair-wise, and hybrid preference-ratio modeling. 
These variants align with interpretable negative-sampling schemes, enabling the framework to adapt to diverse recommendation objectives.
For each setting, we provide analytical solutions for the fading matrix, closed-form expressions for the reference ratios, and simplified training losses.
We believe prior works on negative sampling~\cite{Bert4Rec,InfoNCE} offer valuable guidance for designing more effective rank-$1$ fading matrices; a systematic exploration is left for future work.
Beyond the rank-$1$ case, Appendix~\ref{app:rankr} derives a general rank-$r$ solution for $\mathbf{E}$.
We present a closed-form characterization together with a discussion of its physical interpretation. 
Notably, a rank-$r$ fading matrix induces a quantization of the item space --- i.e., a partition into $r$ components. 
A comprehensive study of this quantization mechanism and its algorithmic implications is deferred to our future work.
}

\subsection{Modeling Target: Preference Ratios with Reference via Score Entropy} \label{sec:preferratios}
\hugq{
We begin by computing the reference ratios $r_t(x_0, x_t, y)$ induced by preference fading, as defined in Equation~\eqref{ref_ratios}. 
These reference ratios are then employed to guide the modeling of the user-conditioned preference ratio 
$\log\frac{p_t(y\mid u)}{p_t(x_t\mid u)}$, where $u$ denotes the user context.
Concretely, given a training instance $(u,x_0)$ in sequential recommendation, we encode the interaction history $u$ with a sequential recommender to obtain a user embedding $\mathbf{u}=\mathrm{SeqRec}(u)$ (instantiated as SASRec~\cite{SASRec} in our experiments). 
We then parameterize a learnable function $s_{\Theta}$ to estimate the preference ratio for each $y\in\mathcal{X}$ at time $t$:
\begin{equation}
s_{\Theta}(x_t,t,u)_y 
= \mathbf{y}^\top \,\textbf{MLP}\big(\operatorname{concat}(\mathbf{x}_t,\mathbf{t},\mathbf{u})\big),
\quad y\in\mathcal{X},
\label{s_theta_prefgrow}
\end{equation}
where $\mathbf{y}$ and $\mathbf{x}_t$ are the embeddings of item $y$ and the faded item $x_t$, respectively, and $\mathbf{t}$ is the embedding of timestep $t$.
For the training objective, we adopt the well-known score entropy loss for discrete diffusion modeling \cite{DDSE}, as defined in Equations~\eqref{SE} and~\eqref{SE2}. 
To legitimately employ it, we first verify that PreferGrow satisfies the applicability conditions of score entropy loss.


\begin{proposition}\label{pro:rate}
The preference-fading discrete diffusion $\{\mathbf{P}_{t|s}\}_{0\le s\le t\le T}$ with an idempotent fading matrix $\mathbf{E}$ satisfies the Kolmogorov forward equation:
\begin{equation}\label{KFE}
\frac{\partial \mathbf{P}_{t|s}}{\partial t} \;=\; \mathbf{Q}_{t}\,\mathbf{P}_{t|s}.
\end{equation}
Here the rate matrix $\mathbf{Q}_{t}$ and $\alpha_t$ are defined succinctly by
\begin{equation}\label{rate_matrix}
\mathbf{Q}_{t}:=\lim_{s\to t}\frac{\partial \mathbf{P}_{t|s}}{\partial t}=\beta(t)\,(\mathbf{E}-\mathbf{I}),
\quad
\alpha_t:=\exp\Big(-\int_{0}^{t}\beta(\tau)\,\mathrm{d}\tau\Big),
\end{equation}
with $\beta(\tau)>0$ and $\mathbf{I}$ the identity matrix.
\end{proposition}

To better explain the effectiveness of the Score Entropy (SE) loss in recommendation, we theoretically analyze its connection to the Binary Cross-Entropy (BCE) loss \cite{Bert4Rec}.

\begin{proposition}\label{prop:se-bce-link}
Define the soft label $\pi_{y \succ x_t \mid x_0} = p(y \succ x_t \mid x_0) = \sigma(r_t(x_0, x_t, y))$ as in Equation \eqref{BTM}, where $\sigma$ denotes the sigmoid function, and use it as the label for the BCE loss. 
Treating $\sigma(s_{\Theta}(x_t,t,u)_y)$ as the prediction for the BCE loss yields the soft-label BCE objective:
$$
\mathcal{L}_{\textit{sBCE}} = -\pi_{y \succ x_t \mid x_0} \log\sigma(s_{\Theta}(x_t,t,u)_y) - (1 - \pi_{y \succ x_t \mid x_0})\log(1 - \sigma(s_{\Theta}(x_t,t,u)_y)).
$$
From a gradient perspective, the SE loss and the soft-label BCE loss are related by:
\begin{equation}\label{grad_loss}
\nabla_{s_\Theta} \mathcal{L}_{\textit{SE}} = (1 + e^{s_{\Theta}(x_t,t,u)_y})(1 + e^{r_t(x_0,x_t,y)}) \nabla_{s_\Theta} \mathcal{L}_{\textit{sBCE}}.
\end{equation}
\end{proposition}

Consequently, the SE loss and the soft-label BCE loss share the same descent direction (up to a positive scalar) and attain the same optima.
Hence, although originating from discrete diffusion modeling, SE loss is equally well suited to ranking-oriented recommendation tasks as the BCE loss.
}

\subsection{Backward Generation: Preference Growing from Preference Ratios} \label{sec:prefergrow}

After estimating the preference ratios via score entropy, we reverse the preference fading process to grow user preferences from the non-preference state.  
We first demonstrate that this reverse process of preference fading, referred to as the preference growing process, is Markovian and satisfies the Kolmogorov backward equation.  
Together with Proposition \ref{pro:rate}, this implies that PreferGrow with an idempotent fading matrix inherits the key properties of previous discrete diffusion models \cite{DDM2,DDSE}.

\begin{theorem} \label{thm:prefergrow}
If the preference fading process $\{\mathbf{P}_{t|0}\}_{t=0}^{T}$ with an idempotent fading matrix $\mathbf{E}$ converges to a unified non-preference state $\vec{p}_T$, the reverse preference growing process $\{\mathbf{P}_{s|T}\}_{s=T}^{0}$ holds:
\begin{itemize}[leftmargin=*]
    \item \textbf{Markov property:} $\{\mathbf{P}_{s|T}\}_{s=T}^{0}$ is also a Markov process:
    \begin{equation}
    \label{Pst_matrix}
    \mathbf{P}_{s|t} = \mathbf{P}_{t|s}^{-1} \cdot \textcolor{blue}{\left[\vec{p}_t\cdot \left(\frac{1}{\vec{p}_t}\right)^\top\right]} \odot \mathbf{P}_{t|s}^{\top}, \quad \forall\, 0 \le s \le t \le T.
    \end{equation}
    \item \textbf{Kolmogorov backward equation:} Preference growing $\{\mathbf{P}_{s|T}\}_{s=T}^{0}$ satisfies:
    \begin{equation}
    \label{P_st_matrix}
   \frac{\partial \mathbf{P}_{s|t}}{\partial t} = \mathbf{R}_{s} \mathbf{P}_{s|t}, \quad \forall\, 0 \le s \le t \le T. 
    \end{equation}
    \begin{equation}
    \label{Rs_matrix}
    \mathbf{R}_{s} := \lim\limits_{t\rightarrow s} \frac{\partial \mathbf{P}_{s|t}}{\partial s}=\mathbf{Q}_{s}^\top \odot \textcolor{blue}{\left[\vec{p}_s\cdot \left(\frac{1}{\vec{p}_s}\right)^\top\right]} - \mathbf{Q}_{s} \cdot \textcolor{blue}{\left[\vec{p}_s\cdot \left(\frac{1}{\vec{p}_s}\right)^\top\right]} \odot \mathbf{I}.
    \end{equation}
\end{itemize}
The matrix $\textcolor{blue}{\vec{p}_t \cdot \left(\frac{1}{\vec{p}_t}\right)^\top}$ denotes exponential preference ratios with entries given by $\frac{p_t(y \mid u)}{p_t(x_t \mid u)} \forall x_t, y \in \mathcal{X}$.
\end{theorem}

At this point, we are equipped with the approximated preference ratios $\mathbf{s}_{\Theta}$ to grow user preferences.  
Given a user condition $u$, we first sample an item $x_T$ from the non-preference state $\vec{p}_T$, and then compute the reverse transition matrix $\mathbf{P}_{s|t}$ to progressively grow user preference over time as:

\begin{equation}
    \label{Pst_vec}
   p_{s|t}(x_s=y|x_t) = {p}_{t|s}(x_t|x_s=y)\cdot \sum\limits_{z\in\mathcal{X}}{p}_{t|s}^{-1}(x_t=y|x_s=z) \cdot e^{\mathbf{s}_{\Theta}(x_t,t,u)_z}, \forall\, 0 \le s \le t \le T. 
\end{equation}

As illustrated in Figure \ref{fig:framework}, we iteratively grow user preferences backward until reaching $x_0$.  
Finally, based on the preference scores at timestep 0, we recommend the top-$K$ items with the highest scores.



\subsection{Modeling Non-preference User for Personalization} \label{sec:non-prefer}

\subsubsection{Non-preference User Modeling}

To enhance the personalization of preference ratios, we introduce the modeling of non-preference users---\ie cold-start users with no interaction history---during training.
This allows the model to perform personalized contrast against the non-preference user during inference, thereby reinforcing user-specific signals for personalized recommendations.
As illustrated in Figure \ref{fig:framework}, we randomly drop the user condition $u$ with a fixed probability $p$ during training, and replace it with a learnable embedding that represents the non-preference user $\phi$.

\subsubsection{Personalized Contrast against Non-preference User}
The non-preference user $\phi$ is defined such that $p_{t|0}(x_t|x_0, \phi) = p_{t|0}(x_t|x_0)$.  
According to Bayes' theorem $p_{t|0}(x_t|x_0, u) = \frac{p_{t|0}(x_t|x_0, \phi)}{p(u|x_0)} \cdot p(u |x_0, x_t)$, there holds the following ratio condition:
\begin{equation}
\label{}
\frac{p_{t|0}(x_t=y|x_0,u)}{p_{t|0}(x_t=x|x_0,u)} = \frac{p_{t|0}(x_t=y|x_0)}{p_{t|0}(x_t=x|x_0)} \cdot \frac{p(u|x_0,x_t=y)}{p(u|x_0,x_t=x)}.
\end{equation}
We then estimate the ratio of the likelihood with a personalization strength parameter $w$ as \cite{CFG}:
\begin{equation}
\label{exp_cfg}
\frac{p(u|x_0,x_t=y)}{p(u|x_0,x_t=x)} \approx \left\{ \frac{p_{t|0}(x_t=y|x_0,\phi)}{p_{t|0}(x_t=x|x_0,\phi)} \cdot \left[\frac{p_{t|0}(x_t=y|x_0,u)}{p_{t|0}(x_t=x|x_0,u)}\right]^{-1} \right\}^{w}.
\end{equation}
Therefore, the personalization enhanced preference ratios with strength $w$ are computed as follows:
\begin{equation}
\label{log_cfg}
\hat{s}_{\Theta}(x_t,t,u)_y = w\cdot s_{\Theta}(x_t,t,u)_y + (1-w) \cdot s_{\Theta}(x_t,t,\phi)_y.
\end{equation}


%% file: chapters/4-exp.tex
\begin{table}[t!]
\scriptsize
\caption{Performance comparison across different datasets and baselines.}
\begin{tabularx}{0.98\textwidth}{c|c|cc|cc|cc|cc|cc}
\toprule
\multicolumn{1}{c|}{Dataset} &  & \multicolumn{2}{c|}{MoviesLens} & \multicolumn{2}{c|}{Steam} & \multicolumn{2}{c|}{Beauty} & \multicolumn{2}{c|}{Toys} & \multicolumn{2}{c}{Sports} \\
\midrule
Method & & HR & NDCG & HR & NDCG & HR & NDCG & HR & NDCG & HR & NDCG\\
\midrule
\multirow{3}{*}{ SASRec}
& @5
& 0.0905 & 0.0502
& \underline{0.0321} & \underline{0.0192}
& 0.0326  & 0.0221
& 0.0340 & 0.0258 
& 0.0154 & 0.0101 
\\
& \cellcolor{lightgray1}{@10}
& \cellcolor{lightgray1}{\underline{0.1709}} 
& \cellcolor{lightgray1}{0.0760}
& \cellcolor{lightgray1}{0.0572} 
&  \cellcolor{lightgray1}{\underline{0.0272}}
& \cellcolor{lightgray1}{\underline{0.0438}}
& \cellcolor{lightgray1}{0.0257}
& \cellcolor{lightgray1}{0.0438}
& \cellcolor{lightgray1}{0.0289} 
& \cellcolor{lightgray1}{\underline{0.0230}}
& \cellcolor{lightgray1}{0.0126}
\\
& \cellcolor{lightgray2}{@20}
& \cellcolor{lightgray2}{\underline{0.2899}} 
& \cellcolor{lightgray2}{\underline{0.1059}}
& \cellcolor{lightgray2}{0.0957}
& \cellcolor{lightgray2}{0.0368}
& \cellcolor{lightgray2}{0.0595}
& \cellcolor{lightgray2}{0.0297}
& \cellcolor{lightgray2}{0.0499} 
& \cellcolor{lightgray2}{0.0305}  
& \cellcolor{lightgray2}{\underline{0.0298}}
& \cellcolor{lightgray2}{0.0143}
\\
\multirow{3}{*}{Caser}
& @5
& \underline{0.0925} & \underline{0.0576}
& 0.0258 & 0.0163
& 0.0170 & 0.0106 
& 0.0093 & 0.0072 
& 0.0051 & 0.0030  
\\
&\cellcolor{lightgray1}{@10}
& \cellcolor{lightgray1}{0.1605}
& \cellcolor{lightgray1}{\underline{0.0794}}
& \cellcolor{lightgray1}{0.0446}
& \cellcolor{lightgray1}{0.0223}
& \cellcolor{lightgray1}{0.0219}
& \cellcolor{lightgray1}{0.0122}
& \cellcolor{lightgray1}{0.0129}
& \cellcolor{lightgray1}{0.0083}
& \cellcolor{lightgray1}{0.0093}
& \cellcolor{lightgray1}{0.0043}   
\\
& \cellcolor{lightgray2}{@20}
& \cellcolor{lightgray2}{0.2592}
& \cellcolor{lightgray2}{0.1042}
& \cellcolor{lightgray2}{0.0736}
& \cellcolor{lightgray2}{0.0295} 
& \cellcolor{lightgray2}{0.0295}
& \cellcolor{lightgray2}{0.0141}
& \cellcolor{lightgray2}{0.0185}
& \cellcolor{lightgray2}{0.0098}   
& \cellcolor{lightgray2}{0.0149}
& \cellcolor{lightgray2}{0.0057}
\\
\multirow{3}{*}{GRURec}
& @5
& 0.0892 & 0.0551
& 0.0255 & 0.0156
& 0.0130 & 0.0082 
& 0.0124 & 0.0078 
& 0.0070 & 0.0048  
\\
& \cellcolor{lightgray1}{@10}
& \cellcolor{lightgray1}{0.1534}
& \cellcolor{lightgray1}{0.0757}
& \cellcolor{lightgray1}{0.0441}
& \cellcolor{lightgray1}{0.0216}
& \cellcolor{lightgray1}{0.0188}
& \cellcolor{lightgray1}{0.0101} 
& \cellcolor{lightgray1}{0.0180}
& \cellcolor{lightgray1}{0.0092}
& \cellcolor{lightgray1}{0.0107}
& \cellcolor{lightgray1}{0.0059}
\\
& \cellcolor{lightgray2}{@20}
& \cellcolor{lightgray2}{0.2501}
& \cellcolor{lightgray2}{0.1000}
& \cellcolor{lightgray2}{0.0769}
& \cellcolor{lightgray2}{0.0298}
& \cellcolor{lightgray2}{0.0313}
& \cellcolor{lightgray2}{0.0132}  
& \cellcolor{lightgray2}{0.0304}
& \cellcolor{lightgray2}{0.0117}
& \cellcolor{lightgray2}{0.0171}
& \cellcolor{lightgray2}{0.0076} 
\\
\midrule 
\multirow{3}{*}{DreamRec}
& @5
& 0.0676 & 0.0437
& 0.0109 & 0.0073
& 0.0300 & 0.0247 
& 0.0381 & 0.0304
& 0.0095 & 0.0084
\\
& \cellcolor{lightgray1}{@10}
& \cellcolor{lightgray1}{0.1083}
& \cellcolor{lightgray1}{0.0568}
& \cellcolor{lightgray1}{0.0157} 
& \cellcolor{lightgray1}{0.0088}
& \cellcolor{lightgray1}{0.0353}
& \cellcolor{lightgray1}{0.0265}
& \cellcolor{lightgray1}{0.0412}
& \cellcolor{lightgray1}{0.0314}
& \cellcolor{lightgray1}{0.0112}
& \cellcolor{lightgray1}{0.0089}
\\
& \cellcolor{lightgray2}{@20}
& \cellcolor{lightgray2}{0.1610}
& \cellcolor{lightgray2}{0.0701}
& \cellcolor{lightgray2}{0.0218} 
& \cellcolor{lightgray2}{0.0104}
& \cellcolor{lightgray2}{0.0402}
& \cellcolor{lightgray2}{0.0277}
& \cellcolor{lightgray2}{0.0463}
& \cellcolor{lightgray2}{0.0327}
& \cellcolor{lightgray2}{0.0149}
& \cellcolor{lightgray2}{0.0098}
\\ 
\multirow{3}{*}{PreferDiff}
& @5
& 0.0538 & 0.0349
& 0.0167 & 0.0105
& \underline{0.0335} & \underline{0.0272} 
& \underline{0.0386} & \underline{0.0308}
& \underline{0.0168} & \underline{0.0130} 
\\
& \cellcolor{lightgray1}{@10}
& \cellcolor{lightgray1}{0.0852}
& \cellcolor{lightgray1}{0.0450}
& \cellcolor{lightgray1}{0.0297}
& \cellcolor{lightgray1}{0.0145}
& \cellcolor{lightgray1}{0.0434}
& \cellcolor{lightgray1}{\underline{0.0304}}  
& \cellcolor{lightgray1}{\underline{0.0494}}
& \cellcolor{lightgray1}{\underline{0.0343}}
& \cellcolor{lightgray1}{0.0211}
& \cellcolor{lightgray1}{\underline{0.0144}}
\\
& \cellcolor{lightgray2}{@20}
& \cellcolor{lightgray2}{0.1270} 
& \cellcolor{lightgray2}{0.0555} 
& \cellcolor{lightgray2}{0.0526} 
& \cellcolor{lightgray2}{0.0203}
& \cellcolor{lightgray2}{0.0577}
& \cellcolor{lightgray2}{\underline{0.0340}} 
& \cellcolor{lightgray2}{\underline{0.0644}}
& \cellcolor{lightgray2}{\underline{0.0380}}  
& \cellcolor{lightgray2}{0.0256}
& \cellcolor{lightgray2}{\underline{0.0155}}
\\
\midrule
\multirow{3}{*}{DiffRec}
& @5
& 0.0266 & 0.0121 
& 0.0268 & 0.0143
& 0.0095 & 0.0055
& 0.0053 & 0.0028 
& 0.0063 & 0.0035
\\
& \cellcolor{lightgray1}{@10}
& \cellcolor{lightgray1}{0.0889}
& \cellcolor{lightgray1}{0.0320}
& \cellcolor{lightgray1}{\underline{0.0657}}
& \cellcolor{lightgray1}{0.0268}
& \cellcolor{lightgray1}{0.0218}
& \cellcolor{lightgray1}{0.0094}
& \cellcolor{lightgray1}{0.0153}
& \cellcolor{lightgray1}{0.0059}
& \cellcolor{lightgray1}{0.0117}
& \cellcolor{lightgray1}{0.0052}
\\
& \cellcolor{lightgray2}{@20}
& \cellcolor{lightgray2}{0.1768}
& \cellcolor{lightgray2}{0.0541}
& \cellcolor{lightgray2}{\underline{0.1159}}
& \cellcolor{lightgray2}{\underline{0.0394}}  
& \cellcolor{lightgray2}{0.0355}
& \cellcolor{lightgray2}{0.0128}
& \cellcolor{lightgray2}{0.0226}
& \cellcolor{lightgray2}{0.0078}
& \cellcolor{lightgray2}{0.0194}
& \cellcolor{lightgray2}{0.0072}
\\
\multirow{3}{*}{DDSR}
& @5
& 0.0871 & 0.0533
& 0.0252 & 0.0157
& 0.0291 & 0.0217
& \underline{0.0386} & 0.0302
& 0.0146 & 0.0109
\\
& \cellcolor{lightgray1}{@10}
& \cellcolor{lightgray1}{0.1523}
& \cellcolor{lightgray1}{0.0742}
& \cellcolor{lightgray1}{0.0437}
& \cellcolor{lightgray1}{0.0216}
& \cellcolor{lightgray1}{0.0434}
& \cellcolor{lightgray1}{0.0262}
& \cellcolor{lightgray1}{0.0479}
& \cellcolor{lightgray1}{0.0332}
& \cellcolor{lightgray1}{0.0213}
& \cellcolor{lightgray1}{0.0130}
\\
& \cellcolor{lightgray2}{@20}
& \cellcolor{lightgray2}{0.2450}
& \cellcolor{lightgray2}{0.0975}
& \cellcolor{lightgray2}{0.0736}
& \cellcolor{lightgray2}{0.0291}
& \cellcolor{lightgray2}{\underline{0.0608}}
& \cellcolor{lightgray2}{0.0306}
& \cellcolor{lightgray2}{0.0618}
& \cellcolor{lightgray2}{0.0367}
& \cellcolor{lightgray2}{\underline{0.0298}}
& \cellcolor{lightgray2}{0.0151}
\\
\midrule
\cellcolor{blue!10}{PreferGrow} & Impr.
& \cellcolor{green!10}{\cmark} &\cellcolor{green!10}{\cmark}
& \cellcolor{green!10}{\cmark} &\cellcolor{green!10}{\cmark}
& \cellcolor{green!10}{\cmark} &\cellcolor{green!10}{\cmark}
& \cellcolor{green!10}{\cmark} &\cellcolor{green!10}{\cmark}
& \cellcolor{green!10}{\cmark} &\cellcolor{green!10}{\cmark}
\\
\midrule  
\multirow{3}{*}{Hybrid}
& @5
& 0.1409 & 0.0911
& \textbf{0.0615} & \textbf{0.0406}
& \textbf{0.0420} & \textbf{0.0322}
& \textbf{0.0419} & \textbf{0.0311}
& \textbf{0.0207} & \textbf{0.0142}
\\
& \cellcolor{lightgray1}{@10}
& \cellcolor{lightgray1}{0.2229}
& \cellcolor{lightgray1}{0.1174}
& \cellcolor{lightgray1}{\textbf{0.0935}}
& \cellcolor{lightgray1}{\textbf{0.0508}}
& \cellcolor{lightgray1}{\textbf{0.0532}}
& \cellcolor{lightgray1}{\textbf{0.0358}}
& \cellcolor{lightgray1}{0.0480}
& \cellcolor{lightgray1}{0.0331}
& \cellcolor{lightgray1}{\textbf{0.0267}}
& \cellcolor{lightgray1}{\textbf{0.0162}}
\\
& \cellcolor{lightgray2}{@20}
& \cellcolor{lightgray2}{0.3355} 
& \cellcolor{lightgray2}{0.1458}
& \cellcolor{lightgray2}{\textbf{0.1413}}
& \cellcolor{lightgray2}{\textbf{0.0628}}
& \cellcolor{lightgray2}{\textbf{0.0708}}
& \cellcolor{lightgray2}{\textbf{0.0402}}
& \cellcolor{lightgray2}{0.0625}
& \cellcolor{lightgray2}{0.0367}
& \cellcolor{lightgray2}{\textbf{0.0343}}
& \cellcolor{lightgray2}{\textbf{0.0181}}
\\ 
\multirow{3}{*}{Adaptive}
& @5
& \textbf{0.1413} & \textbf{0.0912}
& 0.0583 & 0.0375
& 0.0396 & 0.0310
& 0.0413 & 0.0304
& 0.0168  & 0.0127
\\
& \cellcolor{lightgray1}{@10}
& \cellcolor{lightgray1}{\textbf{0.2240}}
& \cellcolor{lightgray1}{\textbf{0.1177}}
& \cellcolor{lightgray1}{0.0914}
& \cellcolor{lightgray1}{0.0481}
& \cellcolor{lightgray1}{0.0508}
& \cellcolor{lightgray1}{0.0347}
& \cellcolor{lightgray1}{\textbf{0.0513}}
& \cellcolor{lightgray1}{\textbf{0.0337}}
& \cellcolor{lightgray1}{0.0207}
& \cellcolor{lightgray1}{0.0140}
\\
& \cellcolor{lightgray2}{@20}
& \cellcolor{lightgray2}{\textbf{0.3362}} 
& \cellcolor{lightgray2}{\textbf{0.1460}}
& \cellcolor{lightgray2}{0.1387}
& \cellcolor{lightgray2}{0.0600}
& \cellcolor{lightgray2}{0.0610}
& \cellcolor{lightgray2}{0.0373}
& \cellcolor{lightgray2}{\textbf{0.0642}}
& \cellcolor{lightgray2}{\textbf{0.0370}}
& \cellcolor{lightgray2}{0.0267}
& \cellcolor{lightgray2}{0.0155}
\\
\bottomrule
\end{tabularx}
\label{tab:overall}
\end{table}

\begin{figure*}[ht!]
    \vspace{-10pt}
    \begin{subfigure}{0.48\linewidth}
        \centering
        \includegraphics[width=\textwidth]{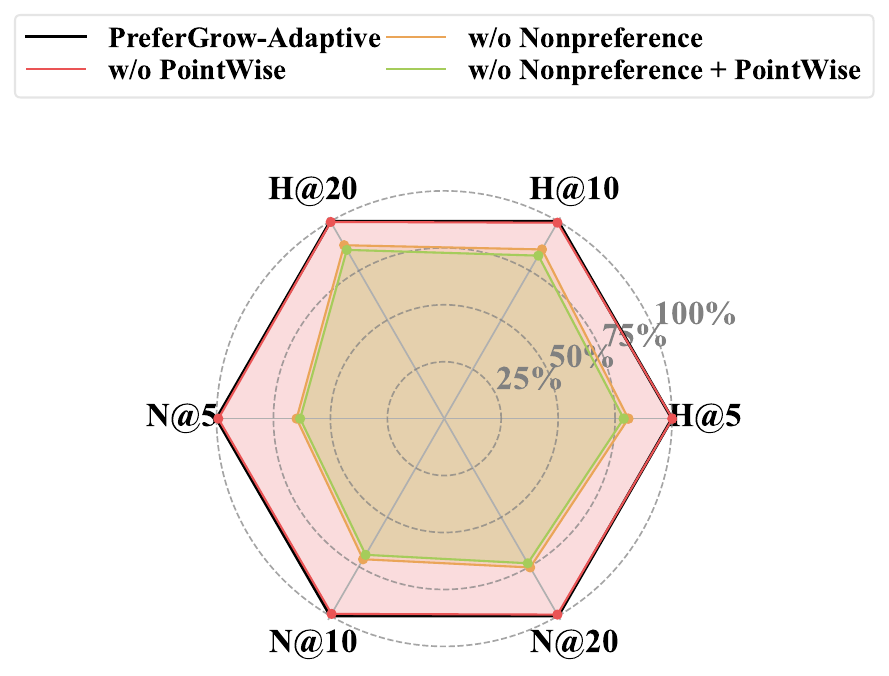}
        \caption{Abalation of Adaptive on Steam.}
        \label{fig:aba1}
    \end{subfigure}
    \begin{subfigure}{0.48\linewidth}
        \centering
        \includegraphics[width=\textwidth]{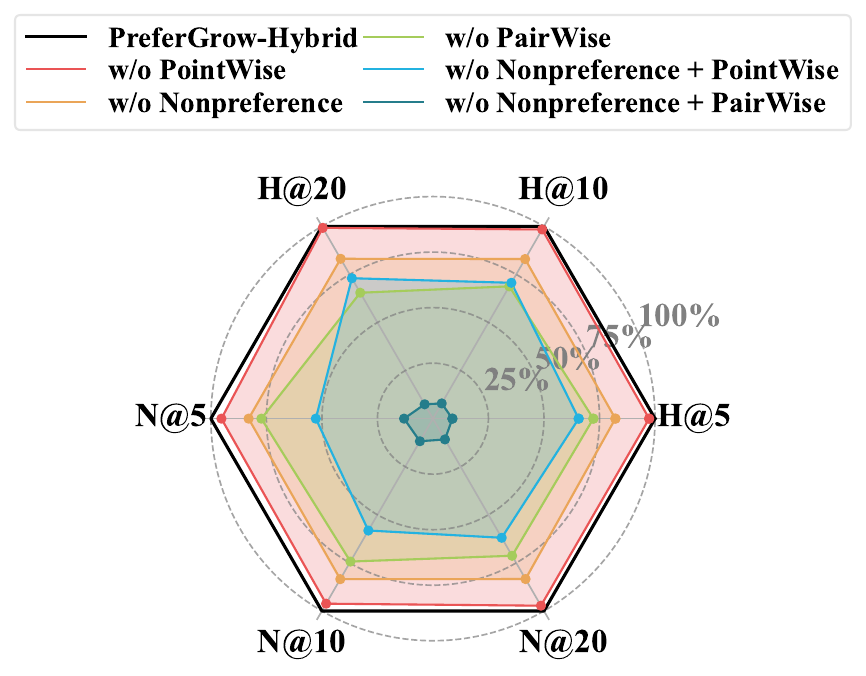}
        \caption{Ablation of Hybrid on Steam.}
        \label{fig:aba2}
    \end{subfigure}
    \caption{Ablation Study on Steam, showing the percentage of the relative effectiveness.}
    \label{fig:aba}
    \vspace{-15pt}
\end{figure*} 

\section{Experiment}\label{sec:exp}


We compare PreferGrow with a variety of baselines under the all-ranking evaluation protocol, including classical recommenders (SASRec \cite{SASRec}, Caser \cite{Caser}, GRURec \cite{GRURec}), item-level diffusion-based recommenders (DreamRec \cite{DreamRec}, PreferDiff \cite{PreferDiff}), and preference score-level diffusion-based recommenders (DiffRec \cite{DiffRec2}, DDSR \cite{DDSR}) across five benchmark datasets. 
Details of the datasets are provided in Appendix \ref{exp:data}, discussions of the baselines are presented in Appendix \ref{exp:baseline}, and the training and evaluation settings are described in Appendix \ref{exp:implement}.

\subsection{Overall Comparsion}

From the results in Table \ref{tab:overall}, we observe that PreferGrow consistently outperforms all baselines. 
We attribute the effectiveness of PreferGrow to three factors: 
1) its discrete diffusion process, which aligns with the discrete nature of the recommendation scenario; 
2) the preference fading noise perturbation mechanism, akin to negative sampling; 
and 3) preference ratios modeling without the simplex constraints, which better captures user preferences.
Additionally, the hybrid and adaptive fading matrix settings, inspired by the negative sampling strategy, also demonstrate the flexibility.

\subsection{Ablation Study}

As presented in Figure \ref{fig:aba}, we perform a thorough analysis and evaluation of each key component within PreferGrow to assess their individual significance. 
The ablation study is conducted using the following three variations: 
(1) \textit{w/o-PointWise}, PreferGrow without the general hard negative item $x_{-1}$;
(2) \textit{w/o-PairWise}, Point-Wise PreferGrow with $\vec{p}_T = \vec{e}_{-1}$; 
(3) \textit{w/o-Nonpreference}, PreferGrow without modeling the non-preference user $\phi$, resulting in no personalized enhancement during the backward generation, as well as the combination of the aforementioned ablations.
Overall, PairWise preference ratios modeling and non-preference user modeling are crucial for the effectiveness of PreferGrow, while PairWise preference ratios modeling adds an additional refinement.

\subsection{Hyper-parameter Analysis}

\hugq{
We further investigate the effect of key hyperparameters on the performance of PreferGrow (Figure~\ref{fig:param}). 
Overall, PreferGrow is fairly robust to the non-preference user proportion $p$, the hybrid coefficient $\lambda$, and the number of sampling steps, while exhibiting relatively higher sensitivity to the personalization strength $w$.
Importantly, $w$ is chosen at inference time, so no retraining is required. 
On the one hand, $w$ is locally stable within a reasonable range (see Figure~\ref{fig:param}); on the other hand, we show that the optimal $w$ is highly consistent across data splits (Appendix~\ref{app:param_ana}). 
Taken together, these properties ensure that tuning $w$ is practical and does not introduce significant overhead.

}

\begin{figure*}[t!]
    \vspace{-5pt}
    \begin{subfigure}{0.325\linewidth}
        \centering
        \includegraphics[width=\textwidth]{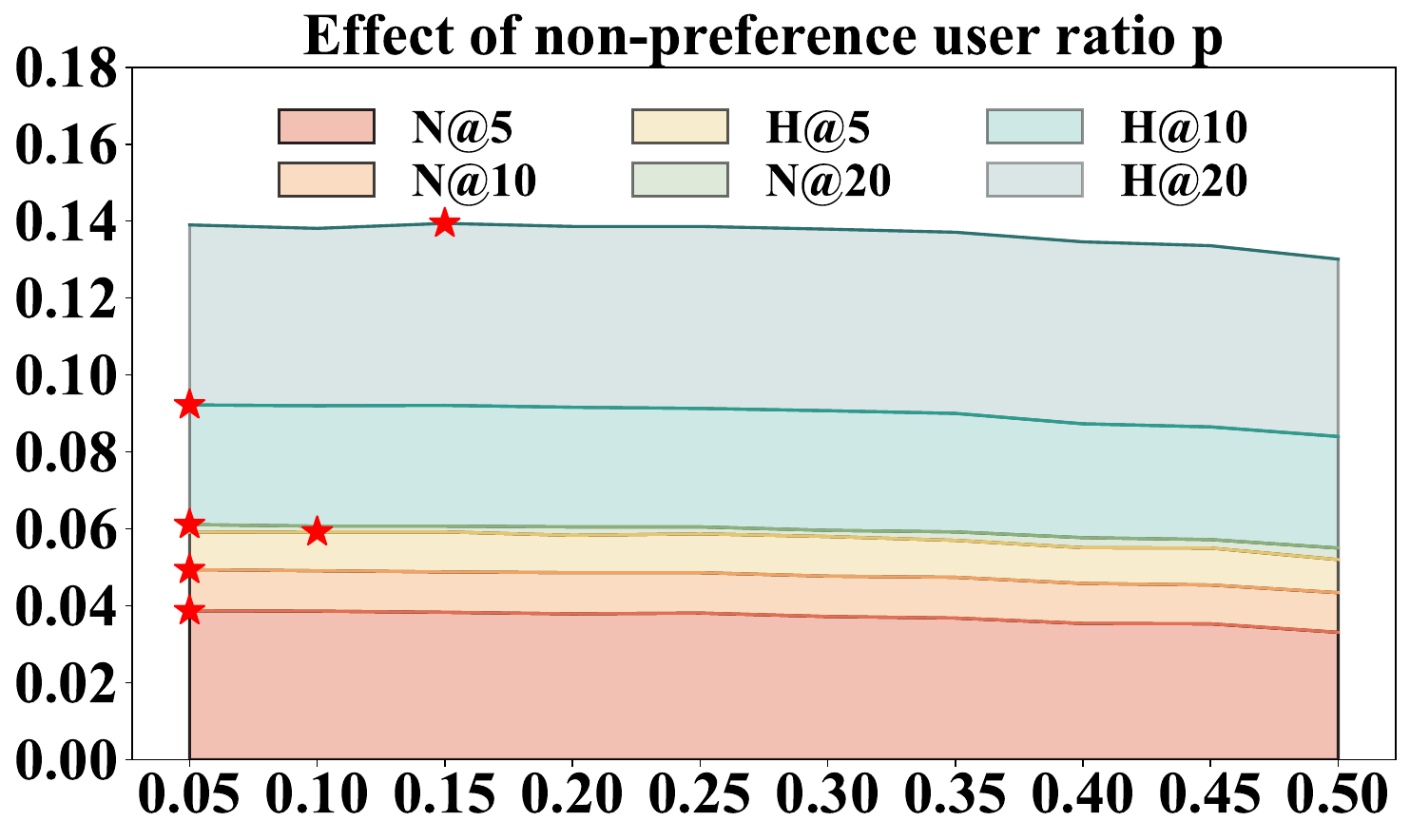}
        \caption{$p$ of Adaptive on Steam.}
        \label{fig:dim_item_embs}
    \end{subfigure}
    \begin{subfigure}{0.325\linewidth}
        \centering
        \includegraphics[width=\textwidth]{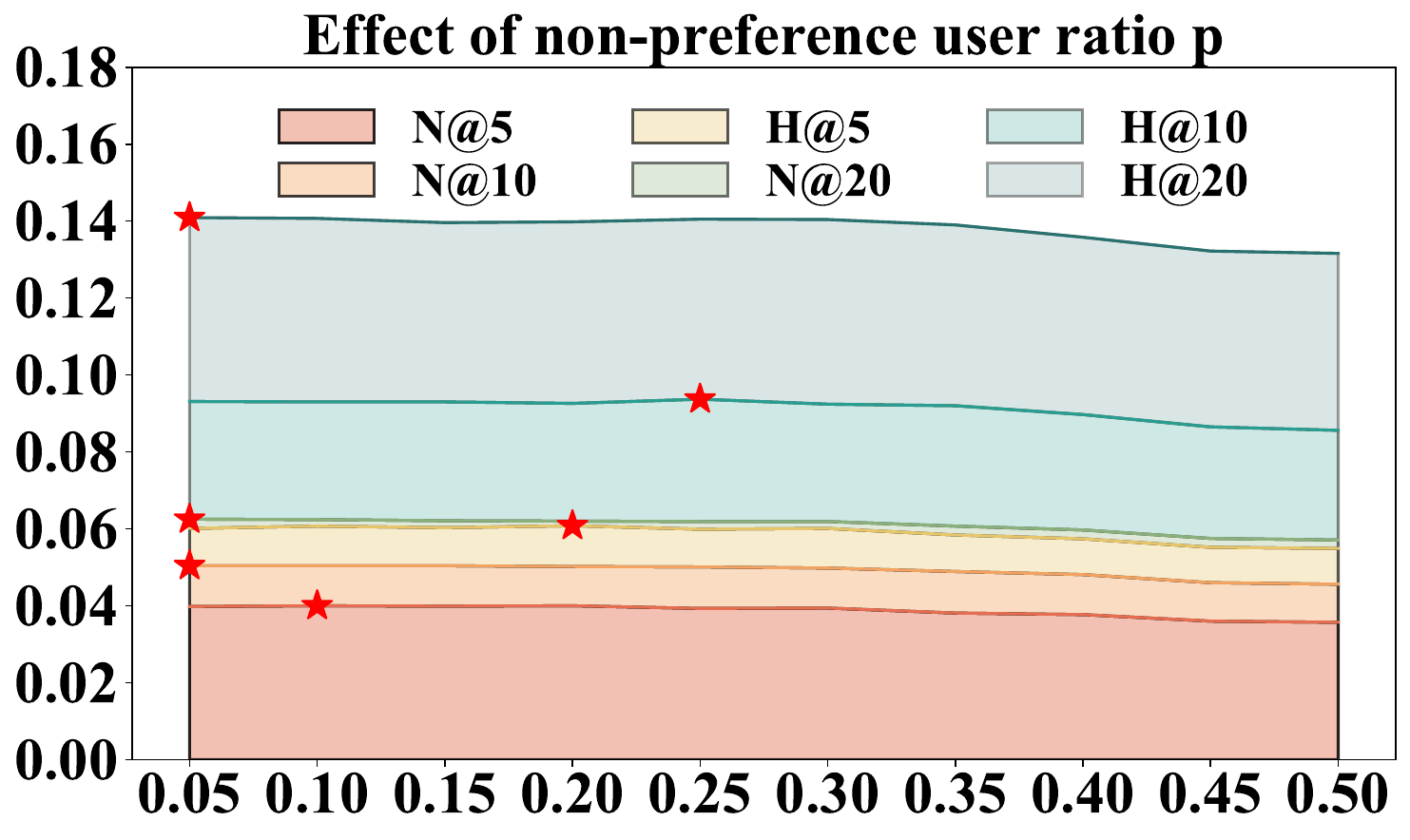}
        \caption{$p$ of Hybrid on Steam.}
        \label{fig:dim_ID_embs}
    \end{subfigure}
    \begin{subfigure}{0.325\linewidth}
        \centering
        \includegraphics[width=\textwidth]{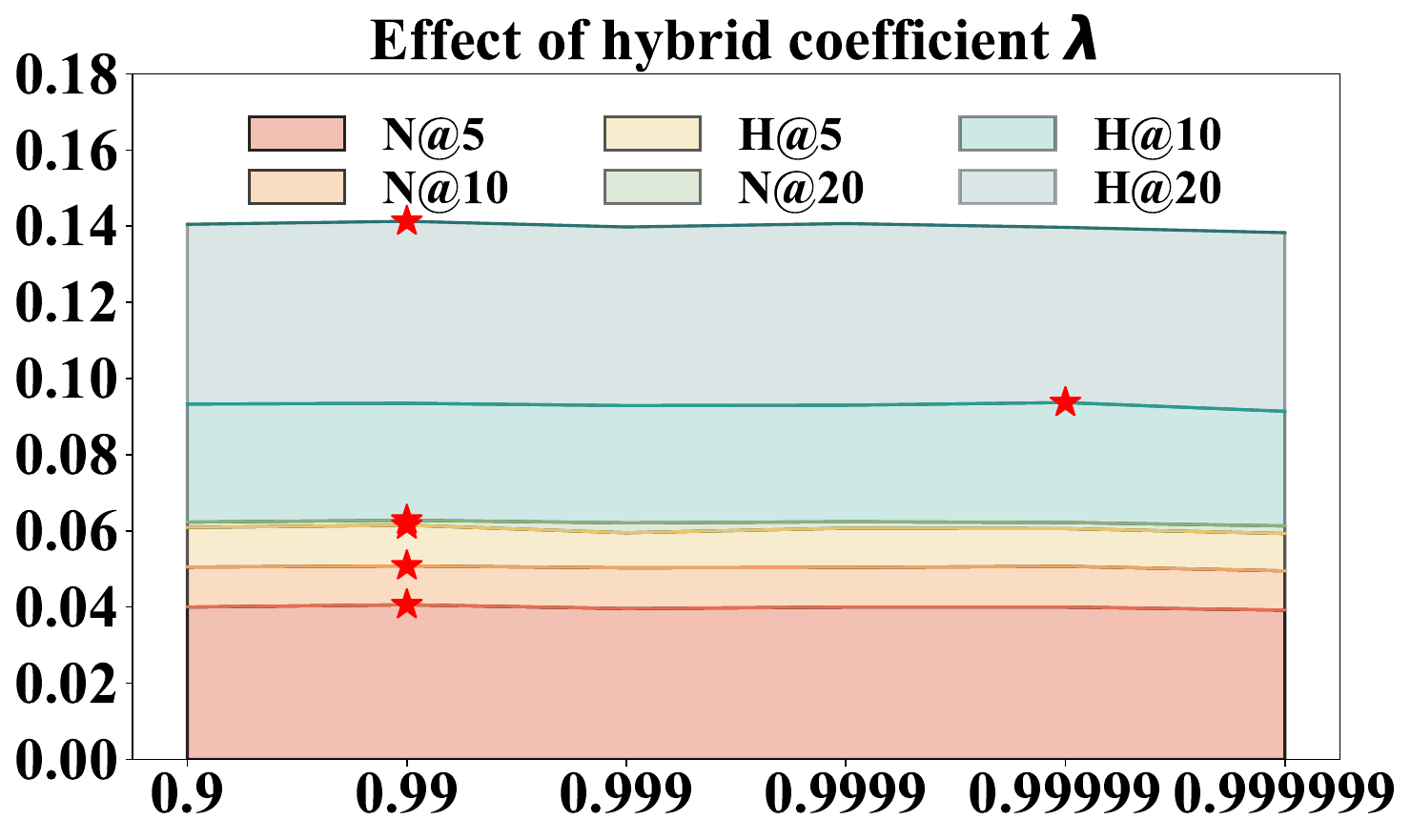}
        \caption{$\lambda$ of Hybrid on Steam.}
        \label{fig:thres}
    \end{subfigure}
    \begin{subfigure}{0.325\linewidth}
        \centering
        \includegraphics[width=\textwidth]{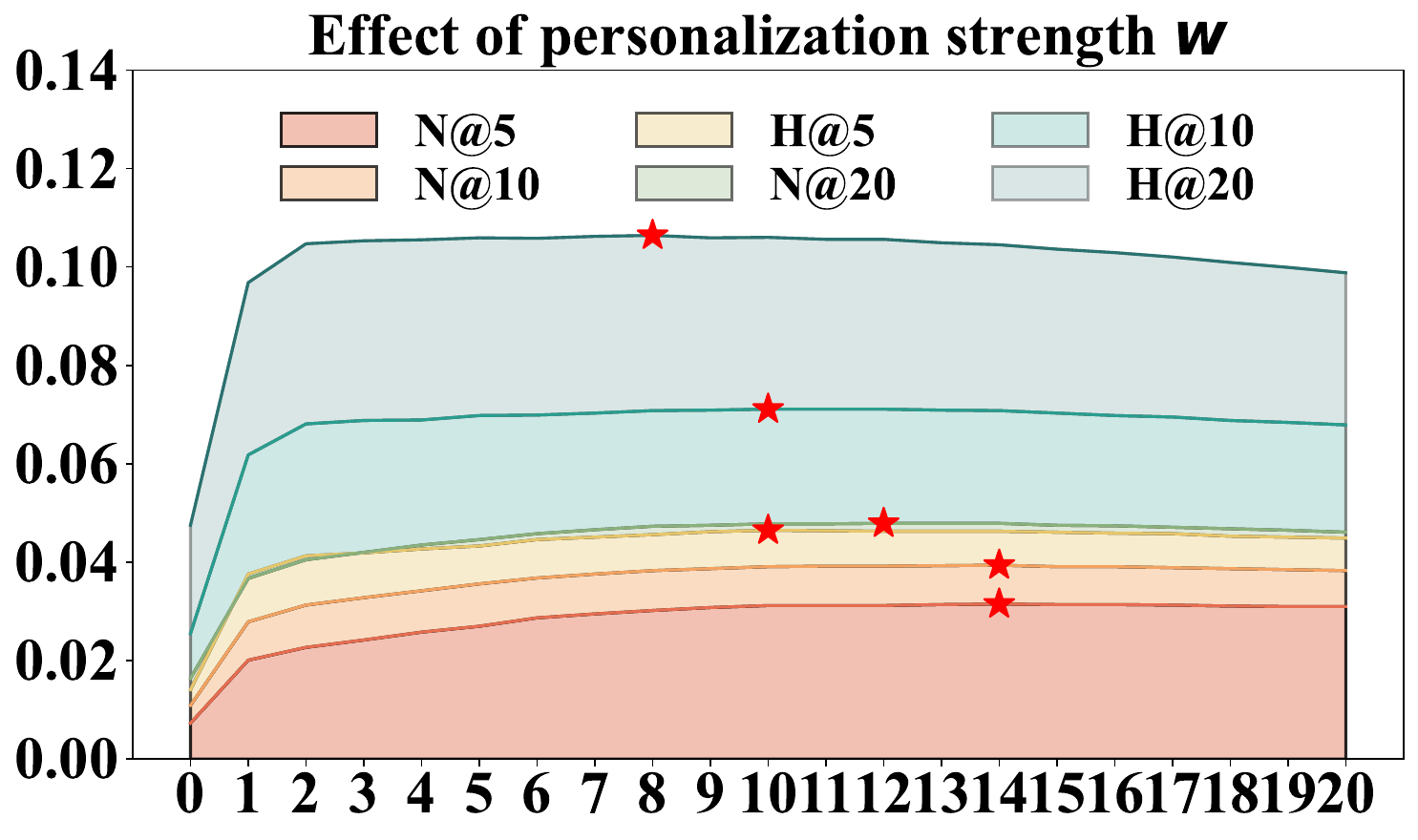}
        \caption{$w$ of Adaptive on Steam.}
        \label{fig:dim_item_embs}
    \end{subfigure}
    \begin{subfigure}{0.325\linewidth}
        \centering
        \includegraphics[width=\textwidth]{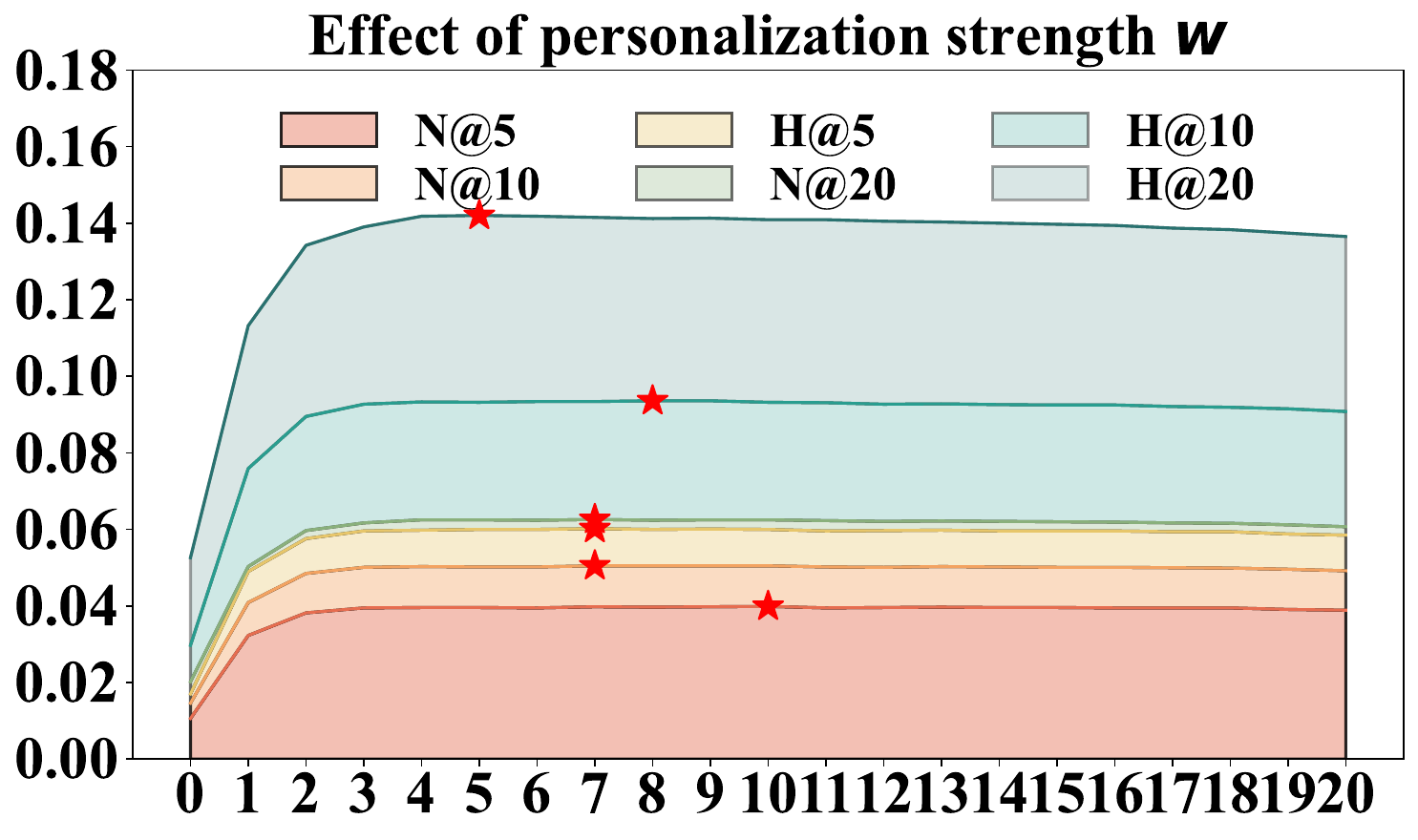}
        \caption{$w$ of Hybrid on Steam.}
        \label{fig:dim_ID_embs}
    \end{subfigure}
    \begin{subfigure}{0.325\linewidth}
        \centering
        \includegraphics[width=\textwidth]{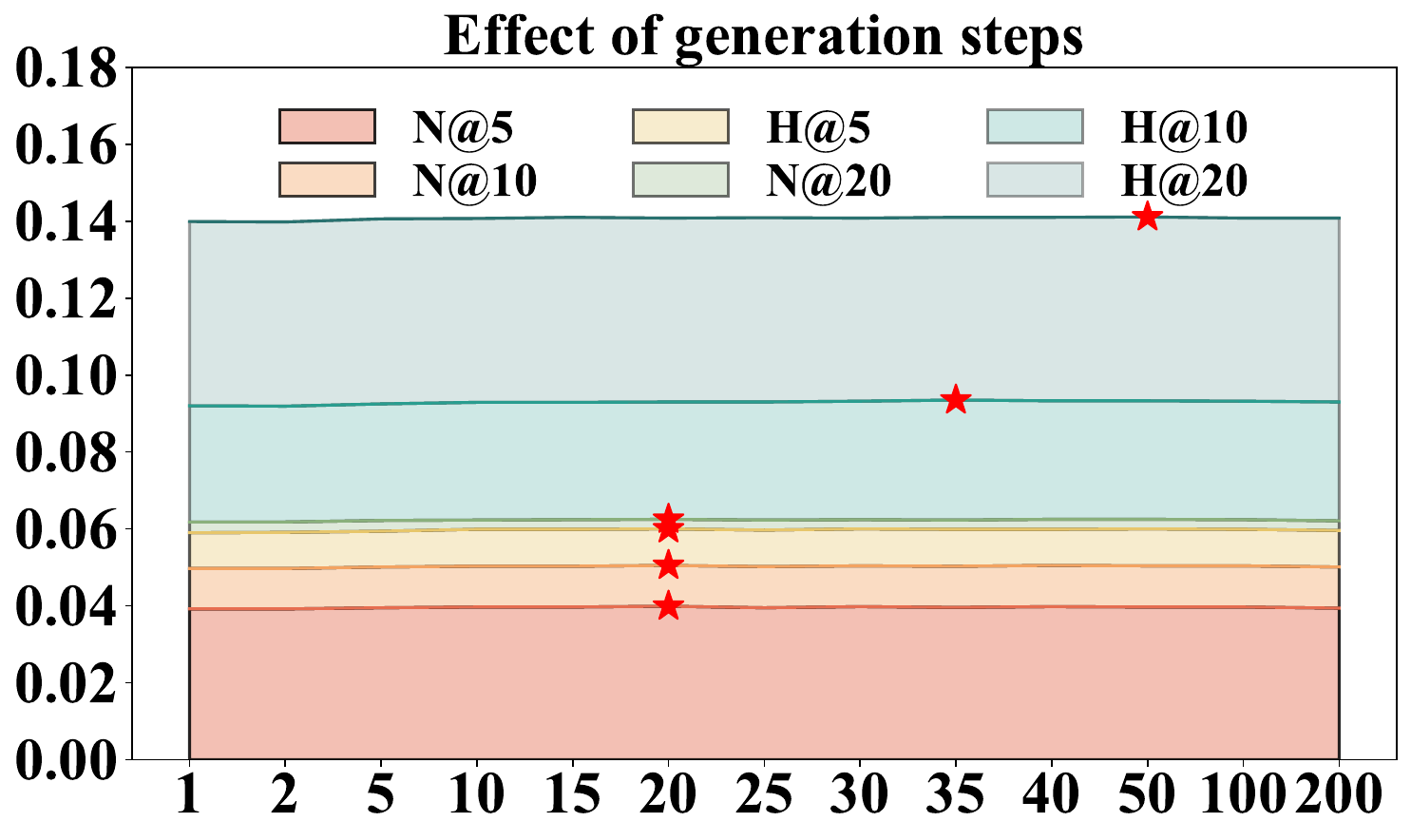}
        \caption{Steps of Hybrid on Steam.}
        \label{fig:thres}
    \end{subfigure}
    \caption{Hyperparameter analysis of PreferGrow on Steam. Red stars denote the best results.}
    \vspace{-15pt}
    \label{fig:param}
\end{figure*}    

%% file: chapters/6-else.tex
\section{Limitations} \label{limi}
\hugq{
While effective, PreferGrow still has several aspects that warrant further optimization.
\begin{itemize}[leftmargin=*]
\item \textbf{Higher modeling complexity.} 
As shown in Table~\ref{tab:complexity}, targeting \emph{preference ratios} incurs substantially higher complexity than prior diffusion-based recommenders. On the one hand, preference ratios are more expressive and thus raise the potential ceiling of the model; on the other hand, they also increase the learning difficulty under finite model capacity---for example, leading to longer training time than previous diffusion-based methods (Appendix~\ref{exp:time}). 
A principled remedy is to extend PreferGrow to a rank-$r$ fading matrix. 
By inducing a quantization structure, this formulation can markedly reduce the complexity of preference ratios while preserving their expressive power, thus achieving a better expressiveness–redundancy trade-off. 
We identify this as a primary avenue for future work.
In addition, we observe that the final 50\% of training yields only about a $5\%$ gain in $\mathrm{NDCG}@5$, suggesting substantial redundancy in the training process. 
A likely cause is uniform timestep sampling: different timesteps reflect different degrees of user preference fading and thus vary in the difficulty of modeling preference ratios. 
Incorporating this difficulty --- by beginning with easier timesteps (mild fading) and gradually advancing to harder ones --- may further accelerate convergence.
\item \textbf{High computational complexity.}
As shown in Table~\ref{tab:complexity}, PreferGrow has $\mathcal{O}(N)$ complexity for both loss computation and inference, where $N$ represents the size of the item corpus. 
While this is comparable to prior diffusion-based recommenders, it becomes impractical when scaling to extremely large item sets (e.g., billions of items), where even $\mathcal{O}(N)$ is prohibitive. 
A promising next step is to incorporate a quantization structure such as \emph{Semantic IDs (SIDs)} to reduce the cost from $\mathcal{O}(N)$ to $\mathcal{O}(mc)$. 
SIDs represent each item with $m$ codebooks of size $c$, enabling up to $c^{m}$ items while maintaining only $\mathcal{O}(mc)$ computation. 
Extending \textsc{PreferGrow} with rank-$r$ fading matrix to operate directly on SIDs is our future work. 

\end{itemize}

}


\section{Conclusion}

Building upon the previous diffusion-based recommenders and discrete diffusion models discussed in Appendix \ref{app:rela}, this paper introduce a new discrete diffusion-based recommender tailored for discrete and sparse recommendation scenarios, named PreferGrow.
In summary, PreferGrow distinguishes itself from existing diffusion-based recommenders in three key aspects: 
(1) \textbf{Discrete Diffusion:} It operates directly on the discrete item set, fully aligning with the discrete nature of recommendation. 
(2) \textbf{Preference Fading:} It fades user preferences by replacing the preferred item with others, akin to negative sampling, thus removing the need for prior noise assumptions. 
(3) \textbf{Preference Ratios:} It estimates preference ratios by modeling the logarithmic ratios of user–item interaction probabilities, circumventing the constraints of the probability simplex.
PreferGrow features a well-defined theoretical formulation and demonstrates superior performance in experiments.
We further discuss the broader impacts of PreferGrow in Appendix~\ref{app:non}.


\section*{Acknowledgments}
This research is supported by the National Natural Science Foundation of China (62572449).

%% file: chapters/5-relatedworks.tex
\section{Related Work}\label{app:rela}


\begin{table}[ht]
    \vspace{-10pt}
    \centering
    \caption{Comparison of Diffusion-based Recommenders}
    \label{tbl:table1}
    \renewcommand{\arraystretch}{1.5}
    \begin{tabular}{c|c|c|c|c}
        \Xhline{1.5pt}
        & \textit{Modeling} &  \textit{Forward} &  \textit{Negative} &  \textit{Representative}  \\
        & \textit{Target} & \textit{Perturbation} & \textit{Sampling} &  \textit{Work}\\
        \Xhline{1pt}
        \multirow{3}{*}{\textbf{Continuous}} 
            & Score Function & Gaussian on Embedding & \xmark & DreamRec \cite{DreamRec} \\
        \cline{2-5}
            & Score Function & Gaussian on Embedding & \cmark & PreferDiff \cite{PreferDiff} \\
        \cline{2-5}
            & Preference Scores & Gaussian on One-hot & \cmark & DiffRec \cite{DiffRec2} \\
        \Xhline{1pt}
        \multirow{3}{*}{\textbf{Discrete}} 
            & Preference Scores & Bernoulli on One-hot & \cmark & RecFusion \cite{Recformer} \\
         \cline{2-5}
            & Preference Scores & Categorical on One-hot & \cmark & DDSR \cite{DDSR}) \\
        \cline{2-5}
            & Preference Ratios & Fading on Items & \cmark & \textbf{PreferGrow} (Ours) \\
        \Xhline{1.5pt}
    \end{tabular}
    \label{tab:diff4rec}
\end{table}

\noindent\textbf{Diffusion-based recommenders} \cite{DreamRec, DiffuRec, DiffRec1, CDDRec, DiffuASR, DDRM, DiffCDR, DiffGT, CFDiff, CaDiRec, DMSR, DiffCL, DMCDR, MoDiCF, DiQDiff, DRGO, DiffuRec2, iDreamRec, PreferDiff, DiffRec2, RecFusion, LD4MRec, PDRec, EdgeRec, DifFaiRec, D3Rec, DDSR, DiffMM, SDiff} utilize forward perturbation in diffusion models to address data sparsity \cite{NCSN}, thereby better adapting to sparse recommendation scenarios. 
They typically consist of three core components: the modeling objective, the forward noise addition, and the corresponding backward generation process, which are summarized in Table \ref{tab:diff4rec}. 
Current research on diffusion-based recommenders can be classified into two primary approaches: one involves adding noise to dense item embeddings at the item level, and the other focuses on perturbing the one-hot interaction probability vector for all items, which is referred to as preference scores.
The first research line, pioneered by DreamRec \cite{DreamRec} and DiffuRec \cite{DiffuRec}, encodes user-preferred items as dense embeddings and adds Gaussian noise to these item embeddings. 
DreamRec \cite{DreamRec} models the score functions (the gradient of the log-likelihood of the perturbed distribution) without negative sampling, whereas DiffuRec \cite{DiffuRec} incorporates a recommendation loss that includes negative sampling. 
Building on them, PreferDiff \cite{PreferDiff} introduced an optimization objective derived from BPR loss \cite{BPR}, which integrates multiple negative samples into the generative modeling framework. 
However, the application of continuous Gaussian noise to positive preferred items, in contrast to the discrete nature of negative samples, creates an inherent mismatch, making it difficult to optimize both simultaneously during training, leading to a trade-off \cite{DiffuRec,DimeRec,PreferDiff}. \cite{DiffuRec2}
Additionally, other works have introduced more sophisticated module designs \cite{DiffRec1, CDDRec, DDRM, DiffGT, CFDiff, CaDiRec, DiffCL, DiQDiff, DiffuRec2, iDreamRec} or applied them to different recommendation tasks \cite{DiffuASR, DDRM, DiffCDR, DiffGT, CFDiff, CaDiRec, DMSR, DiffCL, DMCDR, MoDiCF, DRGO}. 
For instance, DimeRec \cite{DimeRec} incorporates multi-interest models, DiQDiff \cite{DiQDiff} introduces semantic vector quantization, and DiffuASR \cite{DiffuASR} applies item-level diffusion to sequential recommendation data augmentation.
The second research line \cite{DiffRec2, RecFusion, LD4MRec, PDRec, EdgeRec, DifFaiRec, D3Rec, DDSR, DiffMM, SDiff} involves converting user preference data into one-hot vectors, which are then mapped to preference scores within the probability simplex. 
DiffRec \cite{DiffRec2} add continuous Gaussian noise to the preference scores, and then learn to recover the clean preference scores from the perturbed ones. 
Consecutively, LD4MRec \cite{LD4MRec} refines the design for efficient multimedia recommendation, and D3Rec \cite{D3Rec} introduces targeted category preferences to control diversity during inference. 
However, the constraints of the probability simplex---non-negativity and normalization---pose significant challenges in accurately estimating the preference scores \cite{DDSE}.
To address these challenges while considering the probability simplex constraints, RecFusion \cite{Recformer} assumes a Bernoulli noise prior, completing the binomial diffusion process and subsequently modeling the parameters of the reverse binomial distribution to facilitate the reverse generation of preference scores. 
On the other hand, DDSR \cite{DDSR} adopts a categorical noise prior, as proposed in \cite{D3PM}, and directly recovers clean preference scores from the perturbed ones. Nonetheless, the constraints of the probability simplex may limit the effectiveness of preference score modeling.
Moreover, both diffusion-based recommenders rely on prior noise assumptions, such as Gaussian \cite{DreamRec} or Bernoulli noise \cite{RecFusion}, which may not be optimal for recommendation scenarios where user preference data is inherently discrete.

\noindent\textbf{Discrete Diffusion Models} \cite{D3PM, LDR, concretescore, SDDM, DDSE, RADD, DFM} have made substantial advances recently.
Initially, D3PM \cite{D3PM} proposed a discrete diffusion framework based on an arbitrary probability transition matrix, trained with the evidence lower bound of the log-likelihood. Subsequently, LDR \cite{LDR} extended this framework to a continuous-time setting using the Kolmogorov forward and backward equations. 
However, modeling score functions in such models presents challenges, as the gradient of the data distribution is undefined. 
To address this, CSM \cite{concretescore} introduced Concrete Score---discretization of score functions and the ratios of data distributions---as modeling objectives for discrete diffusion models.
Building upon these advances, SEDD \cite{DDSE} further bridges discrete diffusion models and ConcreteScore by introducing the score entropy loss. 
Expanding on these developments, we present PreferGrow, a matrix-based discrete diffusion framework which perturbing data by retaining or replacing items within a discrete corpus. 
The idempotent property of the replacement matrix (or fading matrix) is central to PreferGrow, and we demonstrate that it satisfies the Kolmogorov forward and backward equations as LDR \cite{LDR}, aligning it with prior works. 
Additionally, we introduce a design paradigm for the idempotent replacement matrix, which unifies previous approaches, including absorbing and uniform settings.

\section{Broader Impacts} \label{app:non}

Theoretically, PreferGrow introduces a well-defined discrete diffusion model building upon prior work.
While designed for recommendation, PreferGrow is also applicable to other discrete domains, such as molecular design in chemistry and protein structure prediction.
There are many potential societal consequences of our work, none of which we believe warrant specific attention at this time.

%% file: chapters/7-appendix.tex
\section{Proofs of Main Results} \label{app:proofs}

\begin{proof}[Proof of Theorem \ref{thm:preferfade}:] $\alpha_t: [0,T]\to[0,1]$ is a strictly decreasing function with $\alpha_0=1$ and $\alpha_T=0$ and the preference fading discrete diffusion process is denoted as $\mathbf{P}_{t|0}=\alpha_t\mathbf{I}+(1-\alpha_t)\mathbf{E},\forall t \in [0,T]$.
Then we can rewrite preference fading discrete diffusion process as for $\alpha_0=1$:
\begin{equation}
\label{eq:P_t0}
\mathbf{P}_{t|0}=\frac{\alpha_t}{\alpha_0}\mathbf{I}+(1-\frac{\alpha_t}{\alpha_0})\mathbf{E},\forall t \in [0,T].    
\end{equation}
For any well-defined $\mathbf{P}_{t\mid s,0}$ with $0 < s < t$, the following limiting distribution constraint must hold:
\begin{equation}
\label{eq:P_ts0}
p_{t|0}(x_t|x_0)=\sum\limits_{x_s\in\mathcal{X}}p_{t|s,0}(x_t|x_s,x_0)p_{s|0}(x_s|x_0),\forall x_t,x_0\in\mathcal{X}.    
\end{equation}
In matrix form, this constraint reduces to  
$\mathbf{P}_{t\mid0} = \mathbf{P}_{t\mid s,0}\,\mathbf{P}_{s\mid0}$ for all $0 < s < t$.  
Note that one possible particular solution of this equation is:
\begin{equation}
\label{eq:P_ts}
\mathbf{P}_{t|s,0}:=\mathbf{P}_{t|s}=\frac{\alpha_t}{\alpha_s}\mathbf{I}+(1-\frac{\alpha_t}{\alpha_s})\mathbf{E},\forall 0\le s\le t\le T.    
\end{equation}
We will show that this indeed satisfies the constraint under the condition that $\mathbf{E}$ is idempotent.
\begin{align*}
\mathbf{P}_{t|s}\mathbf{P}_{s|0} &= \left[\frac{\alpha_t}{\alpha_s}\mathbf{I}+(1-\frac{\alpha_t}{\alpha_s})\mathbf{E}\right]\cdot\left[\frac{\alpha_s}{\alpha_0}\mathbf{I}+(1-\frac{\alpha_s}{\alpha_0})\mathbf{E}\right] \tag{26}\\
&=\frac{\alpha_t}{\alpha_0}\mathbf{I} + \left(\frac{\alpha_t}{\alpha_s}-2\frac{\alpha_t}{\alpha_0}+\frac{\alpha_s}{\alpha_0}\right)\mathbf{E} + \left(1+\frac{\alpha_t}{\alpha_0}-\frac{\alpha_s}{\alpha_0}-\frac{\alpha_t}{\alpha_s}\right)\mathbf{E}^2 \\
&= \frac{\alpha_t}{\alpha_0}\mathbf{I} + \mathbf{E}^2 -\frac{\alpha_t}{\alpha_0}\mathbf{E}+ \left(\frac{\alpha_t}{\alpha_s}-\frac{\alpha_t}{\alpha_0}+\frac{\alpha_s}{\alpha_0}\right)(\mathbf{E} - \mathbf{E}^2)\\
&~~~~~\textcolor{gray}{\text{the fading matrix is idempotent} \Rightarrow \mathbf{E}^2 = \mathbf{E}} \\
&= \frac{\alpha_t}{\alpha_0}\mathbf{I} + (1-\frac{\alpha_t}{\alpha_0})\mathbf{E} \\
&= \mathbf{P}_{t|0}.
\end{align*}
Similarly, for all $0\le r\le s\le t\le T$, there holds  the Chapman-Kolmogorov equation:
\begin{align*}
\mathbf{P}_{t|s}\mathbf{P}_{s|r} &= \left[\frac{\alpha_t}{\alpha_s}\mathbf{I}+(1-\frac{\alpha_t}{\alpha_s})\mathbf{E}\right]\cdot\left[\frac{\alpha_s}{\alpha_r}\mathbf{I}+(1-\frac{\alpha_s}{\alpha_r})\mathbf{E}\right] \tag{27}\\
&= \frac{\alpha_t}{\alpha_r}\mathbf{I} + \mathbf{E}^2 -\frac{\alpha_t}{\alpha_r}\mathbf{E}+ \left(\frac{\alpha_t}{\alpha_s}-\frac{\alpha_t}{\alpha_r}+\frac{\alpha_s}{\alpha_r}\right)(\mathbf{E} - \mathbf{E}^2)\\
&~~~~~\textcolor{gray}{\text{the fading matrix is idempotent} \Rightarrow \mathbf{E}^2 = \mathbf{E}} \\
&= \frac{\alpha_t}{\alpha_r}\mathbf{I} + (1-\frac{\alpha_t}{\alpha_r})\mathbf{E} \\
& = \mathbf{P}_{t|r}.
\end{align*}
We further note that $\mathbf{P}_{t|s},\forall 0\le s\le t\le T$ is invertible, with inverse given by:
\begin{align*}
\label{eq:inv_P_ts}
\mathbf{P}_{t|s}^{-1}&=\frac{\alpha_s}{\alpha_t}\mathbf{I}+(1-\frac{\alpha_s}{\alpha_t})\mathbf{E},\forall 0\le s\le t\le T. \tag{28}\\ 
\mathbf{P}_{t|s}^{-1}\mathbf{P}_{t|s}&=\left[\frac{\alpha_s}{\alpha_t}\mathbf{I}+(1-\frac{\alpha_s}{\alpha_t})\mathbf{E}\right]\cdot\left[\frac{\alpha_t}{\alpha_s}\mathbf{I}+(1-\frac{\alpha_t}{\alpha_s})\mathbf{E}\right] \tag{29}\\  
&= \mathbf{I} + \left(\frac{\alpha_t}{\alpha_s}+\frac{\alpha_s}{\alpha_t}-2\right)\cdot(\mathbf{E}-\mathbf{E}^2) \\
&~~~~~\textcolor{gray}{\text{the fading matrix is idempotent} \Rightarrow \mathbf{E}^2 = \mathbf{E}} \\
& = \mathbf{I}.
\end{align*}
$\forall 0\le r\le s\le t\le T$, combined with limiting distribution constraint $\mathbf{P}_{t|r}=\mathbf{P}_{t|s,r}\mathbf{P}_{s|r}$, there are:
\begin{align*}
\label{eq:P_tsr}
\mathbf{P}_{t|s,r}&=\mathbf{P}_{t|r}\mathbf{P}_{s|r}^{-1} \\ 
&=\left[\frac{\alpha_t}{\alpha_r}\mathbf{I}+(1-\frac{\alpha_t}{\alpha_r})\mathbf{E}\right]\cdot\left[\frac{\alpha_r}{\alpha_s}\mathbf{I}+(1-\frac{\alpha_r}{\alpha_s})\mathbf{E}\right] \tag{30}\\  
&= \frac{\alpha_t}{\alpha_s}\mathbf{I} + \mathbf{E}^2 -\frac{\alpha_t}{\alpha_s}\mathbf{E}+ \left(\frac{\alpha_t}{\alpha_r}-\frac{\alpha_t}{\alpha_s}+\frac{\alpha_r}{\alpha_s}\right)(\mathbf{E} - \mathbf{E}^2)\\
&~~~~~\textcolor{gray}{\text{the fading matrix is idempotent} \Rightarrow \mathbf{E}^2 = \mathbf{E}} \\
&= \frac{\alpha_t}{\alpha_s}\mathbf{I} + (1-\frac{\alpha_t}{\alpha_s})\mathbf{E} \\
& = \mathbf{P}_{t|s}.
\end{align*}
$\mathbf{P}_{t|s,r}= \mathbf{P}_{t|s}$ indicates $p_{t|s,r}(x_t|x_s,x_r) = p_{t|s}(x_t|x_s), \forall 0\le r\le s\le t\le T$.
In summary, the preference fading discrete diffusion process is Markovian but not time-homogeneous, satisfies the Chapman–Kolmogorov equation, and is reversible.
\end{proof}

\begin{proof}[Proof of Proposition \ref{pro:fade}:] 
$\alpha_t: [0,T]\to[0,1]$ is a strictly decreasing function with $\alpha_0=1$ and $\alpha_T=0$. 
$\alpha_t$ is further defined as $e^{-\int_{0}^{t}\beta(\tau) \mathrm{d}\tau}$ with $\beta(\tau)>0$.
The preference fading discrete Markov diffusion process is denoted as $\mathbf{P}_{t|s}=\frac{\alpha_t}{\alpha_s}\mathbf{I}+(1-\frac{\alpha_t}{\alpha_s})\mathbf{E},\forall 0\le s\le t\le T$.
We first show that the preference fading discrete Markov diffusion process converges to a unified non-preference state $\vec{p}_T$ is well-defined.
The transition rate matrix $\mathbf{Q}_t$ is computed as follows:
\begin{align*}
\mathbf{Q}_{t} &= \lim\limits_{s\rightarrow t} \frac{\partial \mathbf{P}_{t|s}}{\partial t} \tag{31} \\
&= \lim\limits_{s\rightarrow t} \frac{\partial}{\partial t}
\left(\frac{\alpha_t}{\alpha_s}\mathbf{I}+(1-\frac{\alpha_t}{\alpha_s})\mathbf{E}\right) \\
&= \lim\limits_{s\rightarrow t} \frac{\partial}{\partial \alpha_t}
\left(\frac{\alpha_t}{\alpha_s}\mathbf{I}+(1-\frac{\alpha_t}{\alpha_s})\mathbf{E}\right) \cdot\frac{\partial \alpha_t}{\partial t}\\
&= \lim\limits_{s\rightarrow t} 
\left(\frac{1}{\alpha_s}\mathbf{I}-\frac{1}{\alpha_s}\mathbf{E}\right) \cdot (-\beta(t)) \cdot\alpha_t\\
\label{Q_t}
&= \beta(t) \cdot \left(\mathbf{E}-\mathbf{I}\right).
\end{align*}
The rate matrix $\mathbf{Q}_t$ characterizes the velocity of probability transitions at time $t$, encompassing both the direction and rate of transition. 
As shown in Equation \eqref{Q_t}, $\mathbf{Q}_t$ at any time $t \in [0, T]$ shares a consistent transition direction $\mathbf{I} - \mathbf{E}$, while the transition rate $\beta(t)$ varies over time.
This time-dependent rate results in a non-homogeneous preference fading process. 
However, the shared transition direction ensures that all diffusion paths converge to the same non-preference state, making the process well-defined.
Specifically, the stationary distribution $\vec{\pi}_t$ at time $t$ satisfies the equilibrium condition $\mathbf{Q}_t \vec{\pi}_t = \vec{0}$. 
Since different $\mathbf{Q}_t$ matrices differ only by a scalar factor $\beta(t)$, they yield the same stationary solution $\vec{\pi}$. 
Consequently, the Markov process converges to a common steady-state distribution $\vec{\pi}$, \ie the non-preference state $\vec{p}_T$:
\begin{align*}
\mathbf{Q}_t \vec{\pi} = \vec{0}\Rightarrow \left(\mathbf{I} - \mathbf{E}\right)\vec{p}_T &= \vec{0}. \tag{32}
\end{align*}
Given that $(\mathbf{I} - \mathbf{E})\mathbf{E} = \mathbf{0}$, each column of $\mathbf{E}$ satisfies the non-preference state equation $\left(\mathbf{I} - \mathbf{E}\right)\vec{p}_T= \vec{0}$. To ensure a unique solution $\vec{p}_T$, we therefore assume that all columns of $\mathbf{E}$ are identical, i.e.\ 
$\mathbf{E} \propto  \vec{p}_T \cdot\vec{1}^{\top}$.
Considering the idempotence constraint $\mathbf{E}^2=\mathbf{E}$, we then have:
\begin{align*}
\mathbf{E} &= \frac{\vec{p}_T \cdot\vec{1}^{\top}}{\vec{1}^{\top}\vec{p}_T}. \tag{33}\\
\mathbf{E}^2 &= \frac{\vec{p}_T \cdot\vec{1}^{\top}\cdot \vec{p}_T \cdot\vec{1}^{\top}}{ (\vec{1}^{\top}\vec{p}_T)^2} \tag{34}\\
&= \frac{\vec{p}_T \cdot(\vec{1}^{\top}\vec{p}_T) \cdot\vec{1}^{\top}}{ (\vec{1}^{\top}\vec{p}_T)^2}\\
&= \frac{\vec{p}_T \cdot \vec{1}^{\top}}{ \vec{1}^{\top}\vec{p}_T}\\
&= \mathbf{E}.
\end{align*}

\end{proof}

\begin{proof}[Proof of Proposition \ref{pro:rate}:] 
We have $\mathbf{P}_{t|s}=\frac{\alpha_t}{\alpha_s}\mathbf{I}+(1-\frac{\alpha_t}{\alpha_s})\mathbf{E}$ and $\mathbf{Q}_{t}=\beta(t) \cdot \left(\mathbf{E}-\mathbf{I}\right)$.
Then we compute $\frac{\partial \mathbf{P}_{t|s}}{\partial t}$ as follows:
\begin{align*}
\frac{\partial \mathbf{P}_{t|s}}{\partial t} &= \frac{\partial}{\partial t}
\left(\frac{\alpha_t}{\alpha_s}\mathbf{I}+(1-\frac{\alpha_t}{\alpha_s})\mathbf{E}\right) \tag{35}\\
&= \frac{\partial}{\partial \alpha_t}
\left(\frac{\alpha_t}{\alpha_s}\mathbf{I}+(1-\frac{\alpha_t}{\alpha_s})\mathbf{E}\right) \cdot\frac{\partial \alpha_t}{\partial t}\\
&=  \left(\frac{1}{\alpha_s}\mathbf{I}-\frac{1}{\alpha_s}\mathbf{E}\right) \cdot (-\beta(t)) \cdot\alpha_t\\
\label{P_st_partial_derivative}
&= \beta(t) \cdot \frac{\alpha_t}{\alpha_s}\cdot\left(\mathbf{E}-\mathbf{I}\right).
\end{align*}
There holds the Kolmogorov forward equation:
\begin{align*}
\mathbf{Q}_{t}\mathbf{P}_{t|s} &=
\beta(t) \cdot \left(\mathbf{E}-\mathbf{I}\right) \cdot \left[\frac{\alpha_t}{\alpha_s}\mathbf{I}+(1-\frac{\alpha_t}{\alpha_s})\mathbf{E}\right] \tag{36}\\
&= \beta(t) \cdot \left[-\frac{\alpha_t}{\alpha_s}\mathbf{I}+(2\frac{\alpha_t}{\alpha_s}-1)\mathbf{E}+(1-\frac{\alpha_t}{\alpha_s})\mathbf{E}^2\right]\\
&~~~~~\textcolor{gray}{\text{the fading matrix is idempotent} \Rightarrow \mathbf{E}^2 = \mathbf{E}} \\
&= \beta(t) \cdot \frac{\alpha_t}{\alpha_s}\cdot\left(\mathbf{E}-\mathbf{I}\right)\\
&= \frac{\partial \mathbf{P}_{t|s}}{\partial t}.
\end{align*}
\end{proof}

\begin{proof}[Proof of Proposition \ref{prop:se-bce-link}:] 
We begin by considering the expression for the loss function $\mathcal{L}_{\textit{SE}}$, which is given as:
\begin{equation*}
    \mathcal{L}_{\textit{SE}}(x_0,x_t,y) = e^{s_{\Theta}(x_t,t,u)_y}-s_{\Theta}(x_t,t,u)_y \cdot e^{r_t(x_0,x_t,y)} + e^{r_t(x_0,x_t,y)}(r_t(x_0,x_t,y)-1). \tag{37}
\end{equation*}

Then, we compute the gradient of this loss with respect to $s_\Theta$, which will help establish a link between this and the binary cross-entropy loss.

\begin{align*}
    \nabla_{s_\Theta} \mathcal{L}_{\textit{SE}} &= e^{s_{\Theta}(x_t,t,u)_y}- e^{r_t(x_0,x_t,y)} \tag{38}\\
    &=e^{s_{\Theta}(x_t,t,u)_y}(1+e^{r_t(x_0,x_t,y)})-(1+e^{s_{\Theta}(x_t,t,u)_y})e^{r_t(x_0,x_t,y)}\\
    &=(1 + e^{s_{\Theta}(x_t,t,u)_y})(1 + e^{r_t(x_0,x_t,y)})[\sigma(s_{\Theta}(x_t,t,u)_y)-\sigma(r_t(x_0,x_t,y))].
\end{align*}

Next, we consider the soft binary cross-entropy loss $\mathcal{L}_{\textit{sBCE}}$, which is defined as:

\begin{equation*}
    \mathcal{L}_{\textit{sBCE}}(x_0,x_t,y) = -\pi_{y \succ x_t \mid x_0} \log\sigma(s_{\Theta}(x_t,t,u)_y) - (1 - \pi_{y \succ x_t \mid x_0})\log(1 - \sigma(s_{\Theta}(x_t,t,u)_y)). \tag{39}
\end{equation*}
where the soft label $\pi_{y \succ x_t \mid x_0}=p(y \succ x_t \mid x_0) = \sigma(r_t(x_0, x_t, y))$ is the probability of the preference of $y$ over $x_t$, and $\sigma(\cdot)$ is the sigmoid function. 

Now, we compute the gradient of $\mathcal{L}_{\textit{sBCE}}$ with respect to $s_\Theta$:
\begin{align*}
    \nabla_{s_\Theta} \mathcal{L}_{\textit{sBCE}} &= -\pi_{y \succ x_t \mid x_0}(1-\sigma(s_{\Theta}(x_t,t,u)_y))+(1 - \pi_{y \succ x_t \mid x_0})\sigma(s_{\Theta}(x_t,t,u)_y) \tag{40}\\
    &=\sigma(s_{\Theta}(x_t,t,u)_y)-\pi_{y \succ x_t \mid x_0}\\
    &=\sigma(s_{\Theta}(x_t,t,u)_y)-\sigma(r_t(x_0,x_t,y)).
\end{align*}

We can now relate the gradients of both loss functions:
\begin{equation*}
    \nabla_{s_\Theta} \mathcal{L}_{\textit{SE}} = (1 + e^{s_{\Theta}(x_t,t,u)_y})(1 + e^{r_t(x_0,x_t,y)}) \nabla_{s_\Theta} \mathcal{L}_{\textit{sBCE}}. \tag{41}
\end{equation*}

\end{proof}

\allowdisplaybreaks
\begin{proof}[Proof of Theorem \ref{thm:prefergrow}:] 
The reverse preference growing process is denoted as $\Omega=\{\mathbf{P}_{s|T}\}_{s=T}^{0}$.
A collection $\mathcal{F}$ of subsets of $\Omega$ is called a $\sigma$-algebra on $\Omega$ if it satisfies: 1) $\Omega \in \mathcal{F}$.
2) If $A \in \mathcal{F}$ then its complement $A^{c} = \Omega\setminus A$ also belongs to $\mathcal{F}$.
3) If $\{A_{n}\}_{n=1}^{\infty}\subseteq \mathcal{F}$, then the countable union $\bigcup_{n=1}^{\infty}A_{n}\in \mathcal{F}$.
We first show that the preference growing process satisfies the Markov property.
$\mathcal{F}_t$ is a $\sigma$-algebra of the preference growing process $\Omega_t=\{\mathbf{P}_{s|T}\}_{s=T}^{t}$. 
$\forall A\in \mathcal{F}_t$ and $0\le s \le t \le T$: 
\begin{align*}
p_{s|\ge t}(x_s|x_t, A) &= \frac{p_{s,t|> t}(x_s, x_t| A)}{p_{t|> t}(x_t| A)} \cdot \frac{p(A)}{p(A)} \tag{42}\\
&= \frac{p_{s,\ge t}(x_s, x_t,A)}{p_{\ge t}(x_t, A)}\\
&= \frac{p_{>t|t,s}(A|x_t,x_s)}{p_{>t|t}(A|x_t)} \cdot \frac{p_{t|s}(x_t|x_s)p_s(x_s)}{p_t(x_t)}\\
&~~~~~\textcolor{gray}{\text{the preference fading process is Markovian} \Rightarrow \frac{p_{>t|t,s}(A|x_t,x_s)}{p_{>t|t}(A|x_t)}=1} \\
&= \frac{p_s(x_s)}{p_t(x_t)}p_{t|s}(x_t|x_s)\\
&~~~~~\textcolor{gray}{\text{the Bayes' theorem} \Rightarrow \frac{p_s(x_s)}{p_t(x_t)}p_{t|s}(x_t|x_s)=p_{s|t}(x_s|x_t)} \\
&= p_{s|t}(x_s|x_t).
\end{align*}
We rewrite the equation $ p_{s|t}(x_s|x_t)=\frac{p_s(x_s)}{p_t(x_t)}\cdot p_{t|s}(x_t | x_s)$ in matrix form $\forall 0 \le s \le t \le T$:
\begin{align*}
p_{s|t}(x_s|x_t)&=\frac{p_s(x_s)}{p_t(x_t)}\cdot p_{t|s}(x_t | x_s). \tag{43}\\
&~~~~~\textcolor{gray}{\vec{p}_s \cdot \left(\frac{1}{\vec{p}_t}\right)^\top \text{denotes the matrix with entries} \frac{p_s(x_s=y)}{p_t(x_t=x)}, \forall x, y \in \mathcal{X}} \\
\mathbf{P}_{s|t}&= \left[\vec{p}_s \cdot \left(\frac{1}{\vec{p}_t}\right)^\top\right] \odot \mathbf{P}_{t|s}^\top\\
&~~~~~\textcolor{gray}{\vec{p}_t=\mathbf{P}_{t|s} \cdot\vec{p}_s  \Rightarrow \vec{p}_s=\mathbf{P}_{s|t}^{-1} \cdot\vec{p}_t} \\
&= \mathbf{P}_{t|s}^{-1} \cdot \left[\vec{p}_t\cdot \left(\frac{1}{\vec{p}_t}\right)^\top\right] \odot \mathbf{P}_{t|s}^{\top}.
\end{align*}

The reverse-time preference growing process progresses from $s=T$ to $s=0$, and thus $\alpha_s$ is an increasing function, transitioning from $\alpha_T = 0$ to $\alpha_0 = 1$. 
We then compute the reverse transition rate matrix $\mathbf{R}_s$ as follows:
\begin{align*}
\mathbf{R}_{s} &= \lim\limits_{t\rightarrow s} \frac{\partial \mathbf{P}_{s|t}}{\partial s} \\
&= \lim\limits_{t\rightarrow s} \frac{\partial \mathbf{P}_{s|t}}{\partial \alpha_s}\cdot\frac{\partial \alpha_s}{\partial s} \tag{44}\\
&= \lim\limits_{t\rightarrow s} \frac{\partial}{\partial \alpha_s}\left(\mathbf{P}_{t|s}^{-1} \cdot \left[\vec{p}_t\cdot \left(\frac{1}{\vec{p}_t}\right)^\top\right] \odot \mathbf{P}_{t|s}^{\top}\right)\cdot\frac{\partial \alpha_s}{\partial s}\\
&= \lim\limits_{t\rightarrow s} \left(\frac{\partial \mathbf{P}_{t|s}^{-1}}{\partial \alpha_s}  \cdot\left[\vec{p}_t\cdot \left(\frac{1}{\vec{p}_t}\right)^\top\right]\odot \mathbf{P}_{t|s}^{\top}+\mathbf{P}_{t|s}^{-1} \cdot \left[\vec{p}_t\cdot \left(\frac{1}{\vec{p}_t}\right)^\top\right] \odot \frac{\partial \mathbf{P}_{t|s}^{\top}}{\partial \alpha_s}\right)\cdot\frac{\partial \alpha_s}{\partial s}\\
&~~~~~\textcolor{gray}{\frac{\partial \mathbf{P}_{t|s}^{-1}}{\partial \alpha_s}=\frac{\partial}{\partial \alpha_s}\left(\frac{\alpha_s}{\alpha_t}\mathbf{I}+(1-\frac{\alpha_s}{\alpha_t})\mathbf{E}\right)=\frac{1}{\alpha_t}(\mathbf{I}-\mathbf{E})}\\
&~~~~~\textcolor{gray}{\lim\limits_{t\rightarrow s} \mathbf{P}_{t|s}^{-1}=\lim\limits_{t\rightarrow s}\left(\frac{\alpha_s}{\alpha_t}\mathbf{I}+(1-\frac{\alpha_s}{\alpha_t})\mathbf{E}\right)=\mathbf{I}}\\
&~~~~~\textcolor{gray}{\frac{\partial \mathbf{P}_{t|s}^{\top}}{\partial \alpha_s}=\frac{\partial}{\partial \alpha_s}\left(\frac{\alpha_t}{\alpha_s}\mathbf{I}+(1-\frac{\alpha_t}{\alpha_s})\mathbf{E}^\top \right)=-\frac{\alpha_t}{\alpha_s^2}(\mathbf{I}-\mathbf{E}^\top)}\\
&~~~~~\textcolor{gray}{\lim\limits_{t\rightarrow s} \mathbf{P}_{t|s}^{\top}=\lim\limits_{t\rightarrow s}\left(\frac{\alpha_t}{\alpha_s}\mathbf{I}+(1-\frac{\alpha_t}{\alpha_s})\mathbf{E}^\top \right)=\mathbf{I}}\\
&= \lim\limits_{t\rightarrow s} \left(\frac{1}{\alpha_t}(\mathbf{I}-\mathbf{E})\cdot \left[\vec{p}_t\cdot \left(\frac{1}{\vec{p}_t}\right)^\top\right] \odot \mathbf{I}-\left[\vec{p}_t\cdot \left(\frac{1}{\vec{p}_t}\right)^\top\right] \odot \frac{\alpha_t}{\alpha_s^2}(\mathbf{I}-\mathbf{E}^\top)\right)\cdot\frac{\partial \alpha_s}{\partial s}\\
&~~~~~\textcolor{gray}{\mathbf{Q}_t = \beta(t) \cdot\left(\mathbf{E}-\mathbf{I}\right), \mathbf{Q}_t^\top = \beta(t) \cdot\left(\mathbf{E}^\top-\mathbf{I}\right)}\\
&~~~~~\textcolor{gray}{\text{Note that from time $t$ to time $s < t$, $\alpha_s$ is increasing.} \Rightarrow \frac{\partial \alpha_s}{\partial s}=\alpha_s\beta(s)>0} \\
&= \beta(s)\cdot(\mathbf{E}^\top-\mathbf{I})\odot \left[\vec{p}_s\cdot \left(\frac{1}{\vec{p}_s}\right)^\top\right] - \beta(s)\cdot(\mathbf{E}-\mathbf{I})\cdot \left[\vec{p}_s\cdot \left(\frac{1}{\vec{p}_s}\right)^\top\right] \odot \mathbf{I}\\
&= \mathbf{Q}_s^\top\odot \left[\vec{p}_s\cdot \left(\frac{1}{\vec{p}_s}\right)^\top\right] - \mathbf{Q}_s\cdot \left[\vec{p}_t\cdot \left(\frac{1}{\vec{p}_t}\right)^\top\right] \odot \mathbf{I}.
\end{align*}

Moreover, we compute $\frac{\partial \mathbf{P}_{t|s}}{\partial s}$ and $\frac{\partial \mathbf{P}_{s|t}}{\partial s}$ as follows:
\begin{align*}
\frac{\partial \mathbf{P}_{t|s}}{\partial s} &= \frac{\partial}{\partial \alpha_s}\left(\frac{\alpha_t}{\alpha_s}\mathbf{I} +(1-\frac{\alpha_t}{\alpha_s})\mathbf{E} \right) \cdot \frac{\partial \alpha_s}{\partial s}\tag{45}\\
&~~~~~\textcolor{gray}{\text{Note that from time $t$ to time $s < t$, $\alpha_s$ is increasing.} \Rightarrow \frac{\partial \alpha_s}{\partial s}=\alpha_s\beta(s)>0} \\
&= \left(-\frac{\alpha_t}{\alpha_s^2}\mathbf{I} +\frac{\alpha_t}{\alpha_s^2}\mathbf{E} \right) \cdot \alpha_s\beta(s)\\
&= \beta(s)\cdot \frac{\alpha_t}{\alpha_s}\cdot\left(\mathbf{E}-\mathbf{I} \right) \\
&~~~~~\textcolor{gray}{\left(\mathbf{E}-\mathbf{I}\right)^2=\mathbf{E}^2-2\mathbf{E}+\mathbf{I}=-\left(\mathbf{E}-\mathbf{I}\right), \left(\mathbf{E}-\mathbf{I}\right)\mathbf{E}=\mathbf{0}} \\
&= \beta(s)\cdot \left[-\frac{\alpha_t}{\alpha_s}\cdot\left(\mathbf{E}-\mathbf{I} \right)^2\right] + \beta(s)\cdot \left(\mathbf{E}-\mathbf{I}\right)\mathbf{E} \\
&= \left[ \frac{\alpha_t}{\alpha_s}\mathbf{I} +(1-\frac{\alpha_t}{\alpha_s})\mathbf{E}\right] \cdot \left[\beta(s)\cdot \left(\mathbf{E}-\mathbf{I} \right)\right] \\
&= \mathbf{P}_{t|s}\mathbf{Q}_{s}.
\end{align*}
 
\begin{align*}
\frac{\partial \mathbf{P}_{s|t}}{\partial s} &= \frac{\partial}{\partial s}\left(\mathbf{P}_{t|s}^{-1} \cdot \left[\vec{p}_t\cdot \left(\frac{1}{\vec{p}_t}\right)^\top\right] \odot \mathbf{P}_{t|s}^{\top}\right) \tag{46}\\
&= \frac{\partial}{\partial s}\left(\left[\vec{p}_s\cdot \left(\frac{1}{\vec{p}_t}\right)^\top\right] \odot \mathbf{P}_{t|s}^{\top}\right)\\
&= \left[\frac{\partial \vec{p}_s}{\partial s}\cdot \left(\frac{1}{\vec{p}_t}\right)^\top\right] \odot \mathbf{P}_{t|s}^{\top}
+\left[\vec{p}_s\cdot \left(\frac{1}{\vec{p}_t}\right)^\top\right] \odot \frac{\partial \mathbf{P}_{t|s}^{\top}}{\partial s}\\
&~~~~~\textcolor{gray}{\text{The reverse time $s$ begins with $s=T$.} \Rightarrow \frac{\partial \mathbf{P}_{s|0}}{\partial s}=-\mathbf{Q}_{s}\cdot\mathbf{P}_{s|0}} \\
&~~~~~\textcolor{gray}{\frac{\partial \vec{p}_s}{\partial s}=\frac{\partial \mathbf{P}_{s|0}\cdot\vec{p}_0}{\partial s}
= -\mathbf{Q}_{s}\cdot\mathbf{P}_{s|0}\cdot\vec{p}_0
= -\mathbf{Q}_{s}\cdot\vec{p}_s
}\\
&~~~~~\textcolor{gray}{\frac{\partial \mathbf{P}_{t|s}}{\partial s}=\mathbf{P}_{t|s}\cdot \mathbf{Q}_{s} \Rightarrow
\frac{\partial \mathbf{P}_{t|s}^\top}{\partial s}=(\mathbf{P}_{t|s}\cdot \mathbf{Q}_{s})^\top
= \mathbf{Q}_{s}^\top \cdot \mathbf{P}_{t|s}^\top
}\\
&= \left[\vec{p}_s\cdot \left(\frac{1}{\vec{p}_t}\right)^\top\right] \odot (\mathbf{Q}_{s}^\top \cdot \mathbf{P}_{t|s}^\top)-\left[\mathbf{Q}_{s}\cdot\vec{p}_s\cdot \left(\frac{1}{\vec{p}_t}\right)^\top\right] \odot \mathbf{P}_{t|s}^{\top}.
 \end{align*}

Finally, we prove that the preference growing process satisfies the Kolmogorov backward equation:
\begin{align*}
\mathbf{R}_{s}\mathbf{P}_{s|t} &=\left\{\mathbf{Q}_s^\top\odot \left[\vec{p}_s\cdot \left(\frac{1}{\vec{p}_s}\right)^\top\right] - \mathbf{Q}_s\left[\vec{p}_t\cdot \left(\frac{1}{\vec{p}_t}\right)^\top\right] \odot \mathbf{I}\right\}
 \cdot \left\{\left[\vec{p}_s\cdot \left(\frac{1}{\vec{p}_t}\right)^\top\right] \odot \mathbf{P}_{t|s}^{\top}\right\} \tag{47}\\
 &~~~~~\textcolor{gray}{\text{The idea behind this step is to prove that the elements at each position of the matrices are}}\\
 &~~~~~\textcolor{gray}{\text{identical, thereby establishing their equality}}.\\
  &~~~~~\textcolor{gray}{\text{The details are provided in Equation \eqref{point:0}, \eqref{point:1} and \eqref{point:2}}}.\\
 & = \left\{\left[\vec{p}_s\cdot \left(\frac{1}{\vec{p}_t}\right)^\top\right] \odot (\mathbf{Q}_{s}^\top \cdot \mathbf{P}_{t|s}^\top)\right\} - \left\{\left[\mathbf{Q}_{s}\cdot\vec{p}_s\cdot \left(\frac{1}{\vec{p}_t}\right)^\top\right] \odot \mathbf{P}_{t|s}^{\top}\right\} \\
 & = \frac{\partial \mathbf{P}_{s|t}}{\partial s}.
\end{align*}

\begin{align*}
&~~~~\left\{\mathbf{Q}_s^\top\odot \left[\vec{p}_s\cdot \left(\frac{1}{\vec{p}_s}\right)^\top\right]\right\}\cdot \left\{\left[\vec{p}_s\cdot \left(\frac{1}{\vec{p}_t}\right)^\top\right] \odot \mathbf{P}_{t|s}^{\top}\right\} (x,y) \tag{48} \label{point:0}\\
&=
\sum\limits_{z\in\mathcal{X}}q_s(z,x)\frac{p_s(x)}{p_s(z)}\cdot \frac{p_s(z)}{p_t(y)}p_{t|s}(y|z)
\\
&=
\frac{p_s(x)}{p_t(y)}\cdot \sum\limits_{z\in\mathcal{X}}q_s(z,x)p_{t|s}(y|z)
\\
&= \left\{\left[\vec{p}_s\cdot \left(\frac{1}{\vec{p}_t}\right)^\top\right] \odot (\mathbf{Q}_{s}^\top \cdot \mathbf{P}_{t|s}^\top)\right\}(x,y), \forall x,y\in\mathcal{X}.
\end{align*}

\begin{align*}
&~~~~\left\{\mathbf{Q}_s\cdot \left[\vec{p}_s\cdot \left(\frac{1}{\vec{p}_s}\right)^\top\right] \odot \mathbf{I}\right\}(x,y) \tag{49} \label{point:1} \\
&=
\sum\limits_{z\in\mathcal{X}}q_s(x,z)\frac{p_s(z)}{p_s(y)}\cdot \delta_{x}(y)
\\
&=
\delta_{x}(y)\cdot \sum\limits_{l\in\mathcal{X}}q_s(x,l)\frac{p_s(l)}{p_s(x)}.
\end{align*}

\begin{align*}
&~~~~\left\{\mathbf{Q}_s\cdot \left[\vec{p}_s\cdot \left(\frac{1}{\vec{p}_s}\right)^\top\right] \odot \mathbf{I}\right\}\cdot \left\{\left[\vec{p}_s\cdot \left(\frac{1}{\vec{p}_t}\right)^\top\right] \odot \mathbf{P}_{t|s}^{\top}\right\} (x,y) \tag{50} \label{point:2}\\
&=
\sum\limits_{z\in\mathcal{X}}\delta_{x}(z)\cdot \sum\limits_{l\in\mathcal{X}}q_s(x,l)\frac{p_s(l)}{p_s(x)}\cdot 
\frac{p_s(z)}{p_t(y)}p_{t|s}(y|z)
\\
&=
\sum\limits_{l\in\mathcal{X}}q_s(x,l)\frac{p_s(l)}{p_s(x)}\cdot 
\frac{p_s(x)}{p_t(y)}p_{t|s}(y|x)
\\
&= p_{t|s}(y|x)\cdot\sum\limits_{l\in\mathcal{X}}q_s(x,l)
\frac{p_s(l)}{p_t(y)} \\
&= \left\{\left[\mathbf{Q}_{s}\cdot\vec{p}_s\cdot \left(\frac{1}{\vec{p}_t}\right)^\top\right] \odot \mathbf{P}_{t|s}^{\top}\right\}(x,y), \forall x,y\in\mathcal{X}.
\end{align*}

\end{proof}

\section{Details of Different Fading Matrix Setting} \label{app:variants}

\subsection{Point-Wise Setting}
In the setting of masked discrete diffusion models \cite{RADD, MDD1, MDD2}, we can model point-wise preference ratios. 
Specifically, we introduce an auxiliary general hard negative item $x_{-1}$, which is represented as a learnable embedding. 
The unified non-preference state corresponds to the general hard negative item $x_{-1}$, that is, $\vec{p}_T = \vec{e}_{-1} \in \mathbb{R}^{N+1}$,
where $\vec{e}_{-1}$ denotes the one-hot vector associated with $x_{-1}$. 
In this case, the reference ratios $r_t(x_0, x_t \in \{x_0, x_{-1}\}, y \in \{x_0, x_{-1}\})$ capture only the relative preference between the positive item $x_0$ and the general hard negative $x_{-1}$. Thus, by using $x_{-1}$ as a common reference, we derive the point-wise preference ratios.

Let $\vec{p}_T = \vec{e}_{-1} \in \mathbb{R}^{N+1}$. Then, the rank-$1$ fading matrix $\mathbf{E}$ is defined as follows:
$$
\mathbf{E}=
\begin{pmatrix}
0 & \cdots & 0  & 0 \\
\vdots & \ddots & \vdots & \vdots \\
0 & \cdots & 0 & 0 \\
1 & \cdots & 1 & 1
\end{pmatrix}
$$
In this case, the reference ratios $r_t(x_0, x_t \in \{x_0, x_{-1}\}, y \in \{x_0, x_{-1}\})$ model only the ratios between the general hard negative $x_{-1}$ and real items $\mathcal{X}$. 
$$
r_t(x_0, x_t, y) =
\begin{cases}
    r_t(x_0, x_t, x_t)=0 & \text{if } x_t = y \\
    r_t(x_0, x_0, x_{-1})=\log\frac{1-\alpha_t}{\alpha_t} & \text{if } y=x_{-1} \text{ and }  x_t =x_0  \\
    r_t(x_0, x_{-1}, x_0)=\log\frac{\alpha_t}{1-\alpha_t} & \text{if } y=x_{0} \text{ and }  x_t =x_{-1}  \\
\end{cases}
$$

$$
\mathbf{Q}_t(x,y) = \beta(t)\cdot (\mathbf{E}-\mathbf{I})=
\begin{cases}
    -\beta(t) & \text{if } x = y \ne x_{-1} \\
    \beta(t) & \text{if } x=x_{-1} \text{ and }  x_t \ne x_{-1}  \\
    0 & \text{otherwise} \\
\end{cases}
$$

We eastimates the preference ratios $\log\frac{p_t(y|u)}{p_t(x_t|u)}$ with $s_{\Theta}(x_t, t, u)_y$:
$$
l_{SE}(x_0,x_t,y|u) = e^{s_{\Theta}(x_t, t, u)_y}-e^{r_t(x_0,x_t,y)}s_{\Theta}(x_t, t, u)_y+e^{r_t(x_0,x_t,y)}[r_t(x_0,x_t,y)-1].
$$

For one user preference data $(u,x_0)$, with faded item $x_t$, we compute $\mathcal{L}_{SE}(x_0,x_t,y|u)$:
\begin{align*}
\mathcal{L}_{SE} &=\sum\limits_{y\in \{x_0,x_{-1}\}}\mathbf{Q}_t(x_t,y) \cdot l_{SE}(x_0,x_t,y|u)\\
&=\mathbf{Q}_t(x_t,x_t) \cdot l_{SE}(x_0,x_t,x_t|u)+\sum\limits_{y\in \{x_0,x_{-1}\}\setminus\{x_t\}}\mathbf{Q}_t(x_t,y) \cdot l_{SE}(x_0,x_t,y|u)\\
&~~~~~\textcolor{gray}{s_{\Theta}(x_t, t, u)_{x_t}=r_t(x_0, x_t, x_t)=0\Rightarrow
l_{SE}(x_0,x_t,x_t|u)=0}\\
&= \sum\limits_{y\in \{x_0,x_{-1}\}\setminus\{x_t\}}\mathbf{Q}_t(x_t,y) \cdot l_{SE}(x_0,x_t,y|u)\\
&~~~~~\textcolor{gray}{\mathbf{Q}_t(x_t\ne x_{-1},y\ne x_t)=0\Rightarrow
\mathcal{L}_{SE}(x_0,x_0,y|u)=0\Rightarrow x_t=x_{-1}}\\
&= \mathbf{Q}_t(x_{-1},x_0) \cdot l_{SE}(x_0,x_{-1},x_0|u)\\
&= \beta(t) \cdot \left[e^{s_{\Theta}(x_t, t, u)_y}- \frac{\alpha_t}{1-\alpha_t}\cdot s_{\Theta}(x_t, t, u)_y +\frac{\alpha_t}{1-\alpha_t}[\log\frac{\alpha_t}{1-\alpha_t} -1]\right].
\end{align*}

\subsection{Pair-Wise Setting}

In the field of natural language processing, uniform discrete diffusion models are generally considered inferior to masked discrete diffusion models \cite{DDSE, RADD}. 
However, in the context of recommendation, uniform discrete diffusion models correspond to pair-wise preference ratios, which resemble random negative sampling, and are in fact superior to point-wise masked discrete diffusion models.
Specifically, we define the unified non-preference state as having equal probability over all items, \ie
$\vec{p}_T = \vec{1} \in \mathbb{R}^{N},$
where $\vec{1}$ denotes the all-ones vector. In this setting, the reference ratios $r_t(x_0, x_t \in \mathcal{X}, y \in \mathcal{X})$ capture the relative preferences among all item pairs.

Let $\vec{p}_T=\vec{1}\in \mathbb{R}^{N}$, we have fading matrix $\mathbf{E}$ as follows:
$$
\mathbf{E}=
\begin{pmatrix}
\frac{1}{N} & \cdots & \frac{1}{N}   \\
\vdots & \ddots & \vdots  \\
\frac{1}{N} & \cdots & \frac{1}{N}  
\end{pmatrix}
$$
In this case, the reference ratios $r_t(x_0, x_t \in \mathcal{X}, y \in \mathcal{X})$ model the ratios of all item pairs. 
$$
r_t(x_0, x_t, y) =
\begin{cases}
    r_t(x_0, x_t, x_t)=0 & \text{if } x_t = y \\
    r_t(x_0, x_0, y\ne x_{0})=-\log(1+N\cdot\frac{\alpha_t}{1-\alpha_t}) & \text{if } y\ne x_{0} \text{ and }  x_t =x_0  \\
    r_t(x_0, x_t\ne x_0, x_{0})=\log(1+N\cdot\frac{\alpha_t}{1-\alpha_t}) & \text{if } y= x_{0} \text{ and }  x_t\ne x_0  \\
    r_t(x_0, x_t\ne x_0, y\ne x_0)=0 & \text{otherwise }\\
\end{cases}
$$

$$
\mathbf{Q}_t(x,y) = \beta(t)\cdot (\mathbf{E}-\mathbf{I})=
\begin{cases}
    \beta(t)(\frac{1}{N}-1) & \text{if } x = y  \\
    \beta(t)\cdot\frac{1}{N} & \text{otherwise}  \\
\end{cases}
$$

We eastimates the preference ratios $\log\frac{p_t(y|u)}{p_t(x_t|u)}$ with $s_{\Theta}(x_t, t, u)_y$:
$$
l_{SE}(x_0,x_t,y|u) = e^{s_{\Theta}(x_t, t, u)_y}-e^{r_t(x_0,x_t,y)}s_{\Theta}(x_t, t, u)_y+e^{r_t(x_0,x_t,y)}[r_t(x_0,x_t,y)-1].
$$

For one user preference data $(u,x_0)$, with faded item $x_t$, we compute $\mathcal{L}_{SE}(x_0,x_t,y|u)$:
\begin{align*}
\mathcal{L}_{SE} &=\sum\limits_{y\in \mathcal{X}}\mathbf{Q}_t(x_t,y) \cdot l_{SE}(x_0,x_t,y|u)\\
&=\mathbf{Q}_t(x_t,x_t) \cdot l_{SE}(x_0,x_t,x_t|u)+\sum\limits_{y\in \mathcal{X}\setminus\{x_t\}}\mathbf{Q}_t(x_t,y) \cdot l_{SE}(x_0,x_t,y|u)\\
&~~~~~\textcolor{gray}{s_{\Theta}(x_t, t, u)_{x_t}=r_t(x_0, x_t, x_t)=0\Rightarrow
l_{SE}(x_0,x_t,x_t|u)=0}\\
&= \sum\limits_{y\in \mathcal{X}\setminus\{x_t\}}\mathbf{Q}_t(x_t,y) \cdot l_{SE}(x_0,x_t,y|u)\\
&~~~~~\textcolor{gray}{\mathbf{Q}_t(x_t,y\ne x_t)=\beta(t)\cdot\frac{1}{N}}\\
&= \beta(t)\cdot\frac{1}{N}\cdot\sum\limits_{y\in \mathcal{X}\setminus\{x_t\}}l_{SE}(x_0,x_{-1},x_0|u)\\
&~~~~~\textcolor{gray}{
e^{\overline{\mathbf{s}}_{\Theta}(x_t, t, u)} = \frac{1}{N}\cdot \sum\limits_{y\in \mathcal{X}\setminus\{x_t\}} e^{s_{\Theta}(x_t, t, u)_y} = \frac{1}{N}\cdot \left[\sum\limits_{y\in \mathcal{X}} e^{s_{\Theta}(x_t, t, u)_y} -1\right]}\\
&~~~~~\textcolor{gray}{
\overline{\mathbf{s}}_{\Theta}(x_t, t, u) = \frac{1}{N}\cdot \sum\limits_{y\in \mathcal{X}\setminus\{x_t\}} s_{\Theta}(x_t, t, u)_y = \frac{1}{N}\cdot\sum\limits_{y\in \mathcal{X}} s_{\Theta}(x_t, t, u)_y}\\
&\textcolor{blue}{\text{if } x_t=x_0\text{, denote} \Delta=N\cdot\frac{\alpha_t}{1-\alpha_t}:}\\
&
\textcolor{blue}{
= \beta(t)\cdot\left[
e^{\overline{\mathbf{s}}_{\Theta}(x_t, t, u)}-\frac{1}{1+\Delta}\cdot \overline{\mathbf{s}}_{\Theta}(x_t, t, u)
- \frac{1}{1+\Delta}(\log (1+\Delta)+1)
\right]}\\
&\textcolor{orange}{\text{if } x_t\ne x_0\text{, denote} \Delta=N\cdot\frac{\alpha_t}{1-\alpha_t}:}\\
&
\textcolor{orange}{
= \beta(t)\cdot\left[
e^{\overline{\mathbf{s}}_{\Theta}(x_t, t, u)}- \overline{\mathbf{s}}_{\Theta}(x_t, t, u)
-\Delta\cdot s_{\Theta}(x_t, t, u)_{x_0} +  (1+\Delta)(\log (1+\Delta)-1)
\right].}
\end{align*}

\subsection{Hybrid-Wise Setting}

Note that the point-wise and pair-wise settings are independent of each other. 
Therefore, we can define a hybrid non-preference state as $\vec{p}_T = \lambda(\vec{1} \in \mathbb{R}^{N}, 0) + (1 - \lambda) \vec{e}_{-1} \in \mathbb{R}^{n+1}$, simultaneously modeling point-wise and pair-wise preference ratios. 
Since $|\mathcal{X}|$ is typically large, the hybrid coefficient $\lambda$ should be designed as $1 - 10^{-n_{\lambda}}$, where $n_{\lambda} \in \mathbb{Z}^+$.

Let $\vec{p}_1=\lambda(\vec{1},0)+(1-\lambda)\vec{e}_{-1}\in \mathbb{R}^{n+1}$, we have fading matrix $\mathbf{E}$ as follows:
$$
\mathbf{E}=
\begin{pmatrix}
\frac{\lambda}{N} & \cdots & \frac{\lambda}{N}  & \frac{\lambda}{N} \\
\vdots & \ddots & \vdots & \vdots \\
\frac{\lambda}{N} & \cdots & \frac{\lambda}{N} & \frac{\lambda}{N} \\
1-\lambda & \cdots & 1-\lambda & 1-\lambda
\end{pmatrix}
$$
In this case, the reference ratios $r_t(x_0, x_t \in \mathcal{X}, y \in \mathcal{X})$ model the ratios of all item pairs. 
$$
r_t(x_0, x_t, y) =
\begin{cases}
    0 & \text{if } x_t = y \\
    \log\frac{\alpha_t+(1-\alpha_t)\frac{\lambda}{N}}{(1-\alpha_t)(1-\lambda)} & \text{else if } y= x_{0} \text{ and }  x_t =x_{-1}  \\
    \log\frac{\frac{\lambda}{N}}{(1-\lambda)} & \text{else if } y\ne x_{0} \text{ and }  x_t= x_{-1}  \\
    \log\frac{(1-
    \alpha_t)(1-\lambda)}{\alpha_t+(1-\alpha_t)\frac{\lambda}{N}} & \text{else if } y= x_{-1} \text{ and }  x_t =x_0  \\
    \log\frac{(1-
    \alpha_t)\frac{\lambda}{N}}{\alpha_t+(1-\alpha_t)\frac{\lambda}{N}} & \text{else if } y\ne x_{-1} \text{ and }  x_t= x_0  \\
    \log\frac{(1-
    \alpha_t)(1-\lambda)}{(1-\alpha_t)\frac{\lambda}{N}} & \text{else if } y= x_{-1} \text{ and }  x_t\ne \{ x_{-1},x_0\}  \\
    \log\frac{\alpha_t+(1-
    \alpha_t)\frac{\lambda}{N}}{(1-\alpha_t)\frac{\lambda}{N}} & \text{else if } y= x_{0} \text{ and }  x_t\ne \{ x_{-1},x_0\}  \\
    0 & \text{else if } y\ne \{ x_{-1},x_0\} \text{ and }  x_t\ne \{ x_{-1},x_0\}  \\
\end{cases}
$$

$$
\mathbf{Q}_t(x,y) = \beta(t)\cdot (\mathbf{E}-\mathbf{I})=
\begin{cases}
    \beta(t)(-\lambda) & \text{if } x = y= x_{-1}  \\
    \beta(t)(\frac{\lambda}{N}-1) & \text{if } x = y\ne x_{-1}  \\
    \beta(t)(1-\lambda) & \text{if } x \ne y \text{ and } x= x_{-1}  \\
    \beta(t)(\frac{\lambda}{N}) & \text{if } x \ne y \text{ and } x\ne x_{-1}  \\
\end{cases}
$$

We estimate the preference ratios $\log\frac{p_t(y|u)}{p_t(x_t|u)}$ with $s_{\Theta}(x_t, t, u)_y$:
$$
l_{SE}(x_0,x_t,y|u) = e^{s_{\Theta}(x_t, t, u)_y}-e^{r_t(x_0,x_t,y)}s_{\Theta}(x_t, t, u)_y+e^{r_t(x_0,x_t,y)}[r_t(x_0,x_t,y)-1].
$$

For one user preference data $(u,x_0)$, with faded item $x_t$, we compute $\mathcal{L}_{SE}(x_0,x_t,y|u)$:
\begin{align*}
\mathcal{L}_{SE} &=\sum\limits_{y\in \mathcal{X}}\mathbf{Q}_t(x_t,y) \cdot l_{SE}(x_0,x_t,y|u)\\
&=\mathbf{Q}_t(x_t,x_t) \cdot l_{SE}(x_0,x_t,x_t|u)+\sum\limits_{y\in \mathcal{X}\setminus\{x_t\}}\mathbf{Q}_t(x_t,y) \cdot l_{SE}(x_0,x_t,y|u)\\
&~~~~~\textcolor{gray}{s_{\Theta}(x_t, t, u)_{x_t}=r_t(x_0, x_t, x_t)=0\Rightarrow
l_{SE}(x_0,x_t,x_t|u)=0}\\
&= \sum\limits_{y\in \mathcal{X}\setminus\{x_t\}}\mathbf{Q}_t(x_t,y) \cdot l_{SE}(x_0,x_t,y|u)\\
&~~~~~\textcolor{gray}{\mathbf{Q}_t(x_t=x_{-1},y\ne x_t)=\beta(t)\cdot(1-\lambda),\mathbf{Q}_t(x_t\ne x_{-1},y\ne x_t)=\beta(t)\cdot(\frac{\lambda}{N})}\\
&= \mathbf{Q}_t(x_t,y\ne x_t)\cdot\sum\limits_{y\in \mathcal{X}\setminus\{x_t\}}l_{SE}(x_0,x_{-1},x_0|u)\\
&~~~~~\textcolor{gray}{
e^{\overline{\mathbf{s}}_{\Theta}(x_t, t, u)} = \frac{1}{N}\cdot \sum\limits_{y\in \mathcal{X}\setminus\{x_t\}} e^{s_{\Theta}(x_t, t, u)_y} = \frac{1}{N}\cdot \left[\sum\limits_{y\in \mathcal{X}} e^{s_{\Theta}(x_t, t, u)_y} -1\right]}\\
&~~~~~\textcolor{gray}{
\overline{\mathbf{s}}_{\Theta}(x_t, t, u) = \frac{1}{N}\cdot \sum\limits_{y\in \mathcal{X}\setminus\{x_t\}} s_{\Theta}(x_t, t, u)_y = \frac{1}{N}\cdot\sum\limits_{y\in \mathcal{X}} s_{\Theta}(x_t, t, u)_y}\\
&\textcolor{blue}{\text{if } x_t=x_{-1}\text{, denote } \Delta=N\cdot\frac{\alpha_t}{1-\alpha_t} \text{ and } \Lambda=\frac{\lambda\cdot\Delta}{\lambda\cdot\Delta+N}:}\\
&\textcolor{blue}{= \beta(t)\cdot\left[N(1-\lambda)\cdot e^{\overline{\mathbf{s}}_{\Theta}(x_t, t, u)}-\lambda\cdot \overline{\mathbf{s}}_{\Theta}(x_t, t, u)-\frac{\frac{1}{\Lambda}-\lambda}{N}\cdot s_{\Theta}(x_t, t, u)_{x_0}\right]}\\
&\textcolor{blue}{+ \beta(t)\cdot\left\{\lambda\cdot \left[1+\frac{1}{N}\left(\frac{1}{\Lambda}-1\right)\right]\cdot \left[\log\left(\frac{1}{N}\cdot\frac{\lambda}{1-\lambda}\right) -1
\right] - \frac{\log\Lambda}{\Lambda}\right\}}\\
&\textcolor{orange}{\text{if } x_t= x_0\text{, denote } \Delta=N\cdot\frac{\alpha_t}{1-\alpha_t} \text{ and } \Lambda=\frac{\lambda\cdot\Delta}{\lambda\cdot\Delta+N}:}\\
&\textcolor{orange}{= \beta(t)\cdot\left[\lambda\cdot e^{\overline{\mathbf{s}}_{\Theta}(x_t, t, u)}-\lambda\cdot\Lambda\cdot\overline{\mathbf{s}}_{\Theta}(x_t, t, u)-(1-\lambda-\frac{\lambda}{N})\cdot \Lambda\cdot s_{\Theta}(x_t, t,u)_{x_{-1}}\right]}\\
&\textcolor{orange}{+ \beta(t)\cdot\left\{(1-\frac{\lambda}{N})\cdot\Lambda\cdot(\log\Lambda -1)-(1-\lambda)\cdot\Lambda\cdot(\log(\frac{1}{N}\frac{\lambda}{1-\lambda}))\right\}}\\
&\textcolor{purple}{\text{if } x_t\ne \{x_0,x_{-1}\}\text{, denote }\Delta=N\cdot\frac{\alpha_t}{1-\alpha_t} \text{ and }\Lambda=\frac{\lambda\cdot\Delta}{\lambda\cdot\Delta+N}:}\\
&\textcolor{purple}{= \beta(t)\cdot\left[\lambda\cdot e^{\overline{\mathbf{s}}_{\Theta}(x_t, t, u)}-\lambda\cdot\overline{\mathbf{s}}_{\Theta}(x_t, t, u)\right]}\\
&\textcolor{purple}{- \beta(t)\cdot\left[(\frac{1}{N\cdot\Lambda}-\frac{\lambda}{N})\cdot \Lambda\cdot s_{\Theta}(x_t, t, u)_{x_{0}})+(1-\lambda-\frac{\lambda}{N})\cdot s_{\Theta}(x_t, t, u)_{x_{0}}\right]}\\
&\textcolor{purple}{+ \beta(t)\cdot\left\{\frac{1}{N\cdot\Lambda}\cdot(\log\Lambda+ \log\lambda -1)-(1-\lambda)\cdot(1+\log(\frac{1}{N}\frac{\lambda}{1-\lambda}))-(1-\frac{2}{N})\cdot\lambda\right\}}.
\end{align*}

\subsection{Adative Setting}

Furthermore, we can adaptively update the non-preference state $\vec{p}_T$ initialized with above settings.
Under adaptive setting, $\vec{p}_T = \vec{\mu} =\text{softmax}(\vec{\theta})$, where $\vec{\theta}$ are learnable parameters.

Let $\vec{p}_T=\vec{\mu}\in \mathbb{R}^{N}$ or $\mathbb{R}^{N+1}$, with $\sum\limits_{x\in\mathcal{X}}\mu_x=1$, we have fading matrix $\mathbf{E}$ as follows:
$$
\mathbf{E}=
\begin{pmatrix}
\mu_1 & \cdots & \mu_1   \\
\vdots & \ddots & \vdots  \\
\mu_{|\mathcal{X}|} & \cdots & \mu_{|\mathcal{X}|} 
\end{pmatrix}
$$
By arbitrarily specifying a distribution $\vec{\mu}$ satisfying $\sum\limits_{x\in\mathcal{X}} \mu_x = 1$, we can instantiate any desired non-preference state, which is physically associated with different negative sampling strategies.
In this case, the reference ratios $r_t(x_0, x_t \in \mathcal{X}, y \in \mathcal{X})$ are computed as follows:
$$
r_t(x_0, x_t, y) =
\begin{cases}
    r_t(x_0, x_t, x_t)=0 & \text{if } x_t = y \\
    r_t(x_0, x_0, y\ne x_{0})=-\log(\frac{\alpha_t+(1-\alpha_t)\cdot\mu_{x_t}}{(1-\alpha_t)\cdot\mu_y}) & \text{if } y\ne x_{0} \text{ and }  x_t =x_0  \\
    r_t(x_0, x_t\ne x_0, x_{0})=\log(\frac{\alpha_t+(1-\alpha_t)\cdot\mu_{x_t}}{(1-\alpha_t)\cdot\mu_y}) & \text{if } y= x_{0} \text{ and }  x_t\ne x_0  \\
    r_t(x_0, x_t\ne x_0, y\ne x_0)=\log\frac{\mu_y}{\mu_{x_t}} & \text{otherwise }\\
\end{cases}
$$

$$
\mathbf{Q}_t(x,y) = \beta(t)\cdot (\mathbf{E}-\mathbf{I})=
\begin{cases}
    \beta(t)(\frac{1}{N}-1) & \text{if } x = y  \\
    \beta(t)\cdot\frac{1}{N} & \text{otherwise}  \\
\end{cases}
$$

We eastimates the preference ratios $\log\frac{p_t(y|u)}{p_t(x_t|u)}$ with $s_{\Theta}(x_t, t, u)_y$:
$$
l_{SE}(x_0,x_t,y|u) = e^{s_{\Theta}(x_t, t, u)_y}-e^{r_t(x_0,x_t,y)}s_{\Theta}(x_t, t, u)_y+e^{r_t(x_0,x_t,y)}[r_t(x_0,x_t,y)-1].
$$

For one user preference data $(u,x_0)$, with faded item $x_t$, we compute $\mathcal{L}_{SE}(x_0,x_t,y|u)$:
\begin{align*}
\mathcal{L}_{SE} &=\sum\limits_{y\in \mathcal{X}}\mathbf{Q}_t(x_t,y) \cdot l_{SE}(x_0,x_t,y|u)\\
&=\mathbf{Q}_t(x_t,x_t) \cdot l_{SE}(x_0,x_t,x_t|u)+\sum\limits_{y\in \mathcal{X}\setminus\{x_t\}}\mathbf{Q}_t(x_t,y) \cdot l_{SE}(x_0,x_t,y|u)\\
&~~~~~\textcolor{gray}{s_{\Theta}(x_t, t, u)_{x_t}=r_t(x_0, x_t, x_t)=0\Rightarrow
l_{SE}(x_0,x_t,x_t|u)=0}\\
&= \sum\limits_{y\in \mathcal{X}\setminus\{x_t\}}\mathbf{Q}_t(x_t,y) \cdot l_{SE}(x_0,x_t,y|u)\\
&~~~~~\textcolor{gray}{\mathbf{Q}_t(x_t,y\ne x_t)=\beta(t)\cdot\mu_{x_t}}\\
&= \beta(t)\cdot\frac{1}{N}\cdot\sum\limits_{y\in \mathcal{X}\setminus\{x_t\}}l_{SE}(x_0,x_{-1},x_0|u)\\
&~~~~~\textcolor{gray}{
e^{\overline{\mathbf{s}}_{\Theta}(x_t, t, u)} = \frac{1}{N}\cdot \sum\limits_{y\in \mathcal{X}\setminus\{x_t\}} e^{s_{\Theta}(x_t, t, u)_y} = \frac{1}{N}\cdot \left[\sum\limits_{y\in \mathcal{X}} e^{s_{\Theta}(x_t, t, u)_y} -1\right]}\\
&~~~~~\textcolor{gray}{
\tilde{\mathbf{s}}_{\Theta}(x_t, t, u) = \sum\limits_{y\in \mathcal{X}\setminus\{x_t\}} \mu_y\cdot s_{\Theta}(x_t, t, u)_y = \sum\limits_{y\in \mathcal{X}} \mu_{y}\cdot s_{\Theta}(x_t, t, u)_y}\\
&~~~~~\textcolor{gray}{
\hat{\mathbf{\mu}} = \sum\limits_{y\in \mathcal{X}} \mu_{y}\cdot \log \mu_{y}}\\
&\textcolor{blue}{\text{if } x_t=x_0\text{, denote } \Sigma=\frac{\alpha_t}{1-\alpha_t}:}\\
&\textcolor{blue}{
= \beta(t)\cdot\left[
\mu_{x_t}\cdot e^{\overline{\mathbf{s}}_{\Theta}(x_t, t, u)}-\frac{1}{\Sigma+\mu_{x_t}}\cdot \tilde{\mathbf{s}}_{\Theta}(x_t, t, u)
\right]}\\
&\textcolor{blue}{
+ \beta(t)\cdot\frac{\mu_{x_t}}{\mu_{x_t}+\Sigma}\cdot\left[\hat{\mathbf{\mu}}+(\mu_{x_t}-1)\cdot(\log(\mu_{x_t}+\Sigma)-1)-\mu_{x_t}\cdot\log \mu_{x_t} \right]}\\
&\textcolor{orange}{\text{if } x_t\ne x_0\text{, denote } \Sigma=\frac{\alpha_t}{1-\alpha_t}:}\\
&
\textcolor{orange}{
= \beta(t)\cdot\left[
\mu_{x_t}\cdot e^{\overline{\mathbf{s}}_{\Theta}(x_t, t, u)}- \tilde{\mathbf{s}}_{\Theta}(x_t, t, u)
-\Sigma\cdot s_{\Theta}(x_t, t, u)_{x_0} \right]} \\
&\textcolor{orange}{
+ \beta(t)\cdot\left[\hat{\mathbf{\mu}}+\mu_{x_t}-(1+\Sigma)(1+\log(\mu_{x_t}))+(\Sigma+\mu_{x_0})\log (\Sigma+\mu_{x_0})-(\mu_{x_0})\cdot\log (\mu_{x_0}) \right].}
\end{align*}

\section{Rank-$r$ Solution of Fading Matrix $\mathbf{E}$}
\label{app:rankr}

\begin{theorem}[Nonnegative Idempotent Decomposition]
\label{thm:nonneg-idem-decomp}
Let $\mathbf{E}\in\mathbb{R}^{N\times N}$ be entrywise nonnegative and idempotent ($\mathbf{E}^2=\mathbf{E}$) with ${\rm rank}(\mathbf{E})=r$.
Then the following statements are equivalent:
\begin{enumerate}
\item[\textnormal{(A)}] $\mathbf{E}$ is nonnegative and idempotent with ${\rm rank}(\mathbf{E})=r$.
\item[\textnormal{(B)}] There exist $r$ nonzero, pairwise disjoint\footnote{Disjointness means the supports do not overlap:
$\supp(\vec p^{\,i})\cap \supp(\vec p^{\,j})=\varnothing$ for $i\neq j$.}
nonnegative vectors $\vec p^{\,1},\ldots,\vec p^{\,r}\in\mathbb{R}^N_{\ge 0}$ and the corresponding
support-indicator vectors $\vec s^{\,i}:=\mathbf{1}_{\supp(\vec p^{\,i})}\in\{0,1\}^N$
such that
\begin{equation}
\label{eq:E-decomp}
\mathbf{E}
=\sum_{i=1}^r \mathbf{E}_i,
\qquad
\mathbf{E}_i:=\frac{\vec p^{\,i}(\vec s^{\,i})^\top}{(\vec s^{\,i})^\top \vec p^{\,i}}, \tag{51}
\end{equation}
where $(\vec s^{\,i})^\top \vec p^{\,i}>0$ for each $i$ and
$\supp(\vec p^{\,i}):=\{\,k\in\{1,\dots,N\}: p^{\,i}_k>0\,\}$.
Equivalently, for any column index $j\in\{1,\dots,N\}$,
\[
\mathbf{E}_{:j}=
\begin{cases}
\vec p^{\,i}, & \text{if } j\in\supp(\vec p^{\,i}) \text{ for some } i\in\{1,\dots,r\},\\[2pt]
\vec 0, & \text{if } j\notin \bigcup_{i=1}^r \supp(\vec p^{\,i}).
\end{cases}
\]
\end{enumerate}

Moreover, writing the support indicator as $\vec 1$ restricted to $\supp(\vec p^{\,i})$, the decomposition
\eqref{eq:E-decomp} admits the compact form
\[
\mathbf{E}=\sum_{i=1}^r
\frac{\vec p^{\,i}\,\vec 1_{\supp(\vec p^{\,i})}^{\top}}{\vec 1_{\supp(\vec p^{\,i})}^{\top}\vec p^{\,i}},
\qquad
\vec p^{\,i}\ge 0,\quad
(\vec p^{\,i})^\top \vec p^{\,j}=0\ \ (i\neq j).
\]
\end{theorem}

\paragraph{Proof (Sufficiency $(\textnormal{B})\Rightarrow(\textnormal{A})$).}
Entrywise nonnegativity is immediate from $\vec p^{\,i}\ge 0$ and $\vec s^{\,i}\ge 0$.
Idempotence follows from disjoint supports:
for $i\neq j$ we have $(\vec s^{\,i})^\top \vec p^{\,j}=0$, hence
\[
\mathbf{E}_i\mathbf{E}_j
=\frac{\vec p^{\,i}(\vec s^{\,i})^\top}{(\vec s^{\,i})^\top \vec p^{\,i}}
\cdot
\frac{\vec p^{\,j}(\vec s^{\,j})^\top}{(\vec s^{\,j})^\top \vec p^{\,j}}
=
\frac{\vec p^{\,i}\,(\vec s^{\,i})^\top \vec p^{\,j}\,(\vec s^{\,j})^\top}{(\vec s^{\,i})^\top \vec p^{\,i}\,(\vec s^{\,j})^\top \vec p^{\,j}}
=\mathbf{0}.
\]
For $i=j$ we obtain
\[
\mathbf{E}_i^2
=\frac{\vec p^{\,i}(\vec s^{\,i})^\top\vec p^{\,i}(\vec s^{\,i})^\top}{\big((\vec s^{\,i})^\top \vec p^{\,i}\big)^2}
=\frac{\vec p^{\,i}(\vec s^{\,i})^\top}{(\vec s^{\,i})^\top \vec p^{\,i}}
=\mathbf{E}_i,
\]
so $\mathbf{E}^2=(\sum_i \mathbf{E}_i)^2=\sum_i \mathbf{E}_i^2=\sum_i \mathbf{E}_i=\mathbf{E}$.
To compute the rank, observe that the column space of $\mathbf{E}_i$ is $\operatorname{span}\{\vec p^{\,i}\}$ and
$\operatorname{supp}(\vec p^{\,i})$ are disjoint, hence the $r$ one-dimensional subspaces are independent; therefore
${\rm rank}(\mathbf{E})=\sum_i {\rm rank}(\mathbf{E}_i)=r$.

\paragraph{Proof (Necessity $(\textnormal{A})\Rightarrow(\textnormal{B})$).}
Since $\mathbf{E}$ is idempotent ($\mathbf{E}^2=\mathbf{E}$), its eigenvalues are $0$ or $1$, so the image
$\operatorname{Im}(\mathbf{E})$ has dimension $r$. Because $\mathbf{E}\ge 0$, every column of $\mathbf{E}$ is either
the zero vector or a nonnegative fixed point of $\mathbf{E}$ (indeed, for the $j$-th standard basis vector
$\vec e_j$, we have $\mathbf{E}(\mathbf{E}\vec e_j)=\mathbf{E}\vec e_j\ge 0$).
Group the nonzero columns of $\mathbf{E}$ by proportionality: put indices $j,k$ in the same class if the $j$-th and
$k$-th nonzero columns are positive multiples of each other. This is an equivalence relation, and it produces exactly
$r$ classes, say $S_1,\ldots,S_r$, because each class contributes one linearly independent direction in
$\operatorname{Im}(\mathbf{E})$. Pick one representative nonzero column from each class and denote it by
$\vec p^{\,i}\ge 0$ ($i=1,\dots,r$). By construction, for every $j\in S_i$ the $j$-th column of $\mathbf{E}$ equals
$\vec p^{\,i}$, and for $j\notin \bigcup_i S_i$ the column is zero. The classes are disjoint in support by definition.
Let $\vec s^{\,i}:=\mathbf{1}_{S_i}$ be the indicator of $S_i$. The matrix that places column $\vec p^{\,i}$ on $S_i$
and zeros elsewhere is
\[
\mathbf{E}_i \;=\; \frac{\vec p^{\,i}(\vec s^{\,i})^\top}{(\vec s^{\,i})^\top \vec p^{\,i}},
\quad\text{with }(\vec s^{\,i})^\top \vec p^{\,i}>0.
\]
Therefore
\[
\mathbf{E} \;=\; \sum_{i=1}^r \mathbf{E}_i,
\]
which is exactly the Equation \eqref{eq:E-decomp}. Disjoint supports immediately give
$(\vec p^{\,i})^\top \vec p^{\,j}=0$ for $i\neq j$. \hfill$\square$

\paragraph{Physical intuition:}
\begin{itemize}[leftmargin=*]
    \item \textbf{Clustering.} Partition the item set $\mathcal{X}$ ($|\mathcal{X}|=N$) into $r$ clusters $\mathcal{C}_1,\ldots,\mathcal{C}_r$, where items in the same cluster are more similar. 
    In practice, one may obtain $\{\mathcal{C}_i\}$ via vector-based methods (e.g., $k$-means on semantic embeddings) or via $r$-way classification on the user--item bipartite/weighted sequential graph. Since these constructions require side information beyond our core formulation, a full exploration is deferred to future work.
    \item \textbf{Cluster-wise target distributions.} Define for each cluster a nonnegative target vector
    \[
    \vec p_T^{\,i}(x) \;>\; 0 \ \text{if}\ x\in\mathcal{C}_i, 
    \qquad
    \vec p_T^{\,i}(x) \;=\; 0 \ \text{otherwise}, 
    \quad i=1,\ldots,r.
    \]
    Because $\{\vec p_T^{\,i}\}_{i=1}^r$ have disjoint supports, the resulting fading matrix
    $\mathbf{E}=\sum_{i=1}^r \dfrac{\vec p_T^{\,i}(\vec s^{\,i})^\top}{(\vec s^{\,i})^\top \vec p_T^{\,i}}$ is idempotent with ${\rm rank}(\mathbf{E})=r$ (Theorem~\ref{thm:nonneg-idem-decomp}), where $\vec s^{\,i}=\mathbf{1}_{\supp(\vec p_T^{\,i})}$.
    \item \textbf{Effect.} During fading, each preferred item is replaced only by items within the same cluster, which filters out cross-cluster candidates and better respects item heterogeneity. This rank-$r$ structure also improves efficiency. Using the decomposition, for any vector $\mathbf{x}\in\mathbb{R}^N$,
    \[
    \mathbf{E}\mathbf{x}
    =\sum_{i=1}^r \vec p_T^{\,i}\,
    \frac{(\vec s^{\,i})^\top \mathbf{x}}{(\vec s^{\,i})^\top \vec p_T^{\,i}},
    \]
    so computing $(\vec s^{\,i})^\top \mathbf{x}$ costs $\mathcal{O}(|\mathcal{C}_i|)$ and the combination step costs $\mathcal{O}(rN)$, reducing the overall complexity from the naive $\mathcal{O}(N^2)$ to $\mathcal{O}(rN)$.
    \item \textbf{Takeaway.} PreferGrow equipped with a \emph{rank-$r$} fading matrix --- i.e., cluster-wise replacement --- combines theoretical elegance with practical realism: it confines replacements within semantically coherent groups while offering a scalable $\mathcal{O}(rN)$ implementation. 
\end{itemize}

%% file: chapters/8-exp_details.tex
\section{Algorithms for Training and Inference} \label{app:algorithms}

\begin{minipage}{\textwidth}
\begin{algorithm}[H]  
  \caption{ Training Algorithm of PreferGrow}  
  \label{alg:train}  
  \begin{algorithmic}[ht]
    \Require  
        user preference data $\mathcal{D}=\{ \left(u, x_0\right)\}$, non-preference user ratio $p$, retention probability $\alpha_t=e^{-\int_0^{t}\beta(\tau)\mathrm{d}\tau}$ with ${\int_0^{t}\beta(\tau)\mathrm{d}\tau}=(\beta_{\min})^{1-t}(\beta_{\max})^t$ or $\int_0^{t}\beta(\tau)\mathrm{d}\tau=\log(1-(1-\beta_{\mathrm{scale}} \cdot t))$,
        and the non-preference state $\vec{p}_T$ for preference fading.
    \Ensure  
        estimated \textit{Preference Ratios} $\mathbf{s}_{\Theta}(x_t, t, u)$ and the non-preference user $\phi$.
    \Repeat  
        \State $(u, x_0) \sim \mathcal{D}$  
        \Comment{preference user-item pair.}
        \State $t\sim \text{Uniform}(\{1,\dots,T\})$
        \Comment{sampling timestep $t$ uniformly.}
        \State $u=\phi$ with probability $p$ \Comment{non-preference user modeling.}
        \State$x_t \sim p_{t|0}(\cdot|x_0)=\alpha_t\vec{e}_{x_0}+(1-\alpha_t)\vec{p}_{T}, x_t\in\mathcal{X}$
        \Comment{retain or replace for preference fading.}
        \State Compute the score entropy loss $\mathcal{L}_{SE}=\sum_{y\in\mathcal{X}}\mathbf{Q}_t(x_t,y)\cdot l_{SE}(x_0,x_t,y|u)$ as Appendix \ref{app:variants}.
        \State Take gradient descent step on
        $\nabla_{\theta} \mathcal{L}_{SE}$ and update parameters.
    \Until{converged}
  \end{algorithmic} 
\end{algorithm}
\end{minipage}
\begin{minipage}{\textwidth}
\begin{algorithm}[H]  
  \caption{ Inference Algorithm of PreferGrow}  
  \label{alg:inference}  
  \begin{algorithmic}[ht]
    \Require  
        user condition $u$, sampling timesteps $S_\tau=\{\tau_i\}_{i=0}^{S}$ with $\tau_S=T$ and $\tau_0=0$, personalization strength $w$, estimated \textit{Preference Ratios} $\mathbf{s}_{\theta}(x_t, t, u)$ and the non-preference user $\phi$.
    \Ensure
        grown preference scores $p(x_0|u),x_0\in\mathcal{X}$ of user $u$.
        \State $x_T=x_{\tau_S}\sim \vec{p}_T, x_T\in\mathcal{X}$ 
        \Comment{the non-preference state.}
        \For {$s=S$ to 1}
            \State $\hat{\mathbf{s}}_{\Theta}(x_{\tau_s}, \tau_s, u) =(1+w)\mathbf{s}_{\Theta}(x_{\tau_s}, \tau_s, u) - w\cdot\mathbf{s}_{\Theta}(x_{\tau_s}, \tau_s, \phi)$ 
            \Comment{personalization enhancement.} 
            \State $p_{{\tau_{s-1}}|{\tau_s}}(x_{\tau_{s-1}}|x_{\tau_{s}},u) = {p}_{{\tau_{s}}|{\tau_{s-1}}}(x_{\tau_{s}}|x_{\tau_{s-1}})\cdot \sum_{z\in\mathcal{X}}{p}_{{\tau_{s}}|{\tau_{s-1}}}^{-1}(x_{\tau_{s-1}}|z) \cdot e^{\hat{\mathbf{s}}_{\Theta}(x_{\tau_{s}},{\tau_{s}},u)_z}$ \\
            \Comment{${p}_{{\tau_{s}}|{\tau_{s-1}}}(x_{\tau_{s}}|x_{\tau_{s-1}})=\frac{\alpha_{\tau_s}}{\alpha_{\tau_{s-1}}}\delta_{x_{\tau_{s}}}(x_{\tau_{s-1}})+(1-\frac{\alpha_{\tau_s}}{\alpha_{\tau_{s-1}}})\cdot\vec{p}_T(x_{\tau_{s})}$ as Equation \eqref{Pst}}. \\
            \Comment{${p}^{-1}_{{\tau_{s}}|{\tau_{s-1}}}(x_{\tau_{s-1}}|z)=\frac{\alpha_{\tau_{s-1}}}{\alpha_{\tau_{s}}}\delta_{x_{\tau_{s-1}}}(z)+(1-\frac{\alpha_{\tau_{s-1}}}{\alpha_{\tau_{s}}})\cdot\vec{p}_T(x_{\tau_{s-1})}$ as Equation \eqref{Inv_Pts}}. \\
            \Comment{$\sum_{z\in\mathcal{X}}{p}_{{\tau_{s}}|{\tau_{s-1}}}^{-1}(x_{\tau_{s-1}}|z) \cdot e^{\hat{\mathbf{s}}_{\Theta}(x_{\tau_{s}},{\tau_{s}},u)_z}=\frac{\alpha_{\tau_{s-1}}}{\alpha_{\tau_{s}}}e^{\hat{\mathbf{s}}_{\Theta}(x_{\tau_{s}},{\tau_{s}},u)_{x_{\tau_{s-1}}}}$} ~\\
            $\quad\quad\quad\quad\quad\quad\quad\quad\quad\quad\quad\quad\quad\quad\quad\quad\quad\quad+(1-\frac{\alpha_{\tau_{s-1}}}{\alpha_{\tau_{s}}})\cdot \vec{p}_T(x_{\tau_{s-1}})\cdot\sum_{z\in\mathcal{X}} e^{\hat{\mathbf{s}}_{\Theta}(x_{\tau_{s}},{\tau_{s}},u)_{z}}.$
            \State $x_{\tau_{s-1}}\sim p_{{\tau_{s-1}}|{\tau_s}}(x_{\tau_{s-1}}|x_{\tau_{s}},u), x_{\tau_{s-1}}\in\mathcal{X}$.
            \Comment{reverse preference growing.}\\
        \EndFor \\
        \Return $p(x_0|u)=p_{\tau_{0}|\tau_{1}}(x_{\tau_{0}}|x_{\tau_{1}},u),x_0\in\mathcal{X}$ \Comment{grown preference scores of user $u$.}
  \end{algorithmic} 
\end{algorithm} 
\end{minipage}


\section{Experiments Details} \label{app:exp}

\subsection{Datasets} \label{exp:data}

We evaluate PreferGrow on five real-world benchmark datasets:
\begin{itemize}[leftmargin=*]
    \item \textbf{MoviesLens} \cite{MoviesLens} is a commonly used movie recommendation dataset that contains user ratings, movie titles, and movie genres.
    \item \textbf{Steam} \cite{SASRec} encompasses user reviews for video games on the Steam Store.
    \item \textbf{Beauty} \cite{Amazon2014} contains movie details and user reviews from Jun 1996 to Sep 2023.
    \item \textbf{Toys} \cite{Amazon2014} includes user reviews and metadata for toys and games from Jun 1996 to Jul 2014.
    \item \textbf{Sports} \cite{Amazon2014} comprises reviews and metadata for sports and outdoor products from 1996 to 2014.
\end{itemize}
Following prior works \cite{DreamRec, PreferDiff}, we adopt the user-splitting strategy, which has been shown to effectively prevent information leakage in test sets \cite{DataLeakage}. Specifically, we sort all sequences chronologically for each dataset and then split the data into training, validation, and test sets with an 8:1:1 ratio, while preserving the last 10 interactions as the historical sequence. 
The statistical characteristics of the processed dataset are shown in Table \ref{stat_dataset}. 
As observed from the table, the recommendation datasets face a significant challenge of severe data sparsity.

\begin{table}[t]
\caption{Statistics of datasets after preprocessing.}
\label{stat_dataset}
\centering
\begin{tabularx}{0.6\textwidth}{lcccc}
\toprule 
Dataset & \# users & \# items & \# Interactions & sparsity \\
\midrule
Movies & 6040 & 3883 &  1001456 & 04.27\%   \\
Steam & 39795 & 9265 &  2949605 &  00.80\%   \\
Beauty & 22,363 & 12,101 & 198,502 & 00.07\%   \\
Toys & 19,412 & 11,924 &  138,444 & 00.06\%   \\
Sports & 35,598 & 18,357 & 256,598 & 00.04\%   \\
\midrule
\end{tabularx}
\vspace{-8pt}
\end{table}

\subsection{Baselines} \label{exp:baseline}

We compare PreferGrow with both traditional discriminative recommenders using negative sampling and diffusion-based generative recommenders, including classical recommenders (SASRec \cite{SASRec}, Caser \cite{Caser}, GRURec \cite{GRURec}), item-level diffusion-based recommenders (DreamRec \cite{DreamRec}, PreferDiff \cite{PreferDiff}), and preference score-level diffusion-based recommenders (DiffRec \cite{DiffRec2}, DDSR \cite{DDSR}):

\begin{itemize}[leftmargin=*]
    \item \textbf{SASRec} \cite{SASRec} leverages the self-attention mechanism in Transformer to model user preference scores from interaction histories, addressing data sparsity through negative sampling.
    \item \textbf{Caser} \cite{Caser} utilizes horizontal and vertical convolutional filters to capture sequential patterns at the point-level and union-level, allowing for skip behaviors, and models user preference scores while addressing data sparsity through negative sampling.
    \item \textbf{GRURec} \cite{GRURec} adopts RNNs to model user preference scores from interaction histories, mitigating data sparsity through negative sampling.
    \item \textbf{DreamRec} \cite{DreamRec} reshapes sequential recommendation as oracle item generation, addressing data sparsity by adding Gaussian noise to dense item embeddings.
    \item \textbf{PreferDiff} \cite{PreferDiff} introduces an optimization objective specifically designed for item-level DM-based recommenders, which can integrate multiple negative samples, addressing data sparsity by adding noise to dense item embeddings and using negative sampling both.
    \item \textbf{DiffRec} \cite{DiffRec2} is a preference score-level diffusion-based generative recommender assuming a Gaussian prior, addressing data sparsity by adding Gaussian noise to preference scores, without considering the constraints of the probability simplex.
    \item \textbf{DDSR} \cite{DDSR} is a preference score-level diffusion-based generative recommender assuming a categorical prior, addressing data sparsity by adding discrete noise to preference scores while respecting the constraints of the probability simplex.
\end{itemize}

\subsection{Implementation Details}  \label{exp:implement}

\textbf{Training settings}: 
We implement all models using Python 3.7 and PyTorch 1.12.1 on an Nvidia GeForce RTX 3090. During training, all methods are trained with a fixed batch size of 256 using the Adam optimizer. 
Additionally, we apply early stopping based on the model's performance on the validation set.
To ensure reproducibility, we fix all random seeds to 100 in our main experiments, a randomly chosen value.
For the classic recommenders with negative sampling, we employ the binary cross-entropy (BCE) loss. 
PreferGrow uses SASRec as the encoding model for the user's historical sequence, and we adopt both Hybrid-Wise and Adaptive settings for the fading matrix. 
For all SASRec modules, we apply RoPE position encoding \cite{RoPE}.
The search space of hyperparameters for the baselines is shown in Table \ref{tab:param}, and the optimal parameters for our PreferGrow under both hybrid and adaptive settings are presented in Table \ref{tab:bestparam}.
To ensure the reliability of our findings, we have conducted a comprehensive re-evaluation of our experiments under multiple random seeds $\{100,200,300,400,500\}$ and statistical significance testing ($p<0.001$). 
The results shown in Table \ref{tab:reval} consistently confirm the significant performance gains of PreferGrow over baselines. 

\begin{table}[h]
\centering
\footnotesize
\caption{Re-evaluation results (NDCG@5) under multiple random seeds.}
\label{tab:reval}
\begin{tabular}{c c c c c c}
\toprule
\textbf{Dataset} & \textbf{Beauty} & \textbf{Toys} & \textbf{Sports} & \textbf{Steam} & \textbf{MovieLens} \\
\midrule
\textbf{SASRec} & 0.0229$\pm$.0020 & 0.0258$\pm$.0023 & 0.0104$\pm$.0016 & 0.0193$\pm$.0002 & 0.0507$\pm$.0004 \\
\textbf{PreferDiff} & 0.0224$\pm$.0031 & 0.0312$\pm$.0009 & 0.0122$\pm$.0007 & 0.0104$\pm$.0004 & 0.0348$\pm$.0005 \\
\textbf{PreferGrow} & \textbf{0.0315$\pm$.0010} & \textbf{0.0326$\pm$.0007} & \textbf{0.0146$\pm$.0004} & \textbf{0.0399$\pm$.0004} & \textbf{0.0913$\pm$.0003} \\
\bottomrule
\end{tabular}
\end{table}

\textbf{Evaluation Protocols and Metrics}:
To ensure a comprehensive evaluation and mitigate potential biases, we adopt the all-rank protocol \cite{LightGCN, SGL, RLMRec, DreamRec, PreferDiff}, which evaluates recommendations across all items. 
We employ two widely used ranking-based metrics: \textit{Normalized Discounted Cumulative Gain} (\textbf{N$@K$}) and \textit{Mean Reciprocal Rank} (\textbf{M$@K$}), to assess the effectiveness of the models.

\begin{table}[t!]
\centering
\footnotesize
\caption{Hyperparameters Search Space for Baselines.}
\label{tab:param}
\begin{tabular}{l|l}
\toprule
\textbf{Method} & \textbf{Hyperparameter Search Space} \\
\midrule
\multirow{2}{*}{\textbf{Shared}} & lr $\sim$ \{1e-2, 1e-3, 1e-4, 1e-5\} with decay 0, embedding size $d$ $\sim$ \{128, 256, 512\} \\
 & the number of negative sampling (if using) $\sim$ \{64, 128, 256\}, bath size = 256\\
\midrule
\textbf{DreamRec} & $w \sim$ \{0, 1, 2, 5, 10\}, $T\sim$ \{500, 1000, 2000,3000\}, $p \sim$ \{0.05, 0.1, 0.15, 0.2, 0.25, 0.3\}\\
\midrule
\textbf{PreferDiff} & $\lambda \sim$ \{0.2, 0.4, 0.6, 0.8\}, $w \sim$ \{0, 1, 2, 5, 10\}, $T\sim$ \{500, 1000, 2000,3000\}\\ 
\midrule
\textbf{DiffRec} & noise scale $\sim$ \{1e-1, 1e-2, 1e-3, 1e-4, 1e-5\}, $T \sim$ \{2, 5, 20, 50, 100\} \\
\midrule
\textbf{DDSR} & $T\sim$ \{500, 1000, 2000,3000\} \\
\midrule
\multirow{2}{*}{\textbf{PreferGrow}} & $T \sim$ \{5, 10, 20, 30, 40\}, $p \sim$ \{0.05, 0.1, 0.15, 0.2, 0.25, 0.3\} \\
 & $w \sim$ \{0, 1, 2, 5, 10\}, $\lambda \sim$ \{0.9, 0.99, 0.999, 0.9999, 0.99999\} \\
\bottomrule
\end{tabular}
\end{table}

\begin{table}[t!]
\centering
\footnotesize
\caption{Best Hyperparameters for PreferGrow on five datasets.}
\label{tab:bestparam}
\begin{tabular}{l|l|c|c|c|c|c|c}
\toprule
\textbf{Variant}&\textbf{Dataset} & lr & $d$& $p$ & $T$ & $w$ & $\lambda$\\
\midrule
\multirow{5}{*}{\textbf{Hybrid}} & \textbf{MovieLens} 
& 1e-4 & 256 & 0.1 & 20 & 10 & 0.9999  \\
& \textbf{Steam} & 1e-3 & 256 & 0.1 & 20 & 10 & 0.99999 \\
& \textbf{Beauty} & 1e-4 & 256 & 0.1 & 20 & 5 & 0.999 \\
& \textbf{Toys} & 1e-3 & 256 & 0.2 & 20 & 10 & 0.9999 \\
& \textbf{Sports} & 1e-3 & 256 & 0.2 & 20 & 1 & 0.9999 \\
\midrule
\textbf{Variant}&\textbf{Method} & lr & $d$& $p$ & $T$ & $w$ & $x_{-1}$\\
\midrule
\multirow{5}{*}{\textbf{Adaptive}} & \textbf{Dataset} 
& 1e-4 & 256 & 0.2 & 20 & 10 & True  \\
& \textbf{Steam} & 1e-3 & 256 & 0.05 & 20 & 10 & True \\
& \textbf{Beauty} & 1e-4 & 256 & 0.1 & 20 & 2 & True \\
& \textbf{Toys} & 1e-4 & 256 & 0.2 & 20 & 5 & True \\
& \textbf{Sports} & 1e-4 & 256 & 0.2 & 20 & 5 & True \\
\bottomrule
\end{tabular}
\end{table}

\subsection{Hyper-parameter Analysis} \label{app:param_ana}

The personalization strength $w$ is locally stable within a reasonable range and highly consistent across data splits, thus posing no practical challenge for tuning.
As shown in Figure \ref{fig:param} of our paper, PreferGrow is sensitive to large-magnitude changes in the personalization strength parameter $w$ (e.g., from 0 to 20). 
On datasets like \textit{Steam}, we observe that performance remains locally stable within a reasonable range (e.g., $w \in [5, 15]$), but drops notably when $w$ is set too small (\eg $w=1$). 
This highlights the need to set $w$ appropriately for each dataset. 
To assess sensitivity more systematically, we tested PreferGrow under $w \in \{0, 2, 5, 10\}$ and evaluated the mean absolute error between the optimal $w$ values selected on the training, validation, and test sets. These optimal values are determined by uniformly sampling performance across 40 checkpoints throughout the training process and optimality is defined by maximizing the combined metric $HR@5 + HR@10 + NDCG@5 + NDCG@10$. 
Our findings in Table \ref{tab:mae-w} show that the optimal $w$ is highly consistent across data splits, indicating that $w$ can be tuned on a validation (or even training) set like other standard hyperparameters. 
This ensures that tuning  $w$ does not pose a practical challenge.

\begin{table}[t!]
\centering
\caption{Mean absolute error of $w$ across data splits on five datasets.}
\label{tab:mae-w}
\begin{tabular}{lccccc}
\toprule
\textbf{Dataset} & \textbf{MovieLens} & \textbf{Steam} & \textbf{Beauty} & \textbf{Sports} & \textbf{Toys} \\
\midrule
Mean abs.\ error of $w$ & 0.000 & 0.000 & 0.778 & 0.050 & 0.775 \\
\bottomrule
\end{tabular}
\end{table}

\subsection{Efficiency Analysis}\label{exp:time}

As summarized in Table~\ref{tab:efficiency}, we report each model’s trainable parameters, the number of training epochs, GFLOPs per model output (recall that diffusion models require multiple outputs for denoising), and the total number of outputs.

\textbf{Quantitative scalability analysis.}
We also provide a quantitative analysis of PreferGrow and clarify the solution outlined in the Limitations. 
PreferGrow introduces only $\mathcal{O}(N)$ complexity in the network, loss, and inference paths --- comparable to contemporary diffusion-based recommenders.
Table~\ref{tab:complexity} contrasts model, loss, and inference complexities with representative baselines, where $L$ is the history length, $N$ the item set size, $B$ the number of negatives, $d$ the hidden dimension, and $T$ the number of diffusion steps.

\begin{itemize}[leftmargin=*]
    \item \textbf{Model parameters.} We never materialize the full $N^2$ preference-ratio matrix. 
    Equation~\eqref{s_theta_prefgrow} requires only $N$ scores; the additional computation cost is thus $\mathcal{O}(Nd)$.
    \item \textbf{Loss computation.} With the idempotent fading matrix $\mathbf{E}$, the objective in Eq.~(4) simplifies (Appendix~C.1--C.4) to element-wise means, indexing, and a precomputable constant—overall $\mathcal{O}(N)$. 
    The score-entropy loss contains three parts: (i) a positive term and (ii) a negative term (both simple mean/index operations, $\mathcal{O}(N)$), and (iii) a constant term ensuring non-negativity that depends only on $t$ and can be precomputed in $\mathcal{O}(T)$ with $T{=}20$.
    \item \textbf{Inference cost.} Idempotency yields
    \[
      \mathbf{E}\, s_{\Theta}(x_t,t,u)
      \;=\;
      \Bigl(\frac{\mathbf{p}_T^{\top} s_{\Theta}}{\mathbf{1}^{\top}\mathbf{p}_T}\Bigr)\mathbf{1},
    \]
    reducing Equation~\eqref{Pst_vec} from $\mathcal{O}(N^2)$ to $\mathcal{O}(N)$.
\end{itemize}

\textbf{Where $\mathcal{O}(N^2)$ comes from.}
The apparent $\mathcal{O}(N^2)$ arises only from the \emph{modeling target}: preference ratios are more expressive than conventional logits, which raises PreferGrow’s performance ceiling but does not inflate the network/algorithmic path beyond $\mathcal{O}(N)$.

\textbf{Scaling to industry scale.}
When $N$ is extremely large (\eg billions of items), even $\mathcal{O}(N)$ becomes impractical. 
As discussed in the Limitations, Semantic IDs (SIDs) offer a principled route to reduce complexity.
SIDs encode each item with $m$ codebooks of size $c$, enabling up to $c^m$ distinct items while keeping computation at $\mathcal{O}(mc)$; since $c^m \gg N \gg mc$, this yields a compact yet expressive representation space and allows PreferGrow-on-SIDs to reduce cost from $\mathcal{O}(N)$ to $\mathcal{O}(mc)$.
A practical blocker is that current SID pipelines (e.g., Tiger~\cite{tiger}, ActionPiece~\cite{actionpiece}, DDSR~\cite{DDSR}, OneRec~\cite{onerec}) are not open-sourced. 
Truly deploying PreferGrow at industrial scale therefore requires addressing the open challenge of \emph{large-scale multimodal SID pretraining}. 
Extending PreferGrow to operate over SIDs is our ongoing work.

\hugq{
\begin{table}[t!]
\centering
\scriptsize
\caption{Comparison of model complexities.}
\label{tab:complexity}
\begin{tabular}{c|c|c|c|c}
\toprule
\textbf{Complexity} & \textbf{Model Parameters} & \textbf{Loss Computation} & \textbf{Modeling Target} & \textbf{Inference} \\
\midrule
\textbf{SASRec} & $\mathcal{O}\bigl(Ld^{2}+L^{2}d\bigr)$ & $\mathcal{O}\bigl(Bd\bigr)$ & $\mathcal{O}(N)$ & $\mathcal{O}\bigl(Ld^{2}+L^{2}d+Nd\bigr)$ \\
\midrule
\textbf{DreamRec} & $\mathcal{O}\bigl((L+3)d^{2}+L^{2}d\bigr)$ & $\mathcal{O}(d)$ & $\mathcal{O}(Td)$ & $\mathcal{O}\bigl(T(L+3)d^{2}+TL^{2}d+Nd\bigr)$ \\
\midrule
\textbf{PreferDiff} & $\mathcal{O}\bigl((L+3)d^{2}+L^{2}d\bigr)$ & $\mathcal{O}\bigl(Bd\bigr)$ & $\mathcal{O}(Td)$ & $\mathcal{O}\bigl(T(L+3)d^{2}+TL^{2}d+Nd\bigr)$ \\
\midrule
\textbf{DiffRec} & $\mathcal{O}(LN^{2})$ & $\mathcal{O}(N)$ & $\mathcal{O}(TN)$ & $\mathcal{O}(TLN^{2})$ \\
\midrule
\textbf{DDSR} & $\mathcal{O}\bigl(Ld^{2}+L^{2}d\bigr)$ & $\mathcal{O}(N)$ & $\mathcal{O}(TN)$ & $\mathcal{O}\bigl(TLd^{2}+TL^{2}d+TNd\bigr)$ \\
\midrule
\textbf{PreferGrow} & $\mathcal{O}\bigl((L+3)d^{2}+L^{2}d\bigr)$ & $\mathcal{O}(N+T)$ & $\mathcal{O}(TN^{2})$ & $\mathcal{O}\bigl(T(L+3)d^{2}+TL^{2}d+TN(d+3)\bigr)$ \\
\bottomrule
\end{tabular}
\end{table}
}

\begin{table}[t!]
\centering
\footnotesize
\caption{Comparison of efficiency on Steam.}
\label{tab:efficiency}
\begin{tabular}{lcccc }
\toprule
\textbf{Models} &\# Trainable Parameters & \# training epochs & Inference GFLOPs & Inference steps \\
\midrule
\textbf{SASRec} & 2.70M & 61 & 0.85 & 1   \\
\textbf{Caser} & 2.49M & 58 & 0.38  & 1 \\
\textbf{GRURec} & 2.77M & 55 & 4.06 & 1 \\
\textbf{DreamRec} & 7.25M & 52 & 0.85 & 20  \\
\textbf{PreferDiff} & 7.25M & 55 & 0.85 & 20  \\
\textbf{DiffRec} & 18.55M & 1000 & /  & 20 \\
\textbf{DDSR} & 3.03M & 142 & 1.77 & 20 \\
\textbf{PreferGrow} & 3.03M & 480 & 0.93 & 20 \\
\bottomrule
\end{tabular}
\end{table}

%% file: main.bbl
\begin{thebibliography}{75}
\providecommand{\natexlab}[1]{#1}
\providecommand{\url}[1]{\texttt{#1}}
\expandafter\ifx\csname urlstyle\endcsname\relax
  \providecommand{\doi}[1]{doi: #1}\else
  \providecommand{\doi}{doi: \begingroup \urlstyle{rm}\Url}\fi

\bibitem[Idrissi and Zellou(2020)]{sparse1}
Nouhaila Idrissi and Ahmed Zellou.
\newblock A systematic literature review of sparsity issues in recommender systems.
\newblock \emph{Social Network Analysis and Mining}, 2020.

\bibitem[Singh(2020)]{sparse2}
Monika Singh.
\newblock Scalability and sparsity issues in recommender datasets: a survey.
\newblock \emph{Knowledge and Information Systems}, 2020.

\bibitem[Cui et~al.(2021)Cui, Chen, Lyu, Yang, and Li]{S3Rec}
Jinming Cui, Chaochao Chen, Lingjuan Lyu, Carl Yang, and Wang Li.
\newblock Exploiting data sparsity in secure cross-platform social recommendation.
\newblock \emph{NeurIPS}, 2021.

\bibitem[Yin et~al.(2020)Yin, Wang, Zheng, Li, and Zhou]{sparse4}
Hongzhi Yin, Qinyong Wang, Kai Zheng, Zhixu Li, and Xiaofang Zhou.
\newblock Overcoming data sparsity in group recommendation.
\newblock \emph{TKDE}, 2020.

\bibitem[Song and Ermon(2019)]{NCSN}
Yang Song and Stefano Ermon.
\newblock Generative modeling by estimating gradients of the data distribution.
\newblock \emph{NeurIPS}, 2019.

\bibitem[Song et~al.(2021)Song, Meng, and Ermon]{DDIM}
Jiaming Song, Chenlin Meng, and Stefano Ermon.
\newblock Denoising diffusion implicit models.
\newblock In \emph{{ICLR}}, 2021.

\bibitem[Ho et~al.(2020)Ho, Jain, and Abbeel]{DDPM}
Jonathan Ho, Ajay Jain, and Pieter Abbeel.
\newblock Denoising diffusion probabilistic models.
\newblock In \emph{NeurIPS}, 2020.

\bibitem[Yang et~al.(2023)Yang, Wu, Wang, Wang, Yuan, and He]{DreamRec}
Zhengyi Yang, Jiancan Wu, Zhicai Wang, Xiang Wang, Yancheng Yuan, and Xiangnan He.
\newblock Generate what you prefer: Reshaping sequential recommendation via guided diffusion.
\newblock In \emph{NeurIPS}, 2023.

\bibitem[Li et~al.(2023{\natexlab{a}})Li, Sun, and Li]{DiffuRec}
Zihao Li, Aixin Sun, and Chenliang Li.
\newblock Diffurec: A diffusion model for sequential recommendation.
\newblock \emph{TOIS}, 2023{\natexlab{a}}.

\bibitem[Du et~al.(2023)Du, Yuan, Huang, Zhao, and Zhou]{DiffRec1}
Hanwen Du, Huanhuan Yuan, Zhen Huang, Pengpeng Zhao, and Xiaofang Zhou.
\newblock Sequential recommendation with diffusion models.
\newblock \emph{arXiv}, 2023.

\bibitem[Wang et~al.(2024)Wang, Liu, Yang, and Yu]{CDDRec}
Yu~Wang, Zhiwei Liu, Liangwei Yang, and Philip~S Yu.
\newblock Conditional denoising diffusion for sequential recommendation.
\newblock In \emph{PAKDD}, 2024.

\bibitem[Liu et~al.(2023)Liu, Yan, Zhao, Du, Guo, Tang, and Tian]{DiffuASR}
Qidong Liu, Fan Yan, Xiangyu Zhao, Zhaocheng Du, Huifeng Guo, Ruiming Tang, and Feng Tian.
\newblock Diffusion augmentation for sequential recommendation.
\newblock In \emph{CIKM}, 2023.

\bibitem[Zhao et~al.(2024)Zhao, Wenjie, Xu, Sun, Feng, and Chua]{DDRM}
Jujia Zhao, Wang Wenjie, Yiyan Xu, Teng Sun, Fuli Feng, and Tat-Seng Chua.
\newblock Denoising diffusion recommender model.
\newblock In \emph{SIGIR}, 2024.

\bibitem[Xuan(2024)]{DiffCDR}
Yuner Xuan.
\newblock Diffusion cross-domain recommendation.
\newblock \emph{arXiv}, 2024.

\bibitem[Yi et~al.(2024)Yi, Wang, and Ounis]{DiffGT}
Zixuan Yi, Xi~Wang, and Iadh Ounis.
\newblock A directional diffusion graph transformer for recommendation.
\newblock \emph{arXiv}, 2024.

\bibitem[Hou et~al.(2024)Hou, Park, and Shin]{CFDiff}
Yu~Hou, Jin-Duk Park, and Won-Yong Shin.
\newblock Collaborative filtering based on diffusion models: Unveiling the potential of high-order connectivity.
\newblock In \emph{SIGIR}, 2024.

\bibitem[Cui et~al.(2024)Cui, Wu, He, Cheng, and Ma]{CaDiRec}
Ziqiang Cui, Haolun Wu, Bowei He, Ji~Cheng, and Chen Ma.
\newblock Context matters: Enhancing sequential recommendation with context-aware diffusion-based contrastive learning.
\newblock In \emph{CIKM}, 2024.

\bibitem[Tomasi et~al.(2024)Tomasi, Fabbri, Lalmas, and Dai]{DMSR}
Federico Tomasi, Francesco Fabbri, Mounia Lalmas, and Zhenwen Dai.
\newblock Diffusion model for slate recommendation.
\newblock \emph{arXiv}, 2024.

\bibitem[Huang and Wang(2025)]{DiffCL}
Fan Huang and Wei Wang.
\newblock Diffusion-augmented graph contrastive learning for collaborative filter.
\newblock \emph{arXiv}, 2025.

\bibitem[Li et~al.(2025{\natexlab{a}})Li, Tang, Sheng, Zhang, Gao, Cheng, Yin, and Liu]{DMCDR}
Xiaodong Li, Hengzhu Tang, Jiawei Sheng, Xinghua Zhang, Li~Gao, Suqi Cheng, Dawei Yin, and Tingwen Liu.
\newblock Exploring preference-guided diffusion model for cross-domain recommendation.
\newblock \emph{arXiv}, 2025{\natexlab{a}}.

\bibitem[Li et~al.(2025{\natexlab{b}})Li, Wang, Zhang, Yu, and Chen]{MoDiCF}
Jin Li, Shoujin Wang, Qi~Zhang, Shui Yu, and Fang Chen.
\newblock Generating with fairness: A modality-diffused counterfactual framework for incomplete multimodal recommendations.
\newblock \emph{arXiv}, 2025{\natexlab{b}}.

\bibitem[Mao et~al.(2025)Mao, Liu, Liu, Liu, Li, and Hu]{DiQDiff}
Wenyu Mao, Shuchang Liu, Haoyang Liu, Haozhe Liu, Xiang Li, and Lantao Hu.
\newblock Distinguished quantized guidance for diffusion-based sequence recommendation.
\newblock \emph{WWW}, 2025.

\bibitem[Zhao et~al.(2025)Zhao, Yang, Liang, Zhao, Guo, and Wang]{DRGO}
Chu Zhao, Enneng Yang, Yuliang Liang, Jianzhe Zhao, Guibing Guo, and Xingwei Wang.
\newblock Distributionally robust graph out-of-distribution recommendation via diffusion model.
\newblock \emph{arXiv}, 2025.

\bibitem[Zolghadr et~al.(2024)Zolghadr, Winther, and Jeha]{DiffuRec2}
Sharare Zolghadr, Ole Winther, and Paul Jeha.
\newblock Generative diffusion models for sequential recommendations.
\newblock \emph{arXiv}, 2024.

\bibitem[Hu et~al.(2024)Hu, Yang, Cai, Zhang, and Wang]{iDreamRec}
Guoqing Hu, Zhengyi Yang, Zhibo Cai, An~Zhang, and Xiang Wang.
\newblock Generate and instantiate what you prefer: Text-guided diffusion for sequential recommendation.
\newblock In \emph{arXiv}, 2024.

\bibitem[Liu et~al.(2025)Liu, Zhang, Hu, Qian, and seng Chua]{PreferDiff}
Shuo Liu, An~Zhang, Guoqing Hu, Hong Qian, and Tat seng Chua.
\newblock Preference diffusion for recommendation.
\newblock In \emph{ICLR}, 2025.

\bibitem[Wang et~al.(2023)Wang, Xu, Feng, Lin, He, and Chua]{DiffRec2}
Wenjie Wang, Yiyan Xu, Fuli Feng, Xinyu Lin, Xiangnan He, and Tat-Seng Chua.
\newblock Diffusion recommender model.
\newblock In \emph{SIGIR}, 2023.

\bibitem[B{\'e}n{\'e}dict et~al.(2023)B{\'e}n{\'e}dict, Jeunen, Papa, Bhargav, Odijk, and de~Rijke]{RecFusion}
Gabriel B{\'e}n{\'e}dict, Olivier Jeunen, Samuele Papa, Samarth Bhargav, Daan Odijk, and Maarten de~Rijke.
\newblock Recfusion: A binomial diffusion process for 1d data for recommendation.
\newblock \emph{arXiv}, 2023.

\bibitem[Yu et~al.(2023)Yu, Tan, Lu, and Bao]{LD4MRec}
Penghang Yu, Zhiyi Tan, Guanming Lu, and Bing-Kun Bao.
\newblock Ld4mrec: Simplifying and powering diffusion model for multimedia recommendation.
\newblock \emph{arXiv}, 2023.

\bibitem[Ma et~al.(2024)Ma, Xie, Meng, Chen, Zhang, Lin, and Kang]{PDRec}
Haokai Ma, Ruobing Xie, Lei Meng, Xin Chen, Xu~Zhang, Leyu Lin, and Zhanhui Kang.
\newblock Plug-in diffusion model for sequential recommendation.
\newblock In \emph{AAAI}, 2024.

\bibitem[Priyam et~al.(2024)Priyam, Shah, and Botta]{EdgeRec}
Utkarsh Priyam, Hemit Shah, and Edoardo Botta.
\newblock Edge-rec: Efficient and data-guided edge diffusion for recommender systems graphs.
\newblock \emph{arXiv}, 2024.

\bibitem[Jiang and Fan(2024)]{DifFaiRec}
Zhenhao Jiang and Jicong Fan.
\newblock Diffairec: Generative fair recommender with conditional diffusion model.
\newblock \emph{arXiv}, 2024.

\bibitem[Han et~al.(2024)Han, Kweon, Kim, and Yu]{D3Rec}
Gwangseok Han, Wonbin Kweon, Minsoo Kim, and Hwanjo Yu.
\newblock Controlling diversity at inference: Guiding diffusion recommender models with targeted category preferences.
\newblock \emph{arXiv preprint arXiv:2411.11240}, 2024.

\bibitem[Xie et~al.(2024)Xie, Wang, Zhang, Zhou, Lian, and Chen]{DDSR}
Wenjia Xie, Hao Wang, Luankang Zhang, Rui Zhou, Defu Lian, and Enhong Chen.
\newblock Breaking determinism: Fuzzy modeling of sequential recommendation using discrete state space diffusion model.
\newblock \emph{NeurIPS}, 2024.

\bibitem[Jiang et~al.(2024)Jiang, Xia, Wei, Luo, Lin, and Huang]{DiffMM}
Yangqin Jiang, Lianghao Xia, Wei Wei, Da~Luo, Kangyi Lin, and Chao Huang.
\newblock Diffmm: Multi-modal diffusion model for recommendation.
\newblock In \emph{MM}, 2024.

\bibitem[Xia et~al.(2025)Xia, Cheng, Tang, Liu, Liu, Wang, and Jiang]{SDiff}
Rui Xia, Yanhua Cheng, Yongxiang Tang, Xiaocheng Liu, Xialong Liu, Lisong Wang, and Peng Jiang.
\newblock S-diff: An anisotropic diffusion model for collaborative filtering in spectral domain.
\newblock In \emph{WSDM}, 2025.

\bibitem[Lou et~al.(2024)Lou, Meng, and Ermon]{DDSE}
Aaron Lou, Chenlin Meng, and Stefano Ermon.
\newblock Discrete diffusion modeling by estimating the ratios of the data distribution.
\newblock \emph{ICML}, 2024.

\bibitem[Meng et~al.(2022)Meng, Choi, Song, and Ermon]{concretescore}
Chenlin Meng, Kristy Choi, Jiaming Song, and Stefano Ermon.
\newblock Concrete score matching: Generalized score matching for discrete data.
\newblock \emph{NeurIPS}, 2022.

\bibitem[Ou et~al.(2025)Ou, Nie, Xue, Zhu, Sun, Li, and Li]{RADD}
Jingyang Ou, Shen Nie, Kaiwen Xue, Fengqi Zhu, Jiacheng Sun, Zhenguo Li, and Chongxuan Li.
\newblock Your absorbing discrete diffusion secretly models the conditional distributions of clean data.
\newblock \emph{ICLR}, 2025.

\bibitem[Luo(2022)]{GodDiff}
Calvin Luo.
\newblock Understanding diffusion models: A unified perspective.
\newblock \emph{arXiv}, 2022.

\bibitem[Bai et~al.(2022)Bai, Jones, Ndousse, Askell, Chen, DasSarma, Drain, Fort, Ganguli, Henighan, et~al.]{RLHF1}
Yuntao Bai, Andy Jones, Kamal Ndousse, Amanda Askell, Anna Chen, Nova DasSarma, Dawn Drain, Stanislav Fort, Deep Ganguli, Tom Henighan, et~al.
\newblock Training a helpful and harmless assistant with reinforcement learning from human feedback.
\newblock \emph{arXiv}, 2022.

\bibitem[Dong et~al.(2024)Dong, Xiong, Pang, Wang, Zhao, Zhou, Jiang, Sahoo, Xiong, and Zhang]{OnlineRLHF}
Hanze Dong, Wei Xiong, Bo~Pang, Haoxiang Wang, Han Zhao, Yingbo Zhou, Nan Jiang, Doyen Sahoo, Caiming Xiong, and Tong Zhang.
\newblock {RLHF} workflow: From reward modeling to online {RLHF}.
\newblock \emph{TMLR}, 2024.

\bibitem[Xiong et~al.(2024)Xiong, Dong, Ye, Wang, Zhong, Ji, Jiang, and Zhang]{IterRLHF}
Wei Xiong, Hanze Dong, Chenlu Ye, Ziqi Wang, Han Zhong, Heng Ji, Nan Jiang, and Tong Zhang.
\newblock Iterative preference learning from human feedback: Bridging theory and practice for {RLHF} under kl-constraint.
\newblock In \emph{{ICML}}, 2024.

\bibitem[Rafailov et~al.(2023)Rafailov, Sharma, Mitchell, Manning, Ermon, and Finn]{DPO}
Rafael Rafailov, Archit Sharma, Eric Mitchell, Christopher~D. Manning, Stefano Ermon, and Chelsea Finn.
\newblock Direct preference optimization: Your language model is secretly a reward model.
\newblock In \emph{NeurIPS}, 2023.

\bibitem[Wu et~al.(2024)Wu, Xie, Yang, Wu, Gao, Ding, Wang, and He]{BetaDPO}
Junkang Wu, Yuexiang Xie, Zhengyi Yang, Jiancan Wu, Jinyang Gao, Bolin Ding, Xiang Wang, and Xiangnan He.
\newblock {\(\beta\)}-dpo: Direct preference optimization with dynamic {\(\beta\)}.
\newblock In \emph{NeurIPS}, 2024.

\bibitem[Chen et~al.(2024)Chen, Tan, Zhang, Yang, Sheng, Zhang, Wang, and Chua]{SDPO}
Yuxin Chen, Junfei Tan, An~Zhang, Zhengyi Yang, Leheng Sheng, Enzhi Zhang, Xiang Wang, and Tat{-}Seng Chua.
\newblock On softmax direct preference optimization for recommendation.
\newblock In \emph{NeurIPS}, 2024.

\bibitem[Wu et~al.(2025)Wu, Wang, Yang, Wu, Gao, Ding, Wang, and He]{AlphaDPO}
Junkang Wu, Xue Wang, Zhengyi Yang, Jiancan Wu, Jinyang Gao, Bolin Ding, Xiang Wang, and Xiangnan He.
\newblock {\(\alpha\)}-dpo: Adaptive reward margin is what direct preference optimization needs.
\newblock \emph{ICML}, 2025.

\bibitem[Bradley and Terry(1952)]{BTPrefer}
Ralph~Allan Bradley and Milton~E Terry.
\newblock Rank analysis of incomplete block designs: I. the method of paired comparisons.
\newblock \emph{Biometrika}, 1952.

\bibitem[Austin et~al.(2021{\natexlab{a}})Austin, Johnson, Ho, Tarlow, and van~den Berg]{DDM1}
Jacob Austin, Daniel~D. Johnson, Jonathan Ho, Daniel Tarlow, and Rianne van~den Berg.
\newblock Structured denoising diffusion models in discrete state-spaces.
\newblock In \emph{NeurIPS}, 2021{\natexlab{a}}.

\bibitem[Campbell et~al.(2022{\natexlab{a}})Campbell, Benton, Bortoli, Rainforth, Deligiannidis, and Doucet]{DDM2}
Andrew Campbell, Joe Benton, Valentin~De Bortoli, Thomas Rainforth, George Deligiannidis, and Arnaud Doucet.
\newblock A continuous time framework for discrete denoising models.
\newblock In \emph{NeurIPS}, 2022{\natexlab{a}}.

\bibitem[Shi et~al.(2024)Shi, Han, Wang, Doucet, and Titsias]{MDD1}
Jiaxin Shi, Kehang Han, Zhe Wang, Arnaud Doucet, and Michalis Titsias.
\newblock Simplified and generalized masked diffusion for discrete data.
\newblock \emph{NeurIPS}, 2024.

\bibitem[Zheng et~al.(2025)Zheng, Chen, Mao, Liu, Zhu, and Zhang]{MDD2}
Kaiwen Zheng, Yongxin Chen, Hanzi Mao, Ming-Yu Liu, Jun Zhu, and Qinsheng Zhang.
\newblock Masked diffusion models are secretly time-agnostic masked models and exploit inaccurate categorical sampling.
\newblock \emph{ICLR}, 2025.

\bibitem[Campbell et~al.(2022{\natexlab{b}})Campbell, Benton, De~Bortoli, Rainforth, Deligiannidis, and Doucet]{LDR}
Andrew Campbell, Joe Benton, Valentin De~Bortoli, Thomas Rainforth, George Deligiannidis, and Arnaud Doucet.
\newblock A continuous time framework for discrete denoising models.
\newblock \emph{NeurIPS}, 2022{\natexlab{b}}.

\bibitem[Sun et~al.(2023)Sun, Yu, Dai, Schuurmans, and Dai]{SDDM}
Haoran Sun, Lijun Yu, Bo~Dai, Dale Schuurmans, and Hanjun Dai.
\newblock Score-based continuous-time discrete diffusion models.
\newblock \emph{ICLR}, 2023.

\bibitem[Sun et~al.(2019)Sun, Liu, Wu, Pei, Lin, Ou, and Jiang]{Bert4Rec}
Fei Sun, Jun Liu, Jian Wu, Changhua Pei, Xiao Lin, Wenwu Ou, and Peng Jiang.
\newblock Bert4rec: Sequential recommendation with bidirectional encoder representations from transformer.
\newblock In \emph{CIKM}, 2019.

\bibitem[van~den Oord et~al.(2018)van~den Oord, Li, and Vinyals]{InfoNCE}
A{\"{a}}ron van~den Oord, Yazhe Li, and Oriol Vinyals.
\newblock Representation learning with contrastive predictive coding.
\newblock \emph{CoRR}, 2018.

\bibitem[Kang and McAuley(2018)]{SASRec}
Wang{-}Cheng Kang and Julian~J. McAuley.
\newblock Self-attentive sequential recommendation.
\newblock In \emph{{ICDM}}, 2018.

\bibitem[Ho and Salimans(2022)]{CFG}
Jonathan Ho and Tim Salimans.
\newblock Classifier-free diffusion guidance.
\newblock \emph{arXiv}, 2022.

\bibitem[Tang and Wang(2018)]{Caser}
Jiaxi Tang and Ke~Wang.
\newblock Personalized top-n sequential recommendation via convolutional sequence embedding.
\newblock In \emph{{WSDM}}, pages 565--573. {ACM}, 2018.

\bibitem[Hidasi et~al.(2016)Hidasi, Karatzoglou, Baltrunas, and Tikk]{GRURec}
Bal{\'{a}}zs Hidasi, Alexandros Karatzoglou, Linas Baltrunas, and Domonkos Tikk.
\newblock Session-based recommendations with recurrent neural networks.
\newblock In \emph{{ICLR} (Poster)}, 2016.

\bibitem[Li et~al.(2023{\natexlab{b}})Li, Wang, Li, Fu, Shen, Shang, and McAuley]{Recformer}
Jiacheng Li, Ming Wang, Jin Li, Jinmiao Fu, Xin Shen, Jingbo Shang, and Julian McAuley.
\newblock Text is all you need: Learning language representations for sequential recommendation.
\newblock In \emph{KDD}, pages 1258--1267, 2023{\natexlab{b}}.

\bibitem[Rendle et~al.(2009)Rendle, Freudenthaler, Gantner, and Schmidt-Thieme]{BPR}
Steffen Rendle, Christoph Freudenthaler, Zeno Gantner, and Lars Schmidt-Thieme.
\newblock Bpr: Bayesian personalized ranking from implicit feedback.
\newblock \emph{UAI}, 2009.

\bibitem[Li et~al.(2025{\natexlab{c}})Li, Huang, Zhao, Liu, Zheng, Liu, Mou, Zhou, Lian, Song, et~al.]{DimeRec}
Wuchao Li, Rui Huang, Haijun Zhao, Chi Liu, Kai Zheng, Qi~Liu, Na~Mou, Guorui Zhou, Defu Lian, Yang Song, et~al.
\newblock Dimerec: A unified framework for enhanced sequential recommendation via generative diffusion models.
\newblock In \emph{WSDM}, 2025{\natexlab{c}}.

\bibitem[Austin et~al.(2021{\natexlab{b}})Austin, Johnson, Ho, Tarlow, and Van Den~Berg]{D3PM}
Jacob Austin, Daniel~D Johnson, Jonathan Ho, Daniel Tarlow, and Rianne Van Den~Berg.
\newblock Structured denoising diffusion models in discrete state-spaces.
\newblock \emph{NeurIPS}, 2021{\natexlab{b}}.

\bibitem[Campbell et~al.(2024)Campbell, Yim, Barzilay, Rainforth, and Jaakkola]{DFM}
Andrew Campbell, Jason Yim, Regina Barzilay, Tom Rainforth, and Tommi Jaakkola.
\newblock Generative flows on discrete state-spaces: enabling multimodal flows with applications to protein co-design.
\newblock In \emph{ICML}, 2024.

\bibitem[Harper and Konstan(2016)]{MoviesLens}
F.~Maxwell Harper and Joseph~A. Konstan.
\newblock The movielens datasets: History and context.
\newblock \emph{{ACM} Trans. Interact. Intell. Syst.}, 2016.

\bibitem[McAuley et~al.(2015)McAuley, Targett, Shi, and van~den Hengel]{Amazon2014}
Julian~J. McAuley, Christopher Targett, Qinfeng Shi, and Anton van~den Hengel.
\newblock Image-based recommendations on styles and substitutes.
\newblock In \emph{{SIGIR}}, 2015.

\bibitem[Ji et~al.(2023)Ji, Sun, Zhang, and Li]{DataLeakage}
Yitong Ji, Aixin Sun, Jie Zhang, and Chenliang Li.
\newblock A critical study on data leakage in recommender system offline evaluation.
\newblock \emph{TOIS}, 2023.

\bibitem[Su et~al.(2024)Su, Ahmed, Lu, Pan, Bo, and Liu]{RoPE}
Jianlin Su, Murtadha Ahmed, Yu~Lu, Shengfeng Pan, Wen Bo, and Yunfeng Liu.
\newblock Roformer: Enhanced transformer with rotary position embedding.
\newblock \emph{Neurocomputing}, 2024.

\bibitem[He et~al.(2020)He, Deng, Wang, Li, Zhang, and Wang]{LightGCN}
Xiangnan He, Kuan Deng, Xiang Wang, Yan Li, Yong{-}Dong Zhang, and Meng Wang.
\newblock Lightgcn: Simplifying and powering graph convolution network for recommendation.
\newblock In \emph{{SIGIR}}, 2020.

\bibitem[Wu et~al.(2021)Wu, Wang, Feng, He, Chen, Lian, and Xie]{SGL}
Jiancan Wu, Xiang Wang, Fuli Feng, Xiangnan He, Liang Chen, Jianxun Lian, and Xing Xie.
\newblock Self-supervised graph learning for recommendation.
\newblock \emph{SIGIR}, 2021.

\bibitem[Ren et~al.(2024)Ren, Wei, Xia, Su, Cheng, Wang, Yin, and Huang]{RLMRec}
Xubin Ren, Wei Wei, Lianghao Xia, Lixin Su, Suqi Cheng, Junfeng Wang, Dawei Yin, and Chao Huang.
\newblock Representation learning with large language models for recommendation.
\newblock In \emph{{WWW}}, 2024.

\bibitem[Rajput et~al.(2023)Rajput, Mehta, Singh, Hulikal~Keshavan, Vu, Heldt, Hong, Tay, Tran, Samost, et~al.]{tiger}
Shashank Rajput, Nikhil Mehta, Anima Singh, Raghunandan Hulikal~Keshavan, Trung Vu, Lukasz Heldt, Lichan Hong, Yi~Tay, Vinh Tran, Jonah Samost, et~al.
\newblock Recommender systems with generative retrieval.
\newblock \emph{NeurIPS}, 2023.

\bibitem[Hou et~al.(2025)Hou, Ni, He, Sachdeva, Kang, Chi, McAuley, and Cheng]{actionpiece}
Yupeng Hou, Jianmo Ni, Zhankui He, Noveen Sachdeva, Wang-Cheng Kang, Ed~H Chi, Julian McAuley, and Derek~Zhiyuan Cheng.
\newblock Actionpiece: Contextually tokenizing action sequences for generative recommendation.
\newblock In \emph{ICML}, 2025.

\bibitem[Zhou et~al.(2025)Zhou, Deng, Zhang, Cai, Ren, Luo, Wang, Hu, Huang, Wang, et~al.]{onerec}
Guorui Zhou, Jiaxin Deng, Jinghao Zhang, Kuo Cai, Lejian Ren, Qiang Luo, Qianqian Wang, Qigen Hu, Rui Huang, Shiyao Wang, et~al.
\newblock Onerec technical report.
\newblock \emph{arXiv}, 2025.

\end{thebibliography}
